\documentclass{article}
\usepackage{kr}

\usepackage{times}
\usepackage{soul}
\usepackage{url}
\usepackage[hidelinks]{hyperref}
\usepackage[utf8]{inputenc}
\usepackage[small]{caption}
\usepackage{graphicx}
\usepackage{amsmath}
\usepackage{amsthm}
\usepackage{booktabs}
\usepackage{algorithm}
\usepackage{algorithmic}
\usepackage{dsfont}
\usepackage{boxedminipage}
\newtheorem{example}{Example}
\newtheorem{theorem}{Theorem}
\newtheorem{proposition}{Proposition}
\newtheorem{lemma}{Lemma}

\newtheorem{corollary}{Corollary}
\newtheorem{definition}{Definition}

\usepackage{thmtools}
\usepackage{thm-restate}
\usepackage{amsfonts}
\usepackage{xspace}
\usepackage{pifont}
\usepackage{color}
\usepackage{upgreek, stmaryrd}
\usepackage{amssymb}
\usepackage{misc}   
\usepackage{xargs}
\usepackage[pdftex,dvipsnames]{xcolor}
\usepackage{mathtools}
\usepackage{enumitem}
\usepackage{tikz}
\usetikzlibrary{arrows.meta, positioning}

\usepackage{microtype}

\setlist[enumerate,1]{leftmargin=*, labelsep=.6em}

\title{Fitting Horn DL Ontologies to ABox and Query Examples: \\[1mm]
A Tale of Simulation Quantifiers and Finite Models}

\author{%
 Marvin Grosser \and
 Carsten Lutz
\affiliations
Leipzig University\\
ScaDS.AI Center Dresden/Leipzig\\
\emails
\{grosser, clu\}@informatik.uni-leipzig.com
}

\begin{document}

\maketitle

\begin{abstract}
We study the problem of fitting a
description logic (DL) ontology to a given set of positive and negative examples in the form of an ABox and a Boolean query. While previous work has investigated this problem for the expressive DLs \ALC and
\ALCI, we here focus on the Horn DLs \EL and \ELI, as well as their extensions with the bottom concept. 
As the query language, we consider 
atomic queries (AQs), conjunctive queries (CQs), and unions thereof (UCQs). We provide characterizations of the existence of a fitting ontology  based on simulations, use them to develop decision procedures, and clarify the exact computational complexity. For AQs, the problem is in {\sc PTime} 
for both \EL and~\ELI. For CQs and UCQs, it is $\Sigma^\textsc{P}_2$-complete 
for \EL and  \ExpTime-complete for  \ELI. Adding the bottom concept does not change any of these complexities.
  Interestingly, moving from \ALC and \ALCI to \EL and \ELI introduces additional technical challenges  rather than simplifying the matter. 
\end{abstract}

\section{Introduction}

In knowledge representation
and databases, fitting problems aim to support 
human engineers in constructing ontologies and queries. Given as input a collection of positively and
negatively labeled examples, a fitting algorithm outputs  a formal object that fits all examples
or reports that no such object exists. The study of fitting problems originated in databases, where the formal object to be constructed is a query \cite{DBLP:conf/afips/Zloof75}. This line of work
is known as  query-by-example, and sometimes as
reverse query engineering, and it continues to be studied today, see for instance \cite{DBLP:journals/tods/CohenW16,DBLP:conf/icdt/Barcelo017,DBLP:journals/sigmod/CateFJL23} and references therein.
Although not  the focus of this article, fitting problems are also closely
related to machine learning, where fitting a model to a  set of examples
is  linked to PAC-style generalization guarantees \cite{DBLP:books/daglib/0033642}. 

Fitting problems have also attracted considerable interest in ontology-based KR, 
 particularly in the context of description logics (DLs).
There are at least two ways in which fitting problems can support ontology engineering. On the one hand,
they can help an engineer  construct
\emph{concepts} that  can then serve as building blocks in an ontology. On the other hand, fitting can be used to construct entire \emph{ontologies}, either for 
stand-alone use or as a module in a larger ontology under development.

Interest in concept fitting first emerged under the label of \emph{DL concept learning} although, unlike in machine learning, generalization guarantees were typically not considered; see \cite{DBLP:conf/ijcai/CateFJL23} for an exception. This research line led to the development of practical systems,  most prominently DL Learner \cite{DBLP:journals/ml/LehmannH10,DBLP:conf/www/BuhmannLWB18}. Starting with
\cite{DBLP:conf/ijcai/FunkJLPW19}, the
focus shifted to algorithms for concept fitting that are sound, complete and terminating, as well as to 
investigations of the decidability and computational complexity of concept fitting in the presence of an ontology;
see for instance \cite{DBLP:journals/ai/JungLPW22} and the references therein. 

Ontology fitting, by contrast, has  attracted attention only relatively recently. A key design choice
concerns the form of the examples.
One natural option, adopted in \cite{KR2025-40}, is to use interpretations (logical structures)
as examples. Another,
used in \cite{FGL-KR25}, is to use pairs $(\Amc,q)$ where \Amc is an ABox and $q$ is a (Boolean) query. For a positive example, a fitting ontology \Omc must then satisfy $\Amc \cup \Omc \models q$ while for a negative example it must satisfy \mbox{$\Amc \cup \Omc \not\models q$}. The latter setting
is particularly relevant in the context of ontology-mediated querying \cite{DBLP:conf/rweb/BienvenuO15,DBLP:conf/ijcai/XiaoCKLPRZ18}. 

In this article, we study the problem of fitting ontologies to ABox-query examples. Whereas previous work has focused on 
expressive DLs such as \ALC and \ALCI
as the ontology language, we instead
consider Horn DLs, especially  \EL, \ELI, and their extensions with the bottom concept `$\bot$'. These DLs are of special interest because they are computationally more tractable than their expressive counterparts. In particular, \EL admits subsumption and querying in polynomial time for tree-shaped queries, and
both \EL and \ELI are supported by so-called consequence-based reasoning algorithms that are highly efficient in practice \cite{DBLP:conf/birthday/CucalaGH19}.
This has contributed to an extension of \EL becoming one of three profiles of
the OWL 2 ontology language standard~\cite{profiles}. As query languages,
we consider atomic queries (AQs), conjunctive queries (CQs), and unions thereof (UCQs).
In addition, we consider a fitting problem
in which the examples  consist of an ABox only and where we seek an ontology that is consistent with the positive examples and inconsistent with the negative ones. For all these fitting
problems, we
provide effective characterizations and 
determine the precise complexity of deciding whether a fitting ontology exists. Our
algorithms also produce explicit
fitting ontologies when they exist.

Our characterizations resemble those
obtained in \cite{FGL-KR25} for the expressive DLs \ALC and \ALCI, but they also exhibit  notable differences. 
To begin with, the 
characterizations for $\mathcal{ALC}$ and \ALCI are based on homomorphisms which ours replace with simulations. This introduces technical challenges, especially for
CQs and UCQs as the query language,
intuitively because homomorphisms are more local in nature than simulations. For example, a reflexive loop must map homomorphically to a reflexive loop, but it may be simulated by an infinite path. When turning our characterizations into fitting algorithms, we are therefore led to employ the
extension of \EL and \ELI with simulation quantifiers as studied
for the \EL case in \cite{lutz2010enriching}. In the case
of \ELI, an additional complication
arises: the non-local nature of
simulations clashes, in a certain sense, with the more local 
nature of \ELI-concepts used in ontologies. We address this issue by working with finite interpretations rather than with unrestricted ones in our characterizations, and by considering finite model reasoning in \ELI extended with simulation quantifiers. 

Table~\ref{fig:summary} summarizes our complexity results and compares them with the known complexities of query entailment in \EL and \ELI. All results are completeness results, except the \PTime ones (for which we conjecture completeness under logspace-reductions as well). 
The lower bounds for (U)CQs already hold for rooted queries, i.e.,
(U)CQs in which every connected component contains an answer variable.
The complexities are generally lower than for \ALC and \ALCI.
For example, \ALC- and \ALCI-ontology fitting is \TwoExpTime-complete 
for CQs and UCQs, and
\NPclass-complete for AQs \cite{FGL-KR25}.
The complexity also differs from
that of fitting ontologies  to examples that take the form of interpretations.
That problem is \ExpTime-complete
for all four DLs listed in 
Table~\ref{fig:summary} \cite{KR2025-40}.

 We also discuss, throughout the paper, the size and nature of fitting ontologies. 
For AQs and for consistency-based fitting, the existence of a fitting ontology coincides with the existence of a fitting ontology that does not contain auxiliary symbols (i.e., concept or role names not present in the examples). For conjunctive queries and UCQs, this is not the case. In particular, a fresh role name is required to deal with disconnected queries. 
We also observe that, in most cases, admitting auxiliary symbols enables the
construction of fitting ontologies of polynomial size while, without such symbols, they are of single exponential size. The sole exception is the
case of $\ELI$ and (U)CQs, where we do not obtain any bounds on the size of fitting ontologies at all.

Proofs are provided in the appendix of \cite{grosserKR26}.

\section{Preliminaries}
\label{sect:prelims}

    \begin{table}
        \centering
        \begin{tabular}{lrr}\toprule
        & $\EL_{(\bot)}$&$\ELI_{(\bot)}$ \\\midrule
     AQ-entailment & \PTime   & \ExpTime \\
     (U)CQ-entailment & \NPclass  & \ExpTime \\\midrule
     ontology fitting, consistency &   \PTime & \PTime \\
     ontology fitting, AQs & \PTime  & \PTime \\
      ontology fitting, (U)CQ & $\Sigma^{\textsc{P}}_2$  & \ExpTime \\\bottomrule
    \end{tabular}
        \caption{Complexities of query entailment (top) 
        from \protect\cite{DBLP:conf/ijcai/BaaderBL05,RosatiDL07-ELCQA,BaaderEtAl-OWLED08DC,DBLP:conf/jelia/EiterGOS08}
        and of ontology fitting from this article (bottom).} 
        \label{fig:summary}
        \vspace*{-4mm}
    \end{table}

\subsection{Description Logic}

Let \NC, \NR, and \NI be countably infinite sets of \emph{concept names}, \emph{role names}, and \emph{individual names}.
An \emph{inverse role} takes the form $r^-$ with $r$ a role name, and a \emph{role} is a role name or an inverse role. If $r=s^-$ is an inverse role, then we set $r^-=s$.
An \emph{\ELIbot-concept} $C$ is built according to 
the syntax rule
$$
C,D ::= \top \mid \bot \mid A  \mid C \sqcap D \mid \exists r . D
$$
where $A$ ranges over concept names and $r$ over roles. 
An \emph{\ELbot-concept} is an \ELIbot-concept that does not use inverse
roles, an \emph{\EL-concept} is an $\EL_\bot$-concept that
does not use `$\bot$', and likewise for \emph{\ELI-concepts}.
The \emph{role depth} of a concept $C$,
denoted $\mn{rd}(C)$,
is the maximum nesting depth of the operator $\exists r . D$ in it.
We shall occasionally also mention \ALC-concepts, which extend \EL-concepts
with negation $\neg C$, and likewise
for \ALCI and \ELI.

Let $\Lmc \in \{ \EL, \ELbot, \ELI, \ELIbot, \ALC, \ALCI \}$. 
An \emph{\Lmc-ontology} is a finite set of \emph{concept inclusions
  (CIs)} $C \sqsubseteq D$, where $C,D$ are \Lmc-concepts. 
An \emph{ABox} is a finite and nonempty set of
\emph{concept assertions} $A(a)$ and \emph{role assertions} $r(a,b)$
where $A \in \NC$, $r \in \NR$, and $a,b \in \NI$. We use $\mn{ind}(\Amc)$ to denote the set of individual names used
in~\Amc. 

The semantics of concepts is defined as usual in terms of interpretations
$\Imc=(\Delta^\Imc,\cdot^\Imc)$ with $\Delta^\Imc$ the (non-empty)
\emph{domain} and $\cdot^\Imc$ the \emph{interpretation function}, see \cite{Baader2017} for  details. 
A \emph{pointed interpretation} is a pair $(\Imc,d)$ with \Imc an interpretation and $d \in \Delta^\Imc$. 
An interpretation \Imc
\emph{satisfies} a CI $C \sqsubseteq D$ if $C^\Imc \subseteq D^\Imc$,
a concept assertion $A(a)$ if $a \in A^\Imc$, and
a role assertion $r(a,b)$ if $(a,b) \in r^\Imc$; we thus make the \emph{standard names assumption}. We say that \Imc is a \emph{model} of an
ontology \Omc, written $\Imc \models \Omc$, if it satisfies all CIs
in it, and likewise for ABoxes. An ontology is \emph{satisfiable} if it
has a model and an ABox is \emph{consistent} with an ontology \Omc if \Amc and \Omc 
have a common model. 

\subsection{Simulations, Products, etc.}

A \emph{homomorphism} from an interpretation $\Imc_1$ to an interpretation $\Imc_2$
is a mapping $h \colon \Delta^{\Imc_1} \to \Delta^{\Imc_2}$
such that $d \in A^{\Imc_1}$
implies $h(d) \in A^{\Imc_2}$ and $(d,e) \in r^{\Imc_1}$ implies $(h(d),h(e)) \in r^{\Imc_2}$
for all  $A \in \NC$, $r \in \NR$,  and $d,e \in \Delta^{\Imc_1}$. We write $(\Imc_1,d_1) \rightarrow (\Imc_2,d_2)$ to denote the existence of a homomorphism $h$ from $\Imc_1$ to $\Imc_2$ with $h(d_1)=d_2$. 

For $k \in \mathbb{N} \cup \{ \omega \}$, a \emph{$k$-\EL-simulation} from $\Imc_1$ to $\Imc_2$ is a relation $S \subseteq \Delta^{\Imc_1} \times 
	\Delta^{\Imc_2} \times \{ i \in \mathbb{N} \mid i \leq k\}$ such that for all $A \in \NC$ and
    $r \in \NR$: 
				\begin{description}
				\item[(Atom)] if $d_1 \in A^{\Imc_1}$ and $(d_1,d_2,i) \in S$, then $d_2 \in A^{\Imc_2}$;
				\item[(Rel)] if  
					$ ( d_1,e_1 ) \in r^{\Imc_1}$ 
					and $ ( d_1,d_2,i ) \in S$ with $i \geq 1$, 
                    then there exists $(d_2,e_2) \in r^{\Imc_2}$ with $(e_1,e_2,i-1) \in S$.\footnote{\label{fn:handlingomega}here we take $\omega -1$ to mean $\omega$}
			\end{description}
An \emph{\ELI-simulation} is defined in the same way, but is required to satisfy the additional condition
			\begin{description}
				\item[(iRel)] if  
					$ ( e_1,d_1 ) \in r^{\Imc_1}$ 
					and $ ( d_1,d_2,i) \in S$ with $i \geq 1$, then there exists $(e_2,d_2) \in r^{\Imc_2}$ with $(e_1,e_2,i-1) \in S$.
			\end{description}
       $k$-\ELbot- and \emph{$k$-\ELIbot-simulations} $S$ are defined in
     the same way, but additionally need to be total:
     \begin{description}
				\item[(Tot)]
     for all $d_1 \in \Delta^{\Imc_1}$, there is a $d_2 \in \Delta^{\Imc_2}$
     with \mbox{$(d_1,d_2,k) \in S$}.
     \end{description}
      Let $\Lmc \in \{ \EL, \ELI, \EL_\bot, \ELI_\bot \}$.
		For $k \in \mathbb{N} \cup \{ \omega \}$, $d_1 \in \Delta^{\Imc_1}$, and $d_2 \in \Delta^{\Imc_2}$,  we write 
	$( \Imc_1, d_1) \preceq^k_\Lmc ( \Imc_2, d_2 )$ if there is a
	$k$-\Lmc-simulation $S$ from $\Imc_1$ to $\Imc_2$ with \mbox{$( d_1,d_2 ,k) \in S$}. 
    
    An \emph{\Lmc-simulation}
    from $\Imc_1$ to $\Imc_2$ is a
    relation $S \subseteq \Delta^{\Imc_1} \times 
	\Delta^{\Imc_2}$ such that 
    $\{ (d,e,\omega) \mid (d,e) \in S \}$ is an $\omega$-\Lmc-simulation 
    from $\Imc_1$ to $\Imc_2$. We denote the existence of an \Lmc-simulation from $\Imc_1$ to $\Imc_2$ by $\Imc_1 \preceq_\Lmc \Imc_2$. Note that for $\Lmc \in \{ \EL, \ELI \}$, the empty
    relation is an \Lmc-simulation from $\Imc_1$ to $\Imc_2$, for all interpretations $\Imc_1, \Imc_2$. This is not the case for $\Lmc \in \{ \EL_\bot, \ELI_\bot \}$, where \Lmc-simulations have to be total.
    We write $\left( \Imc_1, d_1 \right) \preceq_\Lmc \left( \Imc_2, d_2 \right)$ if there is an
	\Lmc-simulation $S$ from $\Imc_1$ to $\Imc_2$ that contains $(d_1, d_2)$. 
    We shall often exploit that, clearly, every homomorphism is an $\ELI_\bot$-simulation. 
    
    We state the basic and well-known connection between  DL concepts and $k$-simulations.
\begin{lemma}\label{lem:SimPresConc}
  Let $\Lmc \in \{ \EL,\ELI\}$, $k \geq 0$,   and  $C$ an \Lmc-concept with $\mn{rd}(C) \leq k$. Then  $d_1 \in C^{\Imc_1}$ and $(\Imc_1,d_1) \preceq^k_\Lmc (\Imc_2,d_2)$ implies $d_2 \in C^{\Imc_2}$.
\end{lemma}
The next lemma asserts the existence of characteristic concepts, a construction is in the appendix.
\begin{restatable}{lemma}{lemcharconc} \label{lem:charconc}
    Let $\Lmc \in  \{ \EL, \ELI\}$, let
	\Imc be an interpretation  that interprets only finitely many role and concept names as non-empty, $ d \in \Delta^\Imc $, and $k \geq 0$.
	Then there exists a \emph{characteristic \Lmc-concept $C_{\Imc,d}^{\Lmc,k}$ of role depth $k$ for $d$ in 
    \Imc} such that for all interpretations \Jmc and
	$ d' \in \Delta^\Jmc$:
	\[ (\Imc, d ) \preceq_\Lmc^k ( \Jmc, d' )   \quad\text{ if and only if } \quad 
		d' \in ( C_{\Imc, d}^{\Lmc,k} )^\Jmc 	
	.\]
\end{restatable}
The next lemma establishes a crucial link between bounded and unbounded simulations. For a proof see Lemma~15 in \cite{KR2025-40}.
\begin{lemma}
    \label{lem:finmodsim}
    Let $\Lmc \in  \{ \EL, \ELI\}$ and $\Imc,\Jmc$ be finite interpretations. Then $$ (\Imc, d) \preceq_\Lmc (\Jmc, e) \text{ if and only if }(\Imc, d) \preceq_\Lmc^{k} (\Jmc, e)$$ for all $d \in \Delta^\Imc$, $e \in \Delta^\Jmc$ and $k = |\Delta^\Imc| \cdot |\Delta^\Jmc|$.
\end{lemma}
While it is \NPclass-complete to decide
whether there is a homomorphism 
from one given finite interpretation to
another, the corresponding problem
for simulations can be solved in 
\PTime, see e.g.\ \cite{DBLP:conf/focs/HenzingerHK95}.
\begin{lemma}
\label{lem:simptime}
    Let $\Lmc \in  \{ \EL, \ELbot, \ELI,\ELIbot\}$. Given two finite
    pointed interpretations $(\Imc_1,d_1)$
    and $(\Imc_2,d_2)$, it can be decided
    in \PTime whether $(\Imc_1,d_1) \preceq_\Lmc (\Imc_2,d_2)$.
\end{lemma}
We will sometimes also consider simulations from an ABox to an ABox, or from an ABox to an interpretation. In
such cases, we view an ABox \Amc as an interpretation with $\Delta^\Amc = \mn{ind}(\Amc)$, $r^\Amc = \{ (a,b) \mid r(a,b) \in \Amc \}$, and $A^\Amc = \{ a \mid A(a) \in \Amc \}$ for all $r \in \NR$ and $A \in \NC$. 

\smallskip

The \emph{direct product} of two pointed interpretations
$(\Imc_1,d_1)$ and $(\Imc_2,d_2)$ is the  pointed interpretation
$(\Imc_1\times \Imc_2, (d_1,d_2))$ where 
$$
\begin{array}{r@{\;}c@{\;}l}
  \Delta^{\Imc_1\times \Imc_2} &=& \Delta^{\Imc_1} \times \Delta^{\Imc_2} \\[1mm]
  A^{\Imc_1\times \Imc_2} &=& A^{\Imc_1} \times A^{\Imc_2} \quad \text{ for all } A \in \NC \\[1mm]
  r^{\Imc_1\times \Imc_2} &=& \{ ((d_1,d_2),(e_1,e_2)) \mid (d_1,e_1) \in r^{\Imc_1} \\[1mm]
  && \qquad\text{and } (d_2,e_2) \in r^{\Imc_2}\} \quad \text{ for all } r \in \NR.
\end{array}
$$
The product construction extends naturally to non-empty finite sequences $S$ of pointed instances. We denote the resulting pointed interpretation
by $(\prod S,d_\Pi)$. By convention, $\prod \emptyset$ is the  pointed interpretation $(\Imc,d)$ with $d=()$,
$\Delta^\Imc = \{ d \} = A^{\Imc}$ for all $A \in \NC$
and $r^\Imc = \{(d,d) \}$ for all $r \in \NR$.
We also refer to this \Imc as \emph{the maximal interpretation}. The following
is well-known \cite{DBLP:books/daglib/0013017}. 
\begin{lemma}
\label{lem:prodlem}
    Let $S=(\Imc_i,d_i)_{i \in I}$ be a sequence of pointed interpretations. Then 
    \begin{enumerate}
        
    \item $(\prod S,d_\Pi) \rightarrow (\Imc,d)$ for all $(\Imc,d) \in S$; 
   
    \item  for all $\bar e=(e_i)_{i \in I} \in \Delta^{\prod S}$ and all \ELI-concepts~$C$: $\bar e \in C^{\prod S}$ iff
    $e_i \in C^{\Imc_i}$ for all $i \in I$; 
    \item If $(\Jmc,e)$ is a pointed interpretation with $(\Jmc,e) \rightarrow (\Imc, d)$ for all $(\Imc, d) \in S$, then $(\Jmc, e) \rightarrow (\prod S, d_\Pi)$.
    \end{enumerate}
\end{lemma}

Let \Imc be an interpretation and $d \in \Delta^\Imc$.
A \emph{path} in \Imc is a sequence
$p=d_1 r_1 \cdots d_{n - 1} r_{n - 1} d_n$ of domain elements $d_i$ from $\Delta^\Imc$ and roles $r_i$
such that $(d_i, d_{i + 1}) \in r_{i}^{\Imc}$ for  $1 \leq i < n$.
We say that the path \emph{starts} at $d_1$ and use 
$\mn{tail}(p)$ to denote~$d_n$. 
The \emph{$\ELI$-unraveling} of \Imc at $d$ is the
interpretation $\Imc^{\downarrow\ELI}_d$ defined as follows:
$$
\begin{array}{rcl}
    \Delta^{\Imc^{\downarrow\ELI}_d} &=& \text{set of all paths in \Imc starting at } d \\[1mm]
    A^{\Imc^{\downarrow\ELI}_d} &=& \{ p \mid \mn{tail}(p) \in A^\Imc \} \\[1mm]
    r^{\Imc^{\downarrow\ELI}_d} &=& \{ (p,p') \mid p' = pre \text{ for some } e \} \, \cup \\[1mm]
    && \{ (p',p) \mid p' = pr^-e \text{ for some } e \}.
\end{array}
$$
The \emph{$\EL$-unraveling} $\Imc^{\downarrow\EL}_d$ of $\Imc$ at $d$ is defined in the same way, except that paths must not contain inverse roles. An $\ELbot$-unraveling  is simply an $\EL$-unraveling, and 
likewise for $\ELIbot$-unravelings.
The following properties are well-known. 
\begin{lemma}
\label{lem:unravbasic}
  Let \Imc be an interpretation, $d \in \Delta^\Imc$, and $\Lmc \in \{ \EL, \ELI\}$. Then 
  \begin{enumerate}
    \item $(\Imc^{\downarrow\Lmc}_d,d) \rightarrow (\Imc,d)$ via the homomorphism $\pi \mapsto \mn{tail}(\pi)$;
    \item $(\Imc, d) \preceq_\Lmc (\Imc_d^{\downarrow \Lmc},d)$;
    \item for all \Lmc-concepts
    $C$ and $p \in \Delta^{\Imc^{\downarrow\Lmc}_d}$: $p \in C^{\Imc^{\downarrow\Lmc}_d}$
    iff $\mn{tail}(p) \in C^\Imc$. 
  \end{enumerate}
\end{lemma}

\subsection{Queries} 

A \emph{conjunctive query (CQ)} takes the form
$q = \exists \overline x \, \varphi(\overline x)$
where $\overline x$ is a tuple of variables and $\varphi$ a
conjunction of \emph{atoms} $A(t)$ and $r(t,t')$, with
\mbox{$A \in \NC$}, $r \in \NR$, and $t,t'$ variables from
$\overline x$ or individuals from \NI.
With $\mn{var}(q)$ and $\mn{ind}(q)$, we denote the set of variables and individual names in $\overline x$.
We take the liberty to view
$q$ as a set of atoms, writing e.g.\ $\alpha \in
q$ to indicate that $\alpha$ is an atom in~$q$. 
We may also write $r^-(x,y) \in q$ in place of \mbox{$r(y,x) \in q$}.
 An \emph{atomic query (AQ)}
is a  CQ of the simple form $A(a)$, with $A$ a concept name and $a$ an individual.
A {\em union of conjunctive queries (UCQ)} $q$ is a disjunction of CQs. A CQ is \emph{rooted}
if every variable is reachable from an individual name in the undirected graph \mbox{$G_q=(\mn{ind}(q) \cup \mn{var}(q),\{\{t,t'\} \mid r(t,t') \in q\})$}. A UCQ is \emph{rooted} if
every CQ in it is.

 A CQ $q$ gives rise to an interpretation
 $\Imc_q$ with $\Delta^{\Imc_q}=\mn{ind}(q) \cup \mn{var}(q)$, $A^{\Imc_q} = \{ t \mid
 A(t) \in q \}$, and $r^{\Imc_q}=\{(t,t') \mid r(t,t') \in
 q\}$ for all $A \in \NC$ and $r \in \NR$.
With a homomorphism from a CQ $q$ to an
interpretation \Imc, we mean a homomorphism from $\Imc_q$ to \Imc
 that is the identity
on all individual names. 
For an interpretation \Imc and a UCQ~$q$, we write
$\Imc \models q$ if there is a 
homomorphism $h$ from a CQ in $q$ to \Imc. 
For an ABox \Amc and ontology \Omc, we write $\Amc \cup \Omc \models q$ if  $\Imc \models q$ for all models
\Imc of $\Amc$ and~\Omc.

We shall also consider queries that
in essence are DL concepts. Let $\Lmc \in \{ \EL, \EL_\bot,\ELI, \ELI_\bot \}$. A \emph{rooted \Lmc-query} $q$ takes
the form $C(a)$ with $C$ an \Lmc-concept and $a \in \NI$. We write 
$\Amc \cup \Omc \models q$ if  $a \in C^\Imc$ for all models
\Imc of $\Amc$ and~\Omc. An \emph{existential \Lmc-query} $q$
takes the form $\exists x \, C(x)$
with $C$ an \Lmc-concept. We set
$\Amc \cup \Omc \models q$ if  $C^\Imc \neq \emptyset$ for all models
\Imc of $\Amc$ and~\Omc. When speaking only of an \Lmc-query, we mean an \Lmc-query that is either rooted or existential. There is also a finite
version of \Lmc-query entailment,
denoted by `$\models^{\mn{fin}}$'. While the unrestricted and the finite version  coincide for the DLs \Lmc listed above,
this is not the case for some of their
extensions introduced later on. 
It is easy to 
see that every $\EL$- and $\ELI$-query  can be expressed as a (tree-shaped) CQ.

Note that all queries introduced above are Boolean, that is, they evaluate to true or false instead of producing answers. This is without loss of generality since we
admit individual names in queries. 

We use $||O||$ to denote the \emph{size} of a syntactic object $O$ such as an ontology
or an ABox. It is defined as the length
of the encoding of $O$ as a word over a suitable alphabet.

\subsection{ABox Examples and the Fitting Problem}

Let \Qmc be a query language such as $\Qmc = \text{AQ}$ or $\Qmc = \text{CQ}$. An \emph{ABox-\Qmc example} is a pair $(\Amc,q)$
with \Amc an ABox and $q$ a query from \Qmc such that all individual names that appear in $q$ are from $\mn{ind}(\Amc)$.

By a \emph{collection of labeled examples} we  mean a pair 
$E=(E^+, E^-)$ of finite sets of examples. The examples in $E^+$
are the \emph{positive examples} and the examples in $E^-$ are the \emph{negative
examples}. We say that \Omc \emph{fits} $E$ if 
$\Amc \cup \Omc
\models q$ for all $(\Amc,q) \in E^+$ and $\Amc \cup \Omc
\not\models q$ for all $(\Amc,q) \in E^-$. The
following example illustrates this central
notion.
\begin{example}
\label{ex:firstone}
Let $E=(E^+,E^-)$ be  the collection of labeled ABox-CQ examples  where $E^+$
contains the two examples
$$
\begin{array}{l}
(\ \{\, \mathsf{NewHire}(\mathsf{jane}) \,\}, \ 
\exists x \, \mathsf{mentor}(\mathsf{jane},x) \wedge \mathsf{Emp}(x) \ ) \\[2mm]
( \ \{\, \mathsf{NewHire}(\mathsf{alex}),\ \mathsf{RemoteWorker}(\mathsf{alex}) \,\}, \\[1mm]
 \qquad\qquad \exists x\ \mathsf{mentor}(\mathsf{alex},x) \wedge \mathsf{SeniorEmp}(x)
\ ) .
\end{array}
$$
and $E^-$ contains the single
example
$$
  ( \{ \ \mathsf{NewHire}(\mathsf{bob}) \,\}, \ 
\exists x\ \mathsf{mentor}(\mathsf{bob},x) \wedge \mathsf{SeniorEmp}(x)\ ).
$$
Then, the \EL-ontology \Omc that contains the CIs
$$
\begin{array}{r@{\;}c@{\;}l}
  \mathsf{NewHire} &\sqsubseteq& \exists\mathsf{mentor}.\mathsf{Emp} \\[1mm]
 \mathsf{NewHire} \sqcap \mathsf{RemoteWorker} &\sqsubseteq& \exists\mathsf{mentor}.\mathsf{SeniorEmp} 
\end{array}
$$
fits all examples in $E$.
If we further add the 
CI $
\mathsf{SeniorEmp} \sqsubseteq \mathsf{Emp}
$, the ontology still fits.
\end{example}

Let \Lmc be an ontology language, such as $\Lmc = \ELIbot$, and  \Qmc a query
language. Then \emph{(\Lmc,\Qmc)-ontology fitting} is the problem
to decide, given as 
input a collection of 
labeled \mbox{ABox-\Qmc} examples $E$, whether $E$ admits a fitting \Lmc-ontology.
We generally assume that the ABoxes used in $E$ 
 have pairwise disjoint sets of individual names.
It is not hard to verify that this is without loss of generality because consistently renaming individual names in a collection of examples has no impact on the existence of a fitting ontology.
There is a natural variation of 
$(\Lmc,\Qmc)$-ontology fitting 
where one additionally requires the fitting ontology to be consistent with all ABoxes that occur in  examples. We then speak of 
\emph{consistent $(\Lmc,\Qmc)$-ontology fitting}. The following observation from \cite{FGL-KR25} shows that it suffices to 
design algorithms for  $(\Lmc,\Qmc)$-ontology fitting  as originally introduced.
\begin{proposition}
\label{prop:fittingswithconsistentABoxes}
  Let \Lmc be any ontology language and $\Qmc \in \{ \text{AQ}, \text{CQ}, \text{UCQ} \}$. Then there is a \PTime-reduction from consistent $(\Lmc,\Qmc)$-ontology fitting 
  to  $(\Lmc,\Qmc)$-ontology fitting.
\end{proposition}

\section{Consistency-Based Fitting}
\label{sect:consistency}

We start with a version of ontology fitting that is based on ABox consistency rather than on querying.
Besides being natural in its own right,
it also serves to showcase some central techniques used in this paper.
Here, an example is simply an ABox, and
 an ontology \Omc \emph{fits} a collection
  $E=(E^+,E^-)$ if \Amc is consistent with \Omc
for all $\Amc \in E^+$ and inconsistent with \Omc for all
$\Amc \in E^-$. We refer to the  resulting 
decision problem as \emph{consistent \Lmc-ontology fitting}. 

We only consider $\EL_\bot$ and $\ELI_\bot$ because
the cases of \EL and \ELI are uninteresting. In fact, it is easy to see that a collection of labeled ABox examples $E$  has a fitting \EL-ontology
if and only if $E$ does not contain negative examples, and the same is true for \ELI. The subsequent theorem provides a
characterization of the existence of fitting $\EL_\bot$- and $\ELI_\bot$-ontologies.
\begin{theorem} \label{thm:consfitting}
	Let $ E = \left( E^+, E^- \right)  $ be a collection of labeled ABox examples with $E^+ \neq \emptyset$ and
	$ \Lmc \in  \{ \ELbot, \ELIbot\}$. Then the following
	are equivalent, where $\Amc^+ = \biguplus E^+ $:
	\begin{enumerate}
		\item $ E $ admits a fitting \Lmc-ontology;
		\item $ \Amc \not \preceq_\Lmc \Amc^+ $ for every  $ \Amc \in E^- $. 
	\end{enumerate}
\end{theorem}
The next example shows that the existence of a fitting $\EL_\bot$-ontology does  not, in general, coincide with the
existence of a fitting $\ELI_\bot$-ontology. It is interesting to contrast this with the case of consistency-based fitting for \ALC and \ALCI, where the existence of a fitting ontology coincides \cite{FGL-KR25}.
\begin{example}
\label{ex:ELneqELI}
  Consider the collection of ABox examples $E=(E^+,E^-)$ where 
  $E^+$ contains the single example
  $$
    \Amc_1 = \{ r(a_1,b_1), r(a_2,b_2), A_1(a_1), A_2(a_2) \}
  $$
  and $E^-$ contains the single example
  $$
    \Amc_2 = \{ r(a_1,b), r(a_2,b), A_1(a_1), A_2(a_2) \}. 
  $$
  Then there is a fitting $\ELI_\bot$-ontology, e.g.\ the one that 
   contains the single CI
  $$
    \exists r^- . A_1 \sqcap \exists r^- . A_2 \sqsubseteq \bot.
  $$
  In contrast, it follows from Theorem~\ref{thm:consfitting} that there is no fitting \ELbot-ontology. Just observe that $\{(a_i,a_i),(b,b_i) \mid i \in \{1,2\}\}$  is an  $\EL_\bot$-simulation from $\Amc_2$
  to $\Amc_1$ (but not an $\ELI_\bot$-simulation).
\end{example}
We  work towards a proof of Theorem~\ref{thm:consfitting}.  
We give the following definition
also for \EL and \ELI as it will
be reused in subsequent sections.
\begin{definition}
\label{def:tripleint}
	Let \Amc be an ABox, \Imc an interpretation, 
	and $S$ an \Lmc-simulation from \Amc to \Imc,
    for $\Lmc \in \{ \EL, \ELbot, \ELI,\allowbreak \ELIbot\}$.
	Define interpretation  $ \Imc_0 $ as
	 \begin{align*}
		 \Delta^{\Imc_0} &:= \mn{ind}(\Amc) \\
		 A^{\Imc_0} &:= \{a  \mid (a,d) \in S \text{ implies } d \in  A^\Imc\}
		 &\text{for all } A \in \NC \\
		 r^{\Imc_0} &:= \{( a,b )  \mid  r( a,b) \in \Amc\} 
		 &\text{for all } r \in \NR.
	\end{align*}
    For every $a \in \mn{ind}( \Amc)  $, let $\Jmc_a$
    be  obtained from 
    \[
	   \prod_{(a,d) \in S} ( \Imc, d ) 	\] 
    by renaming the distinguished element $d_\Pi$ to $a$. Assume that this is the only element that $\Jmc_a$ shares with $\Imc_0$ 
     and that $\Delta^{\Jmc_a} \cap \Delta^{\Jmc_b} = \emptyset$ whenever $a \neq b$. 
    The interpretation $\Imc_{\Amc, S, \Lmc}$ is obtained by taking the (non-disjoint) union of $\Imc_0$ and the unraveled
    interpretations $\Jmc^{\downarrow\Lmc}_a$, $a \in \mn{ind}(\Amc)$.
\end{definition}
\begin{restatable}{lemma}{lemIASLsubI}	\label{lem:IASLsubI}
	Let \Amc be an ABox, \Imc an interpretation, \Omc an \Lmc-ontology and $S$ an \Lmc-simulation from \Amc to \Imc, with $\Lmc \in \{ \EL, \ELbot, \ELI,\ELIbot\}$. 
    Then $\Imc \models \Omc$ implies
	 $ \Imc_{\Amc, S, \Lmc} \models \Omc$. 
\end{restatable}
We are now ready to prove Theorem~\ref{thm:consfitting}.

\medskip

\noindent
\begin{proof} {\bf (of Theorem~\ref{thm:consfitting}.)}
	`1$~\Rightarrow~$2'.
	Assume that there is an \Lmc-ontology \Omc that fits $E$ and that, to the contrary of what we have to show,
     there is an $ \Amc_0 \in E^- $ such that $ \Amc_0
	\preceq_\Lmc \Amc^+ $ via some simulation $S$. Since \Omc fits $E$, every
     $ \Amc \in E^+$ is consistent with \Omc. 
     Since \Lmc-ontologies are invariant under disjoint
     union, this implies that also 
     $ \Amc^+$  is consistent with~\Omc. Let  $ \Imc $ be a model of  $ \Amc^+$ and $\Omc $. 
     Then clearly $S$ is a simulation from $\Amc_0$ to \Imc. Consider the interpretation
	$ \Imc_{\Amc_0, S, \Lmc} $
    which by construction is a model of $\Amc_0$. By Lemma~\ref{lem:IASLsubI}, it is also a model of \Omc. This means that $ \Amc_0$  is consistent with \Omc, in contradiction to  $ \Omc $ fitting $ E $.
	
	\medskip
	`2$~\Rightarrow~$1'.
	Assume that $\Amc \not \preceq_\Lmc \Amc^+ $ for every  $ \Amc \in E^- $. Note that since $\Lmc \in \{\EL_\bot, \ELI_\bot\}$, we are concerned with total simulations. This
	means that for every $ \Amc \in E^- $ there is an $ a_\Amc \in
	\mn{ind}(\Amc) $ such that 
		$( \Amc, a_\Amc ) \not \preceq_{\Lmc} ( \Amc^+, b )
	$
	for all $ b \in  \mn{ind}( \Amc^+ )  $.  Lemma~\ref{lem:finmodsim}, 
	gives  \[
		( \Amc, a_\Amc ) \not \preceq^\ell_{\Lmc} ( \Amc^+, b )
		\tag{$*$}
	,\]
	where $  \ell =  | \Amc^+| \cdot  
	\mn{max}( \{ | \Amc|  \mid \Amc \in E^- \}  )$.
	Now consider the \Lmc-ontology
	\begin{align*}
		\Omc := \bigcup_{\Amc \in E^-} 
		\{ C_{\Amc, a_\Amc}^{\Lmc,\ell} \sqsubseteq \bot \}, 
	\end{align*}
    where $C_{\Amc, a_\Amc}^{\Lmc,\ell}$ is the
    characteristic \Lmc-concept of role depth $\ell$ for $a_\Amc$
    in \Amc viewed as an interpretation.
	We show that \Omc fits $E$. 
	Clearly, every $ \Amc \in E^-$
	 is inconsistent w.r.t.\ \Omc, and thus \Omc fits all negative examples. 
	It remains to show that every $\Amc \in E^+$  is consistent w.r.t.\ \Omc. 
    To this end, it clearly suffices to
    show that $\Amc^+$ is consistent
    w.r.t.\ \Omc. Consider $\Amc^+$
    viewed as an interpretation, which trivially
    is a model of $\Amc^+$. We show that
    ${\Amc^+} $ is also a model of
    \Omc, that is, $(C_{\Amc, a_\Amc}^{\Lmc,\ell} )^{\Amc^+} = \emptyset  $ for all $ \Amc \in E^- $. In fact, this already follows
    from ($*$), because by choice of $C_{\Amc, a_\Amc}^\ell$ the existence
    of some $b \in (C_{\Amc, a_\Amc}^{\Lmc,\ell} )^{\Amc^+}$ implies 
		$( \Amc, a_\Amc ) \preceq^\ell_{\Lmc} ( \Amc^+, b ).$
\end{proof}
In the following, we highlight some properties
of the fitting ontologies constructed in the 
proof
of Theorem~\ref{thm:consfitting} and discuss alternative ways of constructing them. For simplicity, we restrict our attention to $\EL_\bot$. However, all observations extend straightforwardly also to $\ELI_\bot$.

The fitting $\EL_\bot$-ontology $\Omc$ constructed for a collection of labeled ABox examples $E$ in the proof
of Theorem~\ref{thm:consfitting} is derived only from the negative examples $E^-$. 
It contains only concept and role names from $E^-$, thus no auxiliary concept or role names, that is, symbols that do not occur in~$E$.  \Omc
 is of  size single exponential in $|E|$,
but can be reduced to polynomial size at the
expense of introducing an auxiliary
concept name $ X_{\Amc, a}^k $ for
every  $ \Amc \in E^- $, $ a \in \mn{ind}(\Amc) $ and $ k \in \left\{0, \ldots,
\ell\right\}  $. The resulting ontology 
$\Omc'$ contains the following
CIs for all $C_{\Amc, a}^{\ell\Lmc} \sqsubseteq \bot \in \Omc$ and 
         $1 \leq k \leq \ell$:
$$    
\begin{array}{rcl}
\displaystyle
	 \bigsqcap_{A(a) \in \Amc} \hspace{-0.1cm} A 
		 \sqcap \hspace{-0.1cm}\bigsqcap_{r\left( a,b \right) \in \Amc } \hspace{-0.1cm} \exists r.\, X_{\Amc, b}^{k-1} 
	 &\sqsubseteq& X_{\Amc, a}^k \\[5mm]
   \displaystyle \bigsqcap_{A(a) \in \Amc}A \ \sqsubseteq \ X_{\Amc, a}^0\qquad
 X_{\Amc, a_\Amc}^\ell &\sqsubseteq& \bot.
\end{array}
$$
Note that the sole purpose of the auxiliary concept names is to represent the
characteristic concepts $ 
C_{\Amc, a_\Amc}^{\Lmc,\ell}$ succinctly.
In fact, $\Omc'$ is a conservative extension of \Omc in the sense that every model of $\Omc'$ is a model of \Omc and, conversely, every model of $\Omc$ can be extended to a model of $\Omc'$ by interpreting the auxiliary concept names. Consequently, $\Omc'$ fits~$E$.

Interestingly, $\EL_\bot$-ontologies 
that fit $E$
can also be constructed in a completely different way, based on the positive examples. 
This alternative construction, which follows  ideas from \cite{FGL-KR25}, requires the use
of auxiliary concept names, and yields an ontology of polynomial size. To define a fitting ontology $\Omc_{\mn{alt}}$ of this kind, we introduce an auxiliary concept name $ \overline{V_a} $ for each individual $ a \in  \mn{ind}\left(\Amc^+
\right)  $. Informally, $ \overline{V_a} $ is satisfied at an element $d$ in an interpretation \Imc
if $(\Imc,d) \not\preceq_{\ELbot} (\Amc^+,a)$.
Let $ \Sigma
$ be the set of concept and role names used in  $ E^- $.
The ontology 
$\Omc_{\mn{alt}}$ then contains the following
CIs for all $a \in \mn{ind}(\Amc^+)$ and $A,r \in \Sigma$:
$$
\begin{array}{c}
 A \sqsubseteq \overline{V_a}  
  \quad \text{ if } A(a) \not\in \Amc^+, \\[1mm]
 \displaystyle
 \exists r. \bigsqcap_{r(a,b) \in \Amc^+} \overline{V_b} \sqsubseteq
			\overline{V_a},
            \qquad 
        \bigsqcap_{b \in \mn{ind}\left(\Amc^+\right)}
		 \overline{V_b}
	 \sqsubseteq \bot.
\end{array}
$$
\begin{restatable}{lemma}{lemsecondontconstrcorr} \label{lem:consaltont}
  If any $\EL_\bot$-ontology fits $E$, then $\Omc_{\mn{alt}}$ fits $E$.
\end{restatable}
It is interesting to note that none of the constructed ontologies uses existential quantifiers on the right-hand side of concept 
 inclusions. 

By Lemma~\ref{lem:simptime},
Theorem~\ref{thm:consfitting}  yields the following result.
\begin{theorem}
	Let $ \Lmc \in  \{ \ELbot, \ELIbot\}$. The consistent \Lmc-ontology
	fitting problem is in \PTime. 
\end{theorem}

\section{Atomic Queries}
\label{sect:AQs}

We consider atomic queries and present a characterization that resembles the one for  consistency-based fitting, but has an additional ingredient. Informally, each positive example $(\Amc,q)$  now acts like a certain kind of implication
on the negative examples, 
 similar to an existential
rule with $q$ as the (atomic unary) rule head
and \Amc as the rule body that, however, must be matched in terms of a simulation
rather than  a homomorphism. We now make
this precise.

Let $ E=(E^+, E^-  )  $ be a  collection of labeled ABox-AQ-examples. Note that the case where $E^- = \emptyset$ is trivial, since then the ontology $\{\top \sqsubseteq A \mid (\Amc, A(a)) \in E^+ \}$ always fits. In the following, we therefore assume that there is at least one negative example. Let $\Amc^- =  \biguplus_{( \Amc, Q(a)) \in E^- } \Amc$.
A \emph{completion} of $\Amc^-$ is an ABox \Cmc that satisfies
$$ \Amc^- \subseteq \Cmc \subseteq \Amc^- \cup \{Q(b)  \mid
( \Amc, Q(a) ) \in E^+, b \in \mn{ind}( \Amc^- )  \}.$$ 
In contrast to consistency-based fitting,
 also \EL and \ELI without the $\bot$ concept are natural choices
for the ontology language.
\begin{restatable}{theorem}{thmcharaq} \label{thm:charaq}
	Let $ E=(E^+, E^-  )  $ be a  collection of labeled ABox-AQ-examples  with $E^- \neq \emptyset$,
    and $ \Lmc \in\{ \EL, \ELbot, \ELI,\allowbreak \ELIbot\} $.
	Then
	the following are equivalent:
	\begin{enumerate}
		\item $ E $ admits a fitting \Lmc-ontology 
		\item there exists a completion $ \Cmc $ of $ \Amc^- $ such that
			\begin{enumerate}[label=(\alph*)]
				\item  $ Q(a) \not \in \Cmc $ for each $ ( \Amc, Q(a)
					) \in E^- $;
				\item if $ ( \Amc, Q(a) ) \in E^+  $ and $(\Amc,a) \preceq_\Lmc (\Cmc,b)$, then  \mbox{$ Q(b) \in \Cmc$}.
			\end{enumerate}
	\end{enumerate}
\end{restatable}
Point~(b) of Theorem~\ref{thm:charaq}
formalizes what we mean when saying `positive examples act as an implication'.
It is not hard to find a variation of Example~\ref{ex:ELneqELI} which shows that, also in the AQ case, the existence of a fitting $\EL_\bot$-ontology does not coincide with the 
existence of a fitting $\ELI_\bot$-ontology, and likewise for \EL and \ELI.
As illustrated by the next example, the presence of the bottom concept makes a difference, too.
\begin{example}
     Consider the collection of ABox examples $E=(E^+,E^-)$ where 
  $E^+$ contains the single example
  $
     (\{ A_1(a), B(b) \}),A_2(a))
  $
  and $E^-$ contains the single example
  $
   (\{ A_1(a) \}, A_2(a) ). 
  $
  Then $\{ B \sqsubseteq \bot \}$ is a 
  fitting $\EL_\bot$-ontology, but there is no fitting \ELI-ontology.
  In fact,  for $\Lmc \in \{\EL, \ELI\}$, every completion \Cmc of $\Amc^-$ that satisfies Point~(b) of
  Theorem~\ref{thm:charaq} must satisfy $A_2(a) \in \Cmc$ and thus violate Point~(a).
\end{example}
The discussion about different ways to construct fitting ontologies as well as
their size and use of auxiliary symbols
given in Section~\ref{sect:consistency}
also applies here. One subtlety is that the positive and negative examples swap r\^oles. In particular, the
fitting  ontology constructed in the proof of Theorem~\ref{thm:charaq},
which parallels the proof
of Theorem~\ref{thm:consfitting}, is derived from the
positive examples and does neither use auxiliary symbols nor existential quantifiers on the right-hand side of CIs. It is of single exponential size and can be made polynomial at the expense of introducing auxiliary symbols.
Again, there is an alternative construction that, in the case of AQs, starts from a completion \Cmc of $\Amc^-$ that satisfies Conditions~(a) and~(b) from Theorem~\ref{thm:charaq}. For simplicity, we restrict our attention to $\EL_\bot$.  Let $\Sigma
$ be the set of concept and role names used in  $ E^+ $. The ontology $\Omc_{\mn{alt}}$ then contains the following CIs for all 
$a \in \mn{ind}(\Amc^+)$ and $A,r \in \Sigma$:
$$
\begin{array}{c}
 A \sqsubseteq \overline{V_a}  
  \quad \text{ if } A(a) \not\in \Cmc, \\[1mm]
 \displaystyle
 \exists r. \bigsqcap_{r(a,b) \in \Cmc} \overline{V_b} \sqsubseteq
			\overline{V_a},
            \qquad 
        \bigsqcap_{b \in \mn{ind}\left(\Cmc\right)}
		 \overline{V_b}
	 \sqsubseteq \bot\\[5mm]
      \displaystyle
     \bigsqcap_{a \in \mn{ind}(\Cmc), A(a) \not \in \Cmc}\overline{V_a}  \sqsubseteq A \quad \text{ if } (\Amc, A(a)) \in E^+.
\end{array}
$$

Using a naive implementation of Theorem~\ref{thm:charaq} to decide $(\Lmc,\text{AQ})$-ontology fitting involves guessing a completion \Cmc
of $\Amc^-$ and thus results in an \NPclass upper bound.
However, instead of guessing \Cmc it is possible
to identify a unique candidate for \Cmc in 
polynomial time, thus improving the upper
bound to \PTime. This   parallels the development for
\ALC and \ALCI in \cite{FGL-KR25}.
\begin{definition}
\label{def:refcand} 
Let $E = (E^+, E^-)$ be a collection of labeled ABox-AQ examples with $E^- \neq \emptyset$. We use $\Amc^\pm$ to denote the ABox that can be obtained from $\Amc^- = \biguplus_{(\Amc,Q(a)) \in E^-} \Amc$ by exhaustively applying the following rule:
    \begin{description}
      \item[(\Rsf)]  
        if $(\Amc,Q(a)) \in E^+$ and
        $(\Amc,a) \preceq_\Lmc (\Amc^\pm,b)$, set $\Amc^\pm = \Amc^\pm \cup \{ Q(b) \mid (a,b)\in S \}$.
    \end{description}
 \end{definition}
 Based on Lemma~\ref{lem:simptime}, it is easy
 to see that $\Amc^\pm$ can be constructed
 in polynomial time. Moreover, it is not hard to show the following.
\begin{restatable}{proposition}{proprefcand}
\label{prop:refcand}
    Let $\Lmc \in \{ \EL, \ELbot, \ELI, \ELIbot\}$ and $E = (E^+, E^-)$ be a collection of labeled ABox-AQ examples  with $E^- \neq \emptyset$.
    Then the following 
    are equivalent:
    \begin{enumerate}
        
        \item     $E$ admits no fitting \Lmc-ontology;

        \item $Q(a) \in \Amc^\pm$ for some $(\Amc,Q(a)) \in E^-$.
       
    \end{enumerate}
\end{restatable}
Proposition~\ref{prop:refcand} clearly yields
the following. 
\begin{theorem}
\label{thm:AQincoNP}
    Let $\Lmc \in \{ \EL, \ELbot, \ELI, \ELIbot \}$. Then  $(\Lmc,\text{AQ})$-ontology fitting  is in \PTime.
\end{theorem}

\section{UCQs and CQs}
\label{sect:UCQCQ}

We start with  characterizations of 
$(\Lmc,\text{UCQ})$-ontology fitting,
for $\Lmc \in \{ \EL, \ELbot, \ELI, \ELIbot \}$,
which also cover the  case
of CQs. These characterizations  are more intricate
than in the AQ case. In particular, 
the finite completions used for the AQ characterizations  have to be replaced with potentially infinite ones, which we choose to represent
as interpretations. The transition
to infinite completions
 is connected to the fact that the fitting
ontologies constructed in the AQ-case do not use existential quantifiers on the right-hand side of CIs while
this  cannot be
achieved
for CQs and UCQs. Example~\ref{ex:firstone}
may serve as an illustration.

Since already for AQs the existence of a fitting ontology does not coincide
for any two of our four languages 
$\Lmc \in \{\EL,\ELbot,\ELI,\ELIbot\}$, the same is clearly true for CQs and
UCQs. In the following characterization,
we assume that there is at least one negative example. Note that the case without negative examples is again not too interesting: for $\Lmc \in \{ \EL_\bot, \ELI_\bot \}$, $\{ \top \sqsubseteq \bot \}$ is a fitting ontology; for $\Lmc \in \{ \EL, \ELI \}$, either the
ontology $\{ \top \sqsubseteq A \mid A \in \NC \text{ occurs in } E\} \cup 
\{ \top \sqsubseteq \exists r . \top \mid r \in \NR \text{ occurs in } E\}$
fits or there is no fitting ontology.
\begin{restatable}{theorem}{charUCQwithFin}
    \label{thm:charUCQwithFin}
	Let $ E= \left( E^+, E^- \right)  $ be a collection of labeled ABox-UCQ examples with 
	$ E^- \neq \emptyset  $ and $\Lmc \in \{\EL,\ELbot,\allowbreak \ELI,\ELIbot\}$. Then the following are equivalent:
	\begin{enumerate}
		\item there is an \Lmc-ontology that fits  $ E $;
		\item there exists a finite interpretation  $ \Imc $ such that
			 \begin{enumerate}
				 \item $ \Imc =\biguplus_{e \in E^-}  \Imc_e$ where, for each 
					 $ e=\left( \Amc, q \right) \in E^- $, $ \Imc_e $ is a 
					  model of \Amc with  $ \Imc_e \not \models q $;
				\item  for all $ \left( \Amc, q \right) \in E^+ $: if $ S$ is an \Lmc-simulation from \Amc to \Imc, then $ \Imc_{\Amc, S, \Lmc}  \models q$.
			\end{enumerate}
	\end{enumerate}
\end{restatable}
The subsequent example
illustrates 
Theorem~\ref{thm:charUCQwithFin}.
\begin{example}
\label{ex:elprepforeli}
    Consider the collection  $E=(E^+, E^-)$ where 
    $$
    \begin{array}{rcl}
       E^+ &=& \{ \ (\{A(a)\},\exists x\ r(a,x)\land A(x) ),\\[1mm]
    && \phantom{\{ \ }(\{r(a,a)\}, B(a))\ \} \\[2mm]
     E^- &=& \{ \ (\{A(a)\}, B(a)) \ \}.
     \end{array}
     $$
    $E$ does not admit a fitting \EL-ontology.
    Indeed, any  fitting \EL-ontology \Omc must satisfy
        $
      \Omc \models A \sqsubseteq \exists r. A$ because of the first positive example.  Thus in every model \Imc of \Omc and for every
    $d \in A^\Imc$, we
    have $(\Amc_2,a) \preceq_\EL (\Imc,d)$ with
    $\Amc_2 = \{ r(a,a) \}$  the ABox from the second positive example. It can be seen that this implies $d \in B^\Imc$; c.f.\ Point~(b) of Theorem~\ref{thm:charUCQwithFin}. It follows that $\Omc \models A \sqsubseteq B$
    and consequently \Omc violates the negative example.

    In contrast, the \ELI-ontology 
    $\{
      A \sqsubseteq \exists r . A, \
      \exists r^- . \top \sqsubseteq B \}
    $
    fits $E$. Note in particular that, for a model \Imc of \Omc and 
    $d \in A^\Imc$, $(\Amc_2,a) \preceq_\ELI (\Imc,d)$ 
    needs \emph{not} hold because there may  not be an infinite outgoing $r^-$-path from $d$ in \Imc. 
\end{example}
To further discuss Theorem~\ref{thm:charUCQwithFin}, we compare it to the characterizations for $(\ALC,\text{UCQ})$ and $(\ALCI,\text{UCQ})$ from \cite{FGL-KR25}. They
differ in three aspects.
First, the characterizations
for \ALC and \ALCI
use homomorphisms in place of simulations. Second, our definition of $\Imc_{\Amc,S,\Lmc}$ involves
direct products while
the one in  \cite{FGL-KR25}
does not.  
And third, the characterizations for \ALC and \ALCI do not require the interpretation \Imc in Point~2 to be finite. 
It is worthwhile
to discuss these differences in some more detail. 

The first difference arises
because of the limited expressive
power of \EL and \ELI, which is tightly linked to simulations. It is interesting to note that the characterizations for \ALC and \ALCI are in terms of homomorphisms, \emph{not} in terms of bisimulations, which are the counterpart of simulations for these DLs. This is related to 
`second-order effects' which sometimes make it possible to express CQs (or in this case: ABoxes from positive examples) in \ALC and \ALCI, see~\cite{FLW-IJCAI18}.
The second
difference is due to the fact that in contrast to homomorphisms, simulations are not functional, and thus we may have $(a,d_1),(a,d_2)\in S$ for distinct $d_1,d_2$, which we  reconcile via products. The
root of the third difference
is the non-local nature of 
simulations. For instance, a reflexive loop must map homomorphically to a reflexive loop, but it may be simulated by an infinite path.
The example below illustrates why this prevents us from using potentially infinite interpretations \Imc in Theorem~\ref{thm:charUCQwithFin}. We remark
that the third difference
is crucial  only for the \ELI case. For \EL,
in contrast, it follows from the technical development in Section~\ref{sect:EL} that the qualifier `finite' can be dropped from Point~2 of 
Theorem~\ref{thm:charUCQwithFin}.\footnote{See in particular the remark after Theorem~\ref{thm:ELsimPTime}.}
\begin{example}
\label{ex:simscausetrouble}
  Consider the collection $E'$ obtained from the collection $E$ in Example~\ref{ex:elprepforeli} by adding a third positive 
  example
  $$
  (\{ B(c), r(b,c)\}, B(b) ),
  $$
  which enforces that any fitting \ELI-ontology \Omc must satisfy
        $
      \Omc \models \exists r . B \sqsubseteq B$.
  Then $E'$ no longer admits a fitting \ELI-ontology. Informally, the
  reason is that \ELI-ontologies cannot capture the infinite nature of simulations. More precisely, it can be shown that for \Omc to satisfy the second positive example of $E$, there must be a \emph{finite subset} \Amc of the
      \ELI-unraveling of the ABox $\{r(a,a)\}$ from that example such that
    $\Amc \cup \Omc \models B(a)$. 
    Now let  \Imc be a model of a 
    fitting \ELI-ontology \Omc and let $d \in A^\Imc$. Since \Amc is finite and acyclic, there must
    be an element $e$ on the infinite outgoing $r$-chain from $d$ in \Imc
    such that $(\Amc,a) \preceq_\ELI (\Imc,e)$. Thus we must have $e \in B^\Imc$. Since the additional positive example  enforces
     $\Omc \models \exists r . B \sqsubseteq B$, this yields $d \in B^\Imc$, thus $\Omc \models A \sqsubseteq B$ and \Omc violates the negative example.

    Let us connect this to 
Theorem~\ref{thm:charUCQwithFin}.
    The infinite interpretation \Imc that is an infinite $r$-path
    with root $a$ and $A^\Imc = \Delta^\Imc$
    satisfies Points~(a) and~(b) of that theorem despite the fact that there is no  \ELI-ontology that fits $E'$. However, there is no finite interpretation that exhibits the same behaviour. Intuitively, the mismatch between the infinite nature of simulations and the finite nature of \ELI-ontologies only manifests on infinite interpretations while on finite interpretations it is mitigated by Lemma~\ref{lem:finmodsim}.
\end{example}
Up to some minor differences, the discussion about  different ways to construct fitting ontologies 
given in Section~\ref{sect:consistency}
also applies to Theorem~\ref{thm:charUCQwithFin}. Notably, any collection of examples that only contains rooted queries admits a fitting ontology without auxiliary symbols whenever it admits one at all, whereas non-rooted queries may  force a fitting ontology to use an auxiliary role. In the case of  $\EL_{(\bot)}$, using additional auxiliary symbols again allows us to construct fitting ontologies of 
polynomial size. For $\ELI_{(\bot)}$, the size of fitting ontologies remains open.

\subsection{Decomposing Conjunctive Queries}

To obtain decision procedures based on the characterizations given in Theorem~\ref{thm:charUCQwithFin}, we first introduce (a variation of) some well-known  technical
machinery that revolves around forest-shaped models and the decomposition
of a CQ into tree-shaped parts, going
back to the \emph{fork rewritings} of \cite{DBLP:conf/dlog/Lutz08}.

A model \Imc of an ABox \Amc is an \emph{\EL-forest model} if
\begin{enumerate}
\item the directed graph
$$
  G_\Imc = (\Delta^\Imc, \bigcup_{r \in \NR} r^\Imc \setminus (\mn{ind}(\Amc) \times \mn{ind}(\Amc)))
$$
is a forest (a disjoint union of trees) and 

\item $r^\Imc \cap (\mn{ind}(\Amc) \times \mn{ind}(\Amc)) = \{ (a,b) \mid r(a,b) \in \Amc \}$. 

\end{enumerate}
\emph{\ELI-forest models}
are defined likewise, but based on 
the undirected version of the graph $G_\Imc$ in Point~1. In other words, in \EL-forest models all edges in trees must point away from the root of the tree while this is not the case for \ELI-forest models. An \ELbot-forest
model is simply an \EL-forest model,
and likewise for \ELIbot.
The following is well-known \cite{DBLP:conf/dlog/Lutz08}.
\begin{lemma}
    \label{lem:forest}
    Let $\Lmc \in \{ \EL, \EL_\bot, \ELI, \ELI_\bot \}$,
    \Omc be an \Lmc-ontology, \Amc an ABox, and $q$ a UCQ. If $\Amc \cup \Omc \not \models q$, then there is an \Lmc-forest model \Imc of
    \Amc and \Omc 
    that satisfies $\Imc \not\models q$.
\end{lemma}

    Let $\Amc$ be an ABox and $p$ a CQ.
An \emph{\Amc-variation} of
$p$ is a CQ $p'$ that can be obtained from $p$ by consistently replacing zero or more variables with individual names from $\mn{ind}(\Amc)$ and possibly identifying  variables. We say that $p'$
is \emph{\Lmc-forest}, for $\Lmc \in \{ \EL, \EL_\bot, \ELI, \ELI_\bot \}$,
if the following conditions are 
satisfied:
\begin{enumerate}

\item if $r(a,b) \in p'$ with $a,b \in \NI$, then  $r(a,b) \in \Amc$;

\item $\Imc_{p'}$ is an $\Lmc$-forest model of the ABox $\Amc \cap p'$. 

\end{enumerate}
Informally, it follows from the above definition that every connected component 
$\widehat p$ of $p' \setminus \Amc$ is a tree in the directed sense if $\Lmc \in \{ \EL, \EL_\bot \}$ and in the undirected
sense otherwise. Consequently, $\widehat p$ can be viewed as
an \Lmc-concept $C_{\widehat p}$.
Some of the trees $\widehat p$ may have an individual name
$a$ as their root and thus give rise
to a rooted \Lmc-query $C_{\widehat p}(a)$. Other trees may not contain individual names and thus give rise to
an existential \Lmc-query $\exists x \, C_{\widehat p}(x)$.\footnote{For $\Lmc=\ELI$, the root is not uniquely determined. We assume that it is chosen arbitrarily.} We call these \Lmc-queries the \emph{reduction queries} of $p'$ and denote them by $\mn{red}(p')$.

\Amc-variations are tightly connected
to forest models. Indeed a homomorphism
$h$ from a CQ $q$ into an \Lmc-forest model
of some ABox \Amc induces an \Lmc-forest \Amc-variation in an obvious way:
for all variables $x$ and $y$, (i)~replace $x$ with $a$ if $h(x)=a \in \mn{ind}(\Amc)$ and then (ii)~identify any remaining variables $x$ and $y$ if $h(x)=h(y)$. 
\begin{restatable}{lemma}{propqEntIffVarEnt}
\label{prop:qEntIffVarEnt}
		Let $ \Lmc \in \{\EL,\ELbot, \ELI, \ELIbot\} $, \Amc be an ABox,
		\Imc a model of \Amc, and $ p $ a CQ.
		Then
        \begin{enumerate}

            \item If \Imc is an \Lmc-forest model of \Amc, then
                    $ \Imc \models p $ if and only if there 
		exists an \Lmc-forest \Amc-variation $ p' $ of  $ p $ 
		such that  $ \Imc \models p' $;

        \item for every  \Lmc-forest \Amc-variation $p'$ of $p$: $ 
		\Imc \models p'$  if and only if $\Imc \models q$ 
        for all  $q \in \mn{red}(p')$.
        \end{enumerate}
\end{restatable}

\subsection[Algorithms and Complexity: EL and EL\_bot]{Algorithms and Complexity: \EL and $\EL_\bot$}
\label{sect:EL}

We use Theorem~\ref{thm:charUCQwithFin} to  develop  decision procedures for $(\Lmc,\text{UCQ})$-ontology fitting with $\Lmc \in \{ \EL, \EL_\bot \}$, and also prove matching lower bounds. Clearly, the same procedures can then also be used for $(\Lmc,\text{CQ})$-ontology fitting.
The general idea is to check the existence of a finite interpretation \Imc  that satisfies Points~(a) and~(b) of
Theorem~\ref{thm:charUCQwithFin}
by independently checking the existence of each finite interpretation $\Imc_e$ from Point~(a). A  challenge is posed by 
Point~(b) due to the  infinite nature of simulations illustrated by Example~\ref{ex:simscausetrouble}. We thus first identify a suitable reformulation of Theorem~\ref{thm:charUCQwithFin}. 

The reformulation of Theorem~\ref{thm:charUCQwithFin} uses ontologies
formulated in the 
extension of $\EL_\bot$
and $\ELI_\bot$ with simulation quantifiers, as 
studied for $\EL_\bot$ in 
\cite{lutz2010enriching}. 
 Let $\Lmc \in \{ \EL, \ELI \}$.
An $\Lmc_\bot^\mn{sim}$-concept $C$ is defined
like an $\Lmc_\bot$-concept, but
additionally admits quantifiers of the form $\exists^{\mn{sim}}_\Lmc(\Imc, d)$    
where $\Imc$ is a finite interpretation with $d\in \Delta^\Imc$ that interprets only finitely many role and concept names (all others are considered empty). An
$\Lmc_\bot^\mn{sim}$-ontology is then defined as expected. 
The semantics of the additional quantifier is defined by \[
(\exists^{\mn{sim}}_\Lmc(\Jmc, e))^\Imc = \{d \in \Delta^\Imc \mid (\Jmc, e) \preceq_{\Lmc} (\Imc, d)\}.
\]
\emph{$\Lmc_\bot^{\mn{sim}}$-query entailment}
is the problem to decide, given an $\Lmc_\bot^{\mn{sim}}$-ontology \Omc, an ABox \Amc, and an \Lmc-query $q$, whether \mbox{$\Amc \cup \Omc \models q$}. \emph{Finite \Lmc-query entailment} is defined likewise, based on `$\models^{\mn{fin}}$'.
The
following is a consequence of Theorems~10 and~11 of  
\cite{lutz2010enriching}.
\begin{theorem} 
\label{thm:ELsimPTime}
Finite $\EL_\bot^{\mn{sim}}$-query entailment is in \PTime.
\end{theorem}
Actually, Theorem~\ref{thm:ELsimPTime}
is proved in \cite{lutz2010enriching}
for the unrestricted case rather than
for the finite case, but it is shown in the same article that the two cases coincide.

\smallskip

We now state the announced reformulation, which will
enable us to implement $(\Lmc,\text{UCQ})$-ontology fitting, for $\Lmc \in \{ \EL, \EL_\bot \}$, in terms of finite $\EL_\bot^{\mn{sim}}$-query entailment.
\begin{restatable}{theorem}{charVarSimQuant}\label{lem:charVarSimQuant}
    Let $E=(E^+, E^-)$ be a collection of labeled ABox-UCQ examples with $E^-\neq\emptyset$ and $ 
    \Lmc \in \{\EL, \ELbot,\allowbreak  \ELI, \ELIbot\}$. Then  
    for each $e \in E^+$, there is a set $\Gamma_{e}$ of $\Lmc_\bot^{\mn{sim}}$-ontologies, each of size linear in $||e||$, such that
    the following are equivalent:

    \begin{enumerate}

    \item there is an \Lmc-ontology that fits  $ E $;
    
        \item   
    there is a finite interpretation \Imc such that
    \begin{enumerate}

        \item $ \Imc = \biguplus_{e \in E^-}  \Imc_e$ where, for each 
					 $ e=( \Amc, q ) \in E^- $, $ \Imc_e $ is a model of \Amc such that $ \Imc_e \not \models p' $, for every \Lmc-forest \Amc-variation $p'$ of a CQ in $q$;
                     
        \item for all $e \in E^+$, there exists some $\Omc_e \in \Gamma_e$ such that $\Imc \models \Omc_e$.
            \end{enumerate}
    \end{enumerate}
\end{restatable}
\noindent
\begin{proof}
    Let $e=(\Amc, q) \in E^+$ and $u$ be a fresh role name. 
    The set $\Gamma_{e}$ contains the following ontologies:
    \begin{itemize}
        \item for every \Lmc-forest \Amc-variation $p'$ of a CQ in $q$ and every way to choose an individual $a_{q'}\in \mn{ind}(\Amc)$ for every existential reduction query $q'$ of $p'$,  the ontology $\Omc$  that contains the following CIs:
        \begin{align*}
            \exists_\Lmc^{\mn{sim}}(\Amc, a) &\sqsubseteq C &\text{ for every }C(a) \in \mn{red}(p');\\
            \exists_\Lmc^{\mn{sim}}(\Amc, a_{q'}) &\sqsubseteq \exists u. C &\hspace*{-0.2cm}\text{for every } q'=\exists x\, C(x) \in \mn{red}(p').
        \end{align*}
        \item if $\Lmc \in \{ \EL_\bot,\ELI_\bot \}$, then 
     for every $a \in \mn{ind}(\Amc)$, the ontology $\Omc_{a} = \{\exists_\Lmc^{\mn{sim}}(\Amc, a) \sqsubseteq \bot\}.
    $
    \end{itemize}
    It is shown in the appendix that  $\Gamma_e$ is as required.
\end{proof}
The new formulation of Point~(b)
is helpful in the following way.
Before independently checking the existence of the interpretations $\Imc_e$ from Point~(a) of 
Theorem~\ref{lem:charVarSimQuant},
we choose 
an ontology $\Omc_e \in \Gamma_e$
for every $e \in E^+$ and make sure that each $\Imc_e$ satisfies all chosen  $\Omc_e$. The existence check for a single $\Imc_e$ then amounts to a 
suitable collection of  non-entailment 
checks of the form $\Amc \cup \bigcup_{e \in E^+}\Omc_{e} \not\models^{\mn{fin}} q'$.

\begin{theorem}
\label{thm:CQinSigmaptwo}
    Let $\Lmc \in \{ \EL, \ELbot \}$. Then  $(\Lmc,\text{UCQ})$-ontology fitting  is in $\Sigma^\textsc{P}_2$.
\end{theorem}    

To prove Theorem~\ref{thm:CQinSigmaptwo},  let $\Lmc \in \{ \EL, \ELbot \}$. We use the following algorithm. Given a collection $E=(E^+,E^-)$ of  labeled ABox-UCQ examples with 
$ E^- \neq \emptyset  $, the algorithm first guesses, for each
positive example $e\in E^+$,
an ontology $\Omc_{e}$ from the
set $\Gamma_{e}$. It then co-guesses, for every negative example
$(\Amc,q) \in E^-$ and every CQ $p$
in~$q$, an \Lmc-forest $\Amc$-variation
$p'$ of $p$ and uses Theorem~\ref{thm:ELsimPTime} to verify in \PTime that $\Amc \cup \bigcup_{e \in E^+}\Omc_{e} \not\models^{\mn{fin}} q'$ for some  $q' \in \mn{red}(p')$ (of which there are at most $|\mn{ind}(\Amc) \cup \mn{Var}(q)|$ many).
The following is proved in the appendix.
\begin{restatable}{lemma}{elicorrectness}
    \label{lem:elicorrectness}
 The algorithm accepts its input $E$ if
 and only if there is an \Lmc-ontology that fits $E$.
\end{restatable}
The proof of Lemma~\ref{lem:elicorrectness} crucially exploits that $\EL_\bot^{\mn{sim}}$ is a Horn DL. More precisely,
it is shown in \cite{lutz2010enriching}
that for every ABox \Amc and $\EL_\bot^{\mn{sim}}$-ontology~\Omc, there is a finite \emph{universal model} $\Imc_{\Amc \cup \Omc}$, that is, $\Imc_{\Amc \cup \Omc}$ is a model of $\Amc$ and \Omc that admits an \Lmc-simulation $S$ into every model \Jmc of 
$\Amc$ and \Omc such that $S$ is the identity
on $\mn{ind}(\Amc)$.This universal model serves as the `glue' that allows us to recover a single interpretation $\Imc_e$ from the different checks that our algorithm makes for every \Amc-variation $p'$ of a negative example $e=(\Amc,q)$ and CQ $p$ in $q$.

\smallskip

We also show that the complexity of our algorithm is optimal. The following theorem is proved by reduction from
\mn{2}-\mn{COLORING}-\mn{EXTENSION}.
\begin{restatable}{theorem}{CQSigmaptwohard}
    \label{thm:CQSigmaptwohard}
    Let $\Lmc \in \{ \EL, \ELbot \}$. Then  $(\Lmc,\text{rooted CQ})$-ontology fitting  is $\Sigma^\textsc{P}_2$-hard.
\end{restatable}
In summary, we have shown that for $\Lmc \in \{ \EL, \EL_\bot \}$ and
$\Qmc \in \{ \text{CQ}, \text{UCQ} \}$, 
$(\Lmc,\Qmc)$-ontology fitting is $\Sigma^\textsc{P}_2$-complete, in the rooted and in the non-rooted case.

\subsection[Algorithms and Complexity: ELI and ELI\_bot]{Algorithms and Complexity: \ELI and $\ELI_\bot$}
\label{sect:ELI}
We now use Theorem~\ref{lem:charVarSimQuant} to develop decision procedures for $(\Lmc,\text{UCQ})$-ontology fitting with $\Lmc \in \{ \ELI, \ELI_\bot \}$, and also prove matching lower bounds. To use Theorem~\ref{lem:charVarSimQuant}, we have to work with $\ELI_\bot^{\mn{sim}}$-ontologies. However, in contrast to the case of $\EL_\bot^{\mn{sim}}$, we are not aware that this DL has ever
been studied in the literature. We therefore
study it here.
A first observation is that, in contrast to $\EL_\bot^{\mn{sim}}$, finite and infinite query entailment no longer coincide.
\begin{example}
  Let the $\ELI_\bot^{\mn{sim}}$-ontology \Omc  contain the CIs
  $$
       A \sqsubseteq \exists r . A \qquad 
\exists^{\mn{sim}}_\ELI(\Imc,d) \sqsubseteq \bot
  $$
  where
  $\Imc_\circlearrowleft$ is the interpretation with $\Delta^{\Imc_\circlearrowleft}=\{d\}$ and
  $r^{\Imc_\circlearrowleft}=\{(d,d)\}$.
  For $\Amc = \{ A(a) \}$, we then have
  $$
  \Amc \cup \Omc \not\models B(a)
  \text{ but }
  \Amc \cup \Omc \models^{\mn{fin}} B(a).
  $$
\end{example}
The above example is closely related to the failure of the finite model property for the two-way $\mu$-calculus \cite{bojarnczyk2002two}.
Our main objective is to establish the following.
\begin{theorem}
\label{thm:ELIsimexptime}
 $\ELI_\bot^{\mn{sim}}$-query entailment is \ExpTime-complete, both in the finite and in the unrestricted case.
\end{theorem}
We only need finite entailment for our
purposes, but clarifying the unrestricted case does not require extra effort. Before proving Theorem~\ref{thm:ELIsimexptime}, we 
 observe that it yields the following.
\begin{restatable}{theorem}{ELIrootedUCQinExpTime}
    \label{thm:ELIrootedUCQinExpTime}
    Let $\Lmc \in \{ \ELI, \ELIbot \}$. Then  $(\Lmc,\text{UCQ})$-ontology fitting  is in \ExpTime.
\end{restatable} 
The algorithm is essentially the same as for \EL and~$\EL_\bot$,
except that we enumerate all possibilities instead of guessing or co-guessing. Given a collection $E=(E^+,E^+)$ of  labeled ABox-UCQ examples with 
$ E^- \neq \emptyset  $, the algorithm  uses an outer loop to iterate over
all ways to choose for each
positive example $e\in E^+$,
an ontology $\Omc_{e}$ from the
set~$\Gamma_{e}$. In an inner loop, it  iterates over all ways to choose, for every negative example
$(\Amc,q) \in E^-$ and every CQ $p$
in~$q$, an \Lmc-forest $\Amc$-variation
$p'$ of $p$. It then uses Theorem~\ref{thm:ELIsimexptime} to verify in \ExpTime that $\Amc \cup \bigcup_{e \in E^-}\Omc_{e} \not\models q'$ for some $q' \in \mn{red}(p')$.
Clearly, we obtain an \ExpTime algorithm overall. The correctness proof is almost identical to that
of Lemma~\ref{lem:elicorrectness}, details are omitted.

\medskip

We now return to the proof of Theorem~\ref{thm:ELIsimexptime}.
With \emph{\ELI-subsumption w.r.t.\
$\ELI_\bot^{\mn{sim}}$-ontologies},
we mean the problem to decide, given two
\ELIbot-concepts $C$ and $D$ and 
an $\ELI_\bot^{\mn{sim}}$-ontology
\Omc, whether $C^\Imc \subseteq D^\Imc$ for all models \Imc of \Omc.
If this is the case, then we write
$\Omc \models C \sqsubseteq D$ and
say that \emph{$C$ is subsumed by $D$ w.r.t.\ \Omc}. There is also a finite version where only finite models \Imc are considered, noted $\Omc \models^{\mn{fin}} C \sqsubseteq D$. The following is essentially a consequence of the
results in \cite{vardi1998reasoning} and \cite{bojarnczyk2002two}.
\begin{restatable}{theorem}{thmELIsimsubsexptime}
\label{thm:ELIsimsubsexptime}
   \ELIbot-subsumption w.r.t.\
$\ELI_\bot^{\mn{sim}}$-ontologies is  \ExpTime-complete, both in the finite and in the unrestricted case.
\end{restatable}    
\noindent
\begin{proof} (sketch)
   The  upper bound can be obtained as follows. Let $\Lmc_\mu$ denote the
   two-way $\mu$-calculus with simultaneous fixed points, see 
   e.g.\ \cite{vardi1998reasoning} for the
   conventional two-way $\mu$-calculus
   and \cite{DBLP:books/el/07/BradfieldS07} for $\mu$-calculi with
   simultaneous fixed points. Given \ELIbot-concepts $C,D$
   and an $\ELI_\bot^{\mn{sim}}$-ontology \Omc, one may construct
   in polynomial time $\Lmc_\mu$-formulas $\varphi_C, \varphi_D, \varphi_\Omc$ such that $\Omc \models^{(\mn{fin})} C \sqsubseteq D$ if and only if $\varphi_O \wedge \varphi_C \wedge \neg \varphi_D$ is (finitely) unsatisfiable.
   It is folklore that the standard translation of an $\Lmc_\mu$-formula without simultaneous fixed points into a two-way parity  automaton given in
   \cite{vardi1998reasoning} 
   and also used in \cite{bojarnczyk2002two} can be generalized to simultaneous fixed points while remaining in polynomial time. It remains to decide emptiness of the
   resulting automaton, which is possible in \ExpTime both in the unrestricted \cite{vardi1998reasoning} and in the finite case \cite{bojarnczyk2002two}.

   The lower bounds are inherited from subsumption in \ELI, which is \ExpTime complete. \cite{BaaderEtAl-OWLED08DC}. This also covers the finite case since \ELI enjoys the finite model property regarding subsumption. 
\end{proof}
We next prove Theorem~\ref{thm:ELIsimexptime}
using Theorem~\ref{thm:ELIsimsubsexptime}, thus completing the proof of .
Note that  $\ELI_\bot^{\mn{sim}}$-query entailment 
is formulated in terms of ABoxes while there are no ABoxes in the problem of \ELIbot-subsumption w.r.t.\
$\ELI_\bot^{\mn{sim}}$-ontologies.
It is nevertheless possible to reduce the former to the latter in polynomial time. The central idea is to replace
 ABoxes by their approximations expressed 
through a simulation quantifier.
This is summarized by the following
lemma.
\begin{restatable}{lemma}{finmodontologyfinmodkb}
    \label{lem:finmodontologyfinmodkb}
Let $\Lmc \in \{\ELI, \ELIbot\}$, \Amc be an ABox,  and \Omc an $\Lmc^{\mn{sim}}_\bot$-ontology. Further let $\Omc_\Amc$ be
the ontology that contains the following CIs, for all $a \in \mn{ind}(\Amc)$:
$$
  \top \sqsubseteq \exists u.\;S_{\Amc,a} \qquad S_{\Amc,a} \sqsubseteq \exists^{\mn{sim}}_\Lmc ({\Amc}, a),
$$
where $u$ is a fresh role name and all $S_{\Amc,a}$ are fresh concept names.
Then the following hold for all \ELIbot-concepts $C$ and $a \in \mn{ind}(\Amc)$:
\begin{enumerate}
    \item 
$ \Omc \cup \Omc_\Amc \models^{\mn{fin}} S_{\Amc, a} \sqsubseteq C \ \text{ if and only if }\ \Amc \cup \Omc \models^{\mn{fin}} C(a)$;
\item $ \Omc \cup \Omc_\Amc \cup \{C \sqsubseteq \bot\} \models^{\mn{fin}} S_{\Amc, a} \sqsubseteq \bot \text{ for all } a \in \mn{ind}(\Amc) \ \text{ if and only if }\ \Amc \cup \Omc \models^{\mn{fin}} \exists x\, C(x)$.
\end{enumerate}
and the same is true in the unrestricted case.
\end{restatable}
We also prove a matching lower bound,
by reduction from the complement of AQ entailment in \ELI which is \ExpTime-hard \cite{BaaderEtAl-OWLED08DC}.

\begin{restatable}{theorem}{ELIfitExpT}
    Let $\Lmc \in \{\ELI, \ELIbot\}$. The $(\Lmc, \text{rooted CQ})$-ontology fitting problem is \ExpTime-hard.
\end{restatable}
W have thus shown that for $\Lmc \in \{ \ELI, \ELI_\bot \}$ and
$\Qmc \in \{ \text{CQ}, \text{UCQ} \}$, 
$(\Lmc,\Qmc)$-ontology fitting is $\Sigma^\textsc{P}_2$-complete, in the rooted and in the non-rooted case.

\section{Conclusion}
\label{sect:concl}

We have obtained a relatively complete picture of ontology fitting problems for the most common Horn DLs based on ABoxes and queries. 
It would, however, be interesting to extend our study to other relevant Horn DLs including $\mathcal{ELO}$, Horn-\ALC, and Horn-\ALCI. Existential rule languages such as guarded existential rules would be another interesting subject to study.

One may also be interested in the case of full UCQs, that is, UCQs that do not use existentially quantified variables. We conjecture that
it is not too difficult to extend the techniques that we have introduced for AQs  to full UCQs, and that $(\Lmc,\text{full UCQ})$-ontology fitting is in \PTime for all $\Lmc \in \{ \EL, \EL_\bot, \ELI, \ELI_\bot \}$. In the case of \ALC and \ALCI, such an extension has been presented in \cite{FGL-KR25}.

It should  be straightforward to adapt 
the characterizations and decidability results presented in this paper to the DLs 
$\EL^{\mn{sim}}_\bot$ and 
$\ELI^{\mn{sim}}_\bot$. An interesting nuance is that 
the qualifier `finite' in Theorem~\ref{thm:charUCQwithFin} can then be dropped for both $\EL^{\mn{sim}}_\bot$ \emph{and} 
$\ELI^{\mn{sim}}_\bot$.

Another question left open is the size of fitting ontologies in $(\ELI,\text{UCQ})$-ontology fitting, where we have obtained no bounds at all. In particular, it would be interesting to understand whether polynomial size fittings always exist if auxiliary symbols are admitted. 

\section*{Acknowledgements} The third author was supported by DFG project LU 1417/4-1. The authors acknowledge the financial support by the Federal Ministry of Research, Technology and Space of Germany and by Sächsische Staatsministerium für Wissenschaft, Kultur und Tourismus in the programme Center of Excellence for AI-research 'Center for Scalable Data Analytics and Artificial Intelligence Dresden/Leipzig', project identification number: ScaDS.AI.

\bibliographystyle{kr}
\bibliography{local}

\cleardoublepage

\appendix

\section{Proofs for Section~\ref{sect:prelims}}
To present the subsequent constructions uniformly, we fix, once and for all,  $\mn{N}_{\mn{R}}^\Lmc$ to denote the set of role names if $\Lmc \in \{\EL, \ELbot\}$, and the set of all roles (including inverse roles) if $\Lmc \in \{\ELI, \ELIbot\}$.
\lemcharconc*
\noindent
\begin{proof}
    Let $\Lmc \in \{\EL, \ELI\}$ and \Imc be an interpretation.
    For every element $d \in \Delta^\Imc$, we define its characteristic \Lmc-concept inductively as follows:
    \begin{align*}
        C_{\Imc,d}^{0,\Lmc} &:= \bigsqcap_{d \in A^\Imc} A \\
        C_{\Imc,d}^{k+1, \Lmc} &:= C^{0, \Lmc}_{\Imc, d} \sqcap \bigsqcap_{\substack{(d,e) \in r^\Imc, \\ r\in \mn{N}_{\mn{R}}^\Lmc}} \exists r. C^{k, \Lmc}_{\Imc, e} \qquad \text{ for all } k\in\mathbb{N}, 
    \end{align*}
    For $\Lmc = \EL$, the correctness of the above construction was  proved in Proposition~5.1.5 of \cite{piro2013model}. It is straightforward to see that the  argument generalizes to \ELI.
\end{proof}

\section{Proofs for Section~\ref{sect:consistency}}

\lemIASLsubI*
\noindent
\begin{proof}
    Let \Amc, \Imc, \Omc, $S$,   and \Lmc be as in the lemma,
    and assume that $\Imc \models \Omc$. Recall
    that $\Imc_{\Amc,S,\Lmc}$ is stitched together
    from interpretations $\Jmc^{\downarrow\Lmc}_a$, $a \in  \mn{ind}(\Amc)$ where $\Jmc_a =  
	   \prod_{(a,d) \in S} ( \Imc, d )$ with the distinguished element $d_\Pi$ renamed to $a$. Thus every element $p \in \Jmc^{\downarrow\Lmc}_a$ is a path with $\mn{tail}(p)$ an element of $\Jmc_a$, or, as a special case, $p=a$. In the latter case, for uniformity we set $\mn{tail}(p)=d_\Pi$.
    We show the following by induction on the structure of~$C$.
    \\[2mm]
    {\bf Claim.} For all \Lmc-concepts $C$, all
    $(a,\cdot) \in S$, and all paths $p \in \Delta^{\Jmc^{\downarrow\Lmc}_a}$ with $\mn{tail}(p)=(d_i)_{i \in I}$:
    $$
      p \in C^{\Imc_{\Amc,S,\Lmc}}
      \text{ if and only if } d_j \in C^\Imc \text{ for all } j \in I    $$
    \emph{Proof of claim.} 
    The cases that $C$ is
    $\top$, $\bot$, a concept name, or a conjunction are straightforward using the induction hypothesis and Lemmas~\ref{lem:prodlem}
    and~\ref{lem:unravbasic}. 
    
    The only interesting case is thus $C = \exists r . D$, where $r$ may be an inverse role
    if $\Lmc \in \{ \ELI, \ELI_\bot \}$. Also in this case, the `if' direction is straightforward. We thus only prove the `only if' direction. Assume that
    $p \in (\exists r. D)^{\Imc_{\Amc,S,\Lmc}}$ with $\mn{tail}(p)=(d_i)_{i \in I}$. Then there is a $p' \in \Delta^{\Imc_{\Amc,S,\Lmc}}$ with $(p,p') \in r^{\Imc_{\Amc,S,\Lmc}}$ and $p' \in D^{\Imc_{\Amc,S,\Lmc}}$.

    First assume that $p' \in \Delta^{\Jmc^{\downarrow\Lmc}_a}$ and let $\mn{tail}(p')=(e_i)_{i \in I}$. The induction hypothesis
    yields $e_j \in D^\Imc$ for all $j \in I.$
    The constructions 
    of $\Imc_{\Amc,S,\Lmc}$, $\Jmc^{\downarrow\Lmc}_a$, and $\Jmc_a$ further yield $(d_j,e_j) \in r^{\Imc}$ for all $j \in I$. Thus, $d_i \in C^\Imc$ for all $i \in I$, as required. 
    
    The interesting case is that
    $p' \notin  \Delta^{\Jmc^{\downarrow\Lmc}_a}$. Then, by construction 
    of $\Imc_{\Amc,S,\Lmc}$, we must have $p=a$ and $p'=b \in \mn{ind}(\Amc)$ for some $b$ with $r(a,b) \in \Amc$. 
    For every $i \in I$,
    by construction of $\Jmc_a$
    we know that $(a,d_i) \in S$.
    Since $r(a,b) \in \Amc$,
    this implies $(b,e_i) \in S$ for some $e_i$ with $(d_i,e_i) \in r^\Imc$.
     Since $(b,e_i )\in S$, $e_i$ is among the elements of $\mn{tail}(b)=(e_j')_{j \in J}$ for all $i \in I$, by construction of $\Imc_{\Amc, S, \Lmc}$. Applying the induction hypothesis, we get $e_i \in D^\Imc$ for all $i \in I$. It thus follows that $d_i \in (\exists r. D)^\Imc$ for all $i \in I$, as required. This finishes the proof of the claim.

    \medskip
    Take any CI
    $C \sqsubseteq D \in \Omc$ 
	and  $p \in \Delta^{\Imc_{\Amc,S, \Lmc}}$ with  $p \in C^{\Imc_{\Amc,S, \Lmc}}$. 
    By construction of $\Imc_{\Amc,S, \Lmc}$, there is an $a \in \mn{ind}(\Amc)$,
    such that $p \in \Delta^{\Jmc^{\downarrow\Lmc}_a}$. First assume 
    that there is no $(a,\cdot) \in S$, which is only possible when $\Lmc \in \{\EL,\ELI\}$. Then $\Jmc_a$
    is the maximal interpretation and consequently $p \in E^{\Jmc^{\downarrow\Lmc}_a}$
    for every \ELI-concept $E$. This in
    particular implies $p \in D^{\Imc_{\Amc,S, \Lmc}}$, as required. Now assume that there
    is an  $(a,\cdot) \in S$ and let 
    $\mn{tail}(p)=(d_i)_{i\in I}$. Then the
    claim yields $d_i \in C^\Imc$ for all $i \in I$. Since \Imc is a model
    of \Omc, we obtain $d_i \in D^\Imc$ for each $i \in I$.  By applying the claim a second time we derive $p \in D^{\Imc_{\Amc,S, \Lmc}}$, as required.
    
\end{proof}

Hereinafter, for any relation $S \subseteq A \times B$, we use the notation $aS$ to denote the set $ \{b\mid(a,b)\in S\}$, thereby treating $S$ as a function from $A$ to the powerset of $B$. 

Using this notation, the following corollary is an immediate consequence of the preceding proof.
\begin{corollary} \label{lem:IASLsim}
	Let $\Lmc \in \{\EL, \ELbot, \ELI, \ELIbot\}$, \Amc be an ABox, $a \in \mn{ind}(\Amc)$, \Imc an interpretation, and $S$ an \Lmc-simulation from \Amc to \Imc. Then,  \[ \Imc_{\Amc, S, \Lmc} \models C (a) \quad \text{ if and only if } \quad aS\subseteq C^\Imc \]
    for every \Lmc-concept $C$.
\end{corollary}
\noindent

\lemsecondontconstrcorr*
\noindent
\begin{proof} 
Assume that $E$ admits a fitting \ELbot-ontology. To show that $\Omc_{\mn{alt}}$ fits the positive examples, 
let $\Amc \in E^+$.
It can be readily verified that the interpretation $\Jmc$, defined by 
\begin{align*}
    \Delta^\Jmc &:= \mn{ind}(\Amc),\\
    A^\Jmc & := \{a \mid A(a) \in \Amc\}, \hspace{-0.6cm}& \text{ for all } A\in \NC \text{ with } A\neq \overline{V_b},\\
    \overline{V_a}^\Jmc&:= \mn{ind}(\Amc)\setminus \{a\} \hspace{-0.6cm}& \text{ for all } a \in \mn{ind}(\Amc),\\
    r^\Jmc & := \{(a,b) \mid r(a,b)\in \Amc\} \hspace{-0.4cm}& \text{ for all }r\in \NR,
\end{align*}
is a model of $\Amc \cup \Omc_{\mn{alt}}$.

In order to show that $\Omc_{\mn{alt}}$ fits the negative examples, let $\Amc \in E^-$. The main argument of the proof boils down to the following claim.
\\[2mm]
    {\bf Claim.} If $ \Amc \cup \Omc_{\mn{alt}}$ is consistent, then there exists an \ELbot-simulation from \Amc to $\Amc^+ : = \biguplus_{\Bmc \in E^+} \Bmc$.
\\[1mm]
Suppose that $\Amc \cup \Omc_{\mn{alt}}$ is consistent with witnessing model \Imc.
Define the relation $S \subseteq \mn{ind}(\Amc) \times \mn{ind}(\Amc^+)$ by
\[
S = \left\{(a,b) \in\mn{ind}(\Amc) \times \mn{ind}(\Amc^+)  \mid a \notin \overline V_b^\Imc \right\}.
\]
We show that $S$ is indeed an \ELbot-simulation.
First, consider the Condition~(atom). Assume that $A(a) \in \Amc$ and $(a,b) \in S$. Since \Imc is a model of \Amc, we have $a \in A^\Imc$. 
Because \Imc is also a model of $\Omc_{\mn{alt}}$, it follows that $a \in \overline{V_c}{\vphantom{V}}^\Imc$ for all $c \in \mn{ind}(\Amc^+)$ with $A(c) \notin \Amc^+$. By definition of $S$, $a \notin \overline{V_b}{\vphantom{V}}^\Imc$ and thus $A(b) \in \Amc^+$. Next, consider Condition~(Rel). Assume that $r(a,a') \in \Amc$ and $(a,b) \in S$. Then, $a \notin \overline{V_b}{\vphantom{V}}^\Imc$ by definition. Since $\Imc$ is a model of $\Omc_{\mn{alt}}$, we obtain \[a \not\in (\exists r. \bigsqcap_{r(b,b') \in \Amc^+} \overline{V_{b'}})^\Imc.\]
As $(a,a') \in r^\Imc$, this implies $a' \not\in \overline{V_{b'}}{\vphantom{V}}^\Imc$ for some $r(b,b') \in \Amc^+$. By definition of $S$, this yields $(a',b') \in S$, as required.  
Finally, to see that $S$ satisfies Condition~(total), observe that since \Imc is a model of $\Omc_{\mn{alt}}$, for every element $d \in \Delta^\Imc$ there exists some $b \in \mn{ind}(\Amc^+)$ such that $d \notin \overline{V_b}{\vphantom{V}}^\Imc$. In particular, this holds for every $ a \in \mn{ind}(\Amc)$, so $S$ is left-total. Consequently, $S$ is an \ELbot-simulation, which proves the claim.
\medskip 

Now, since $E$ admits a fitting \ELbot-ontology, from Theorem~\ref{thm:consfitting} we obtain $\Amc \not \preceq_{\ELbot} \Amc^+$. By contraposition, the claim implies that $\Amc \cup \Omc_{\mn{alt}}$ is inconsistent, proving that $\Omc_{\mn{alt}}$ also fits the negative examples.
\end{proof}

It is easy to see that the construction of the alternative fitting ontology as well as the proof of correctness can be extended to \ELIbot.

\section{Proofs for Section~\ref{sect:AQs}}

\thmcharaq*
\noindent
\begin{proof}
`1$~\Rightarrow~$2'.
Suppose that $ E $ admits a fitting \Lmc-ontology~\Omc. Then for every  $e= (
\Amc, Q(a) ) \in E^-$ there is a model $\Imc_e$ of \Amc and \Omc with  $ \Imc_e \not
\models Q(a)$. 
%
Since \Lmc-ontologies are invariant under disjoint union, the interpretation $ \Imc = \biguplus_{e \in E^-}
\Imc_e$ is also a model of \Omc. Moreover, \Imc is clearly a model of $\Amc^-$ and therefore induces the following completion of $ \Amc^- $:
\[
	\Cmc = \Amc^- \cup \{ Q(b)  \mid  ( \Amc, Q(a) )\in E^+ ,b \in \mn{ind}(\Amc^-) \cap Q^\Imc  \} 
.\] 
Note that \Cmc  satisfies Condition~(a) by assumption. It is thus left to show that
Condition~(b) holds as well. We prove this by contradiction. 

Assume that for some  $ ( \Amc, Q(a)
) \in E^+ $ there is a \Lmc-simulation $ S $ from  \Amc to \Cmc such that for some $ ( a,b
)\in S  $ we have $ Q(b) \not\in \Cmc $. By construction of \Cmc and since \Imc is a model of~$\Amc^-$,  $ S $ is also a \Lmc-simulation from
\Amc to 
\Imc. Now consider the interpretation $
\Imc_{\Amc, S, \Lmc} $ which, by construction, is a model of \Amc. Since \Imc is a model of
\Omc, by Lemma~\ref{lem:IASLsubI}, $ \Imc_{\Amc, S, \Lmc} $ is also 
a model of \Omc. However, $Q(b) \notin \Cmc$ implies $b \notin Q^\Imc$  by design
of \Cmc. Together with $(a,b) \in S$, the
construction of $ \Imc_{\Amc, S, \Lmc} $ thus
 yields $ a \not \in Q^{\Imc_{\Amc, S, \Lmc}}
$. This means that $ \Imc_{\Amc, S, \Lmc} \not \models Q(a) $ and thus $\Imc_{\Amc, S, \Lmc}$
witnesses that $ \Omc $ does not fit  $ E $, a contradiction.

\medskip
`2$~\Rightarrow~$1'. Assume that there is a completion \Cmc of $ E $ that satisfies Conditions~(a) and~(b). We show how to construct
an \Lmc-ontology that fits $E$. We first
define an ontology $\Omc_e$ for each   $ e=( \Amc, Q(a) ) \in E^+$, distinguishing two cases:
\begin{enumerate}
	\item There is no \Lmc-simulation from \Amc to \Cmc. 
    
    This is only possible if $\Lmc \in \{\ELbot, \ELIbot\}$ because then simulations are required to be total. As in the \mbox{`(ii)~$\,\Rightarrow\,$~(i)'} direction of the proof of Theorem
		\ref{thm:consfitting}, we find an $ a_\Amc \in \mn{ind}( \Amc )  $ 
		such that for all $ b \in \mn{ind}(\Cmc)$,  \[
		( \Amc, a_\Amc ) \not \preceq^\ell_{\Lmc} ( \Cmc, b )
		\]
		where $ \ell = | \Cmc | \cdot |\Amc| $. We then set \[
			\Omc_e := \{C_{\Amc, a_\Amc}^{\ell,\Lmc} \sqsubseteq \bot\}  
		.\] 
	\item There is an \Lmc-simulation from \Amc to \Cmc.
    
    We set \[
			\Omc_e := \{ C_{\Amc, a}^{\ell,\Lmc} \sqsubseteq Q\}  
		,\]
		where $ \ell = |\Cmc| \cdot |\Amc| $.
        
\end{enumerate}
We show that $\Omc := \bigcup_{e \in E^+} \Omc_e $ fits $E$. For each
$ e= ( \Amc, Q(a) )  \in E^+ $, we  have $ \Amc \cup \Omc_e \models Q(a)
$, thus $ \Amc \cup \Omc \models Q(a)
$ as required. In both cases this is due to $b \in (C_{\Amc, b}^{\ell,\Lmc})^\Jmc$ for all models \Jmc of \Amc and individuals $b \in \mn{ind}(\Amc)$.

Next 
consider a negative example $ e=( \Amc, Q(a) ) \in E^-  $. Let $ \Imc_{\Cmc} $ be \Cmc viewed
as an interpretation. We show that $\Imc_\Cmc$
is a model of \Amc and \Omc with $\Imc_\Cmc \not\models Q(a)$, thus proving $\Amc \cup \Omc \not\models Q(a)$, as required. In fact, it is
clear that $\Imc_\Cmc$ is a model of \Amc since
$\Amc^- \subseteq \Cmc$. Moreover, $\Imc_\Cmc \not\models Q(a)$  because \Cmc satisfies
Condition~(a). It thus remains to show that $\Imc_\Cmc$
is a model of \Omc.

To this end, it suffices to show that $\Imc_\Cmc$ is a
model of every ontology $\Omc_{e'}$,  $e' \in E^+$.
Take any such $ e'=( \Bmc, P(b) )$. 
By combining Lemma~\ref{lem:finmodsim} and Lemma~\ref{lem:charconc}, we obtain that the following  are
equivalent for all $ c \in \mn{ind}( \Bmc )  $ and $ d \in \Delta^{\Imc_\Cmc} $:
\begin{enumerate}
	\item $ ( \Bmc, c ) \preceq_\Lmc ( \Imc_\Cmc, d )   $;
	\item $ ( \Bmc, c ) \preceq^\ell_\Lmc ( \Imc_\Cmc, d )   $,
		where $ \ell = | \Cmc | \cdot | \Bmc | $;
	\item $ d \in ( C_{\Bmc, c}^{\ell,\Lmc} )^{\Imc_\Cmc}  $. 
\end{enumerate}
If $ \Omc_{e'} $ was constructed in Case~1 above, then the equivalences above immediately imply
	$( C_{\Bmc, a_\Bmc}^{\ell,\Lmc} )^{\Imc_\Cmc} = \emptyset $ and thus $ \Imc_\Cmc
	\models \Omc_{e'} $. 
Now assume that $ \Omc_{e'} $ was constructed in Case~2. 
Since \Cmc satisfies Condition~(b), it follows from the same equivalences that $ ( C_{\Bmc, b}^{\ell,\Lmc} )^{\Imc_\Cmc} \subseteq
			Q^{\Imc_\Cmc}$
		 and therefore $ \Imc_\Cmc \models \Omc_{e'} $. 
\end{proof}

For an adequate discussion of the different fitting ontologies, we refer the reader to the discussion in Section~\ref{subsect:DiscFitOnt}. 

\proprefcand*
\noindent
\begin{proof}
    `1$~\Rightarrow~$2' We prove the contraposition. Thus assume that $Q(a) \notin \Amc^\pm$ for all $(\Amc, Q(a))\in E^-$. Clearly, $\Amc^\pm$ is a completion of $\Amc^-$ that satisfies Condition~(a) of  Theorem~\ref{thm:charaq}. It suffices to show that it also satisfies Condition~(b). To this end, let $(\Amc, Q(a)) \in E^+$ and 
    $(\Amc,a) \preceq_\Lmc (\Amc^\pm,b)$. Then (\Rsf) yields $Q(b) \in \Amc^\pm$, as required. 
    
    \medskip
    `2$~\Rightarrow~$1'. Let $\Cmc_0 = \Amc^- \subseteq \Cmc_1 \subseteq \ldots \subseteq \Cmc_n = \Amc^\pm$ be the sequence of ABoxes obtained by the exhaustive  application of (\Rsf). We  first observe the following:
    \\[2mm]
    {\bf Claim.} If $\Cmc$ is a completion of $\Amc^-$ that satisfies Condition~(b) of Theorem~\ref{thm:charaq}, then $\Cmc_i \subseteq \Cmc$ for all $i \leq n$.
    \\[2mm]
    \emph{Proof of claim.} We prove the claim by induction on $i$.
    Let $\Cmc$ be a completion of $\Amc^-$ that satisfies Condition~(b) of Theorem~\ref{thm:charaq}.  The base case $\Cmc_0= \Amc^- \subseteq \Cmc$ is immediate by definition of completions. Now let $i>0$ and $\Cmc_{i-1} \subseteq \Cmc$. The identity on $\mn{ind}(\Cmc_{i-1})$ is an \Lmc-simulation from $\Cmc_{i-1}$ to $\Cmc$. Assume that $\Cmc_{i}$ was obtained from $\Cmc_{i-1}$ due to an application 
    of (\Rsf) with respect to $(\Amc, Q(a))\in E^+$
    and $(\Amc,a) \preceq_\Lmc (\Cmc_{i-1},b)$. Composing this simulation with the one from $\Cmc_{i-1}$ to $\Cmc_i$ yields $(\Amc,a) \preceq_\Lmc (\Cmc_{i},b)$. Since $\Cmc$ satisfies Condition~(b) of Theorem~\ref{thm:charaq}, we then have 
    $Q(b) \in \Cmc$. This shows $\Cmc_i \subseteq \Cmc$, as required, and finishes the proof of the claim.

    \smallskip

 Assume that $Q(a)\in \Amc^\pm$ for some $(\Amc ,Q(a)) \in E^-$. To prove that Condition~1 is satisfied, it suffices to show that there is no completion \Cmc of $\Amc^-$ that satisfies
    Conditions~(a) and~(b) of Theorem~\ref{thm:charaq}. 
    Suppose that \Cmc is a completion satisfying Condition~(b). By
    assumption, we have $Q(a) \in \Amc^\pm$
    and the claim yields $Q(a) \in \Cmc$, thus
    \Cmc does not satisfy Condition~(a).
\end{proof}

\section{Proofs for Section~\ref{sect:UCQCQ}}
\label{sect:appendUCQ}

Before proving the results of Section~\ref{sect:UCQCQ}, we formalize the concepts behind the reduction queries, a well-known machinery used for ontology-mediated querying in description logics.
These concepts encode tree-shaped CQs in the spirit of the characteristic concepts introduced earlier.
Although the characteristic concept would suffice to capture the structure of a tree-shaped query, its size may be exponential in the size of the query.
This blow-up occurs in the presence of inverse roles, which enable the inductive definition to repeatedly traverse edges towards the root, resulting in a redundant encoding of identical substructures.
We therefore refine this construction to avoid such redundancies. The resulting concept is linear in the size of the tree and equivalent to the characteristic concept.

Let $\Lmc \in \{\EL, \ELbot, \ELI, \ELIbot\}$. An \emph{\Lmc-tree} CQ $T$ is a tuple $(p,t_r)$ consisting of a connected CQ $p$ and an element $t_r \in \mn{ind}(p) \cup \mn{Var}(p)$ satisfying the following conditions:
\begin{itemize}
    \item if $\Lmc \in \{\EL, \ELbot\}$, then $ p$ is a directed tree with root $t_r$;
    \item if $\Lmc \in \{\ELI, \ELIbot\}$, then $p$ viewed as an undirected graph is an undirected tree with root $t_r$.
\end{itemize}
For simplicity, we adopt the convention that each $(\emptyset, t_r)$, where $t_r$ is an individual or variable, is an \Lmc-tree CQ, interpreting it as an unlabeled one-node tree with root $t_r$.

  In the following, let $T=( p, t_r)$ be an \Lmc-tree CQ. 
  For each $r(t_r, t) \in  p$ we set $T_{|t}:=(p ', t)$, where $ p'$ is the restriction of $ p$ to the connected component of $p \setminus \{r(t_r,t)\}$ that contains $t$, if such a component exists, and $p'=\emptyset$ otherwise. By construction, $T_{|t}$ is also an \Lmc-tree CQ. 
The  \Lmc-concept $C^\Lmc_{T}$ that encodes $T$ can then be inductively defined as follows:
\begin{align*}
    C_{T}^{ \Lmc} &:= \bigsqcap_{A(t_r) \in p} A \sqcap \bigsqcap_{r(t_r,t) \in p} \exists r. C^{\Lmc}_{T_{|t}} \qquad 
\end{align*}
Observe that $\mn{rd}(C^\Lmc_{T_{|t}})<\mn{rd}(C^\Lmc_{T})$ for every $r(t_r,t)\in p$.
It is also clear from the construction that $C^\Lmc_{T}$ is of linear size in $||p||$. 

\begin{lemma} \label{lem:treeconcept}
    Let $\Lmc \in \{\EL, \ELbot, \ELI, \ELIbot\}$, $T= (p,t_r)$ an $\Lmc$-tree CQ with $p \neq \emptyset$, \Imc an interpretation and $d \in \Delta^\Imc$. Then 
    \[
    (p,t_r) \rightarrow (\Imc, d) \quad \text{ if and only if } \quad d \in ( C^\Lmc_{T})^\Imc.
    \]
\end{lemma}
\noindent
\begin{proof}
    We proceed by induction on $\mn{rd}(C^\Lmc_{T})$. If $\mn{rd}(C^\Lmc_{T})=0$, then $p$ contains only concept atoms, each occurring in the conjunction $C^\Lmc_{T}$, and the statement follows immediately. Suppose $\mn{rd}(C^\Lmc_{T})> 0$ and assume the equivalence holds for all \Lmc-tree CQs whose corresponding concept has smaller role depth. 
    
    `\emph{only if}'. Assume that $h$ witnesses
    $(p,t_r) \rightarrow (\Imc, d)$. Then, $d \in A^\Imc$ for all $A(t_r)\in p$, so the first conjunction of $C^\Lmc_{T}$ holds at $d$.
    It remains to show that the second conjunction holds. Let $r(t_r,t)\in p$ and $T_{|t}= (p', t)$. Since $h$ is a homomorphism, $(d,h(t))\in r^\Imc$.
    Consequently, if $p'=\emptyset$, we immediately obtain $d \in (\exists r. \top)^\Imc=(\exists r.  C^\Lmc_{T_{|t}})^\Imc$. Otherwise, the restriction of $h$ to $p'$ witnesses $(p', t) \rightarrow (\Imc, h(t))$. By induction hypothesis, $h(t) \in (C^{\Lmc}_{T_{| t}})^\Imc$, and thus $d \in (\exists r. C^{\Lmc}_{T_{| t}})^\Imc$. Hence, the second conjunction of $C^\Lmc_{T}$ holds at $d$, and we conclude that $d \in (C^\Lmc_{T})^\Imc$.

    \smallskip
    `\emph{if}'. Assume that $d\in ( C^\Lmc_{T})^\Imc$. By construction of $C^\Lmc_{T}$, for every $r(t_r, t) \in p$ there exists $(d,e_t) \in r^\Imc$ such that $e_t \in ( C^{\Lmc}_{T_{|t}})^\Imc$, where $T_{|t}= (p', t)$. If $p'\neq \emptyset$, the induction hypothesis yields a homomorphism $h_t$ from $p'$ to \Imc with $h_t(t)=e_t$. Otherwise, define $h_t: \{t\} \to \Delta^\Imc$ by setting $h_t(t)=e_t$.
     Since the domains of the $h_t$ are disjoint, we may construct a function $h$ from $p$ to $\Imc$ by 
     \[
    h = \{(t_r, d)\} \cup \bigcup_{r(t_r, t) \in p} h_t.
    \]
    Observe that each $h_t$ is a homomorphism from a proper subtree of $p$ to \Imc. Thus, to verify that $h$ is a homomorphism from $p$ to $\Imc$, it suffices to consider the assertions of the form $A(t_r)\in p$ and $r(t_r, t) \in  p$. The former are preserved by the assumption $d \in (C^\Lmc_{T})^\Imc$, and the latter by the choice of $e_t$. Thus, $h$ witnesses $(p,t_r) \rightarrow (\Imc, d)$.
\end{proof}

It remains to establish the final connection to the \Lmc-forest \Amc-variations. Let $\Lmc \in \{\EL, \ELbot, \ELI, \ELIbot\}$, \Amc be an ABox, $p$ a CQ and $p'$ an \Lmc-forest \Amc-variation of $p$. Recall that each connected component $q$ of $p' \setminus \Amc$ is tree-shaped, as noted in the original definition. It therefore contains some $t_r \in \mn{ind}(q) \cup \mn{Var}(q)$ such that $(q, t_r)$ is an \Lmc-tree CQ. Moreover, whenever $q$ contains an individual $a$, this is the only individual that occurs in $q$ and we may assume $t_r = a$. Otherwise, we fix $t_r$ to be an arbitrary but fixed element of $\mn{Var}(q)$. 
The \emph{reduction queries} of $p'$ are the queries $ \widehat q$, one for each \Lmc-tree CQ $T=(q,t_r)$ arising from a connected component of $p' \setminus \Amc$, defined by:
\[
    \widehat q : = \begin{cases}
        C_{T}^\Lmc(a), \qquad &\text{ if } t_r \text{ is an individual};\\[1mm]
        \exists x\, C_{T}^\Lmc(x), & \text{ otherwise}.
    \end{cases}   
\]
We use $\mn{red}(p')$ to denote the set of reduction queries of $p'$.
\propqEntIffVarEnt*
\noindent
\begin{proof}
We start with Point~1. `\emph{only if}'. Let $\Imc$ be an \Lmc-forest model of \Amc and suppose that $ \Imc \models p$, witnessed by a homomorphism $g$ from $p$ to $\Imc$. Using $g$, we construct an \Lmc-forest \Amc-variation $p'$ of $p$ as follows. First, for every variable $x$ with $g(x)=a\in\mn{ind}(\Amc)$, replace $x$ by $a$. Second, identify variables $x$ and $y$ whenever $g(x) = g(y) \notin \mn{ind}(\Amc)$. Since \Imc is a forest model of \Amc and $g$ is a homomorphism that satisfies $g(a)=a$ for all individuals $a$ in $p$, the result $p'$ is an \Lmc-forest \Amc-variation of $p$. 
    Now let $h = g \cup \mathrm{id}_{|\mn{ind}(\Amc)}$ and $h'$ be the restriction of $h$ to the terms occurring in $p'$. It can be readily verified that $h'$ is a homomorphism from $p'$ to \Imc satisfying $h'(a)=a$ for all individual $a$ occurring in $p'$.
	
	\smallskip
	`\emph{if}'. For this direction, the assumption that \Imc is an \Lmc-forest model of \Amc can be dropped.
    Let \Imc be a model of \Amc and assume there exists some \Lmc-forest \Amc-variation $p'$ of $p$ such that $\Imc \models p'$, witnessed by a homomorphism $g$. Then by definition of variations, there is a canonical homomorphism  $h$ from $p$ to $p'$ such that $h(a)=a$ for all individuals $a$ occurring in $p$. Since homomorphisms are closed under composition, the mapping $h' = g \circ h$ is a homomorphism from $p$ to \Imc that satisfies $h'(a)=a$ for all individuals $a$ in $p$, as required.

    \medskip 
    For Point~2, let \Imc be a model of \Amc, $p'$ be an \Lmc-forest \Amc-variation of $p$, and $(p_1,t_1), \dots, (p_n, t_n)$ be the \Lmc-tree CQs that arise from the connected components of $p' \setminus \Amc$.
    By construction, whenever an individual $a$ occurs in $p_i$, we have $t_i=a$ and $\mn{ind}(p_i)=\{a\}$. 
    Hence, an application of Lemma~\ref{lem:treeconcept} yields $\Imc \models p_i$ if and only if $\Imc \models \widehat p_i$, for every $1 \leq i \leq n$, where $\widehat p_i$ is the reduction query corresponding to $(p_i, t_i)$. Moreover, since 
    $p' \setminus ( p_1 \cup \dots \cup p_n) \subseteq \Amc$ and $\Imc \models \Amc$, it follows that
    $\Imc \models p' $ if and only if $\Imc \models p_i$ for all $1 \leq i \leq n$.
    Combining these equivalences yields the desired conclusion.
\end{proof}
It follows from the proof, that the `if'-direction in Point~1 holds even for arbitrary models of \Amc. We will use this fact in the proof of the following theorem.

\charUCQwithFin*

To prove the theorem, we need a way to handle the non-rooted parts of the queries in Condition~(b) uniformly when constructing a fitting ontology from a finite interpretation that satisfies Conditions~(a) and~(b). We achieve this by introducing a fresh role $u$ that simulates a universal role within each component of the disjoint union from Condition~(a). The following lemma shows that, with this gadget, existential \Lmc-queries can be treated as a union of rooted \Lmc-queries.

\begin{lemma}\label{lem:univrole}
    Let $\Lmc \in \{\EL, \ELbot, \ELI, \ELIbot\}$ and $u \in \NR$. Let $(\Jmc_i)_{i\in I}$ be a non-empty family of interpretations such that $u^{\Jmc_i}= \Delta^{\Jmc_i} \times \Delta^{\Jmc_i}$ for every $i \in I$, and set $\Imc = \biguplus_{i \in I} \Jmc_i$. For every ABox \Amc, every \Lmc-simulation $S$ from \Amc to \Imc, and every \Lmc-concept $C$, in none of which $u$ occurs, the following holds:
    \[\Imc_{\Amc, S, \Lmc} \models \exists x\, C(x) \quad \text{ if and only if }\quad\Imc_{\Amc, S, \Lmc} \models \exists u. C(a)\]
    for some $a \in \mn{ind}(\Amc)$.
\end{lemma}
\noindent
\begin{proof} 
    Let \Amc be an ABox, $S$ an \Lmc-simulation from \Amc to \Imc, and $C$ an \Lmc-concept such that $u$ does not occur in $C$ and \Amc.
    The `if'-direction is trivial. 
    To prove the other direction we proceed with the contraposition. Suppose $\Imc_{\Amc, S, \Lmc} \not \models \exists u. C(a)$ for all $a \in \mn{ind}(\Amc)$. Then, Corollary~\ref{lem:IASLsim} yields an element $d \in aS$ for which $d \not \in (\exists u. C)^{\Imc}$. 
    Let $i \in I$ such that $d \in \Delta^{\Jmc_i}$. Since $u$ acts as a universal role in $\Imc_i$, we obtain $C^{\Jmc_i} = \emptyset$. Hence, for every finite sequence $(r_j)_{1 \leq j \leq n}$ with $r_j \in \mn{N}_{\mn{R}}^\Lmc$, we have $d \not \in (\exists r_1. \dots \exists r_n C)^\Imc$. A second application of Corollary~\ref{lem:IASLsim} gives $\Imc_{\Amc, S, \Lmc} \not \models \exists r_1. \dots \exists r_n C(a)$ for every individual $a\in \mn{ind}(\Amc)$. Since every element of $\Imc_{\Amc, S,\Lmc}$ is reachable from an individual of \Amc by some path, that does not contain inverse roles if $\Lmc \in \{\EL, \ELbot\}$, this proves $C^{\Imc_{\Amc, S, \Lmc}}= \emptyset$.
\end{proof}

We are now ready to prove the characterization.

\noindent
\begin{proof} \textbf{(of Theorem~\ref{thm:charUCQwithFin})}
    `$1~\Rightarrow~2$'. 
    We prove the contraposition. Assume that Point~2 does not hold. We want to show that there is no fitting \Lmc-ontology. It suffices to prove that any \Lmc-ontology \Omc fitting the negative examples does not fit the positive examples. Hence, let \Omc be a \Lmc-ontology that fits the negative examples. 
	 Then, for each $ e = \left( \Amc, q \right) \in E^- $, there exists a model $ \Imc_e $ 
	 of \Amc witnessing that $ \Amc \cup \Omc \not\models q $.  Define the interpretation
	 $ \Imc := \biguplus_{e \in E^-}  \Imc_e$, which is non-empty since $ E^- $ is non-empty.
	 W.l.o.g., we may assume that each $\Imc_e$, and thus also \Imc, is finite by employing the finite model property for UCQ entailment in \Lmc. See, for instance, Corollary~1 in \cite{DBLP:conf/kr/GogaczIM18}.  
	 Since Point~2 is violated by assumption, there exists $ \left( \Amc, q \right) \in E^+  $ and an \Lmc-simulation $S$
	from \Amc to \Imc such that  $ \Imc_{\Amc, S, \Lmc} \not \models q$. 
	By Lemma \ref{lem:IASLsubI}, the interpretation $ \Imc_{\Amc, S, \Lmc}$ is a model of $\Amc \cup \Omc$. Consequently, $\Amc \cup \Omc \not \models q$, proving that \Omc does not fit the positive examples.

    \medskip
    `2$~\Rightarrow~$1'
    We prove this direction for a slightly weaker variant of (a):
    \begin{enumerate}[label=(\alph*$'$)]
        \item $ \Imc = \biguplus_{e \in E^-}  \Imc_e$, where for each 
					 $ e=( \Amc, q ) \in E^- $, $ \Imc_e $ is a model of \Amc such that $ \Imc_e \not \models p' $ for every \Lmc-forest \Amc-variation $p'$ of each  CQ in $q$.  
    \end{enumerate}
    Clearly, if a finite interpretation satisfies Conditions~(a) and~(b), then it also satisfies 
    Conditions~(a$'$) and~(b), by the remark after Lemma~\ref{prop:qEntIffVarEnt}. 
    We will show that the latter pair of properties is enough to prove Point~1.
    
    Suppose that \Imc is a finite interpretation satisfying Conditions~(a$'$) and~(b). Moreover, let $|\Delta^\Imc| = k$ and $\Imc = \biguplus_{e \in E^-} \Imc_e$ be the disjoint union from Condition~(a$'$). To deal with the non-rooted parts of queries appearing in $E$, we introduce a fresh role name $u$. Since $u$ does not appear in $E$, w.l.o.g., we may assume that $u$ acts as an universal role on $\Imc_e$, that is, $u^{\Imc_e}= \Delta^{\Imc_e} \times \Delta^{\Imc_e}$, for each $e \in E^-$. This assumption allows us to use Lemma~\ref{lem:univrole} later on.

	We proceed with the construction of the fitting \Lmc-ontology \Omc virtually in the same way as in the proof of Theorem~\ref{thm:charaq}.
	First, for every $ e=( \Amc, q ) \in E^+ $, we define the ontology $\Omc_e$, distinguishing the following cases:
	\begin{enumerate}
		\item There exists no \Lmc-simulation from \Amc to \Imc. Observe that then $ \Lmc=\Lmc'_{\bot} \in \{\ELbot, \ELIbot\}$  and there exists an individual 
			$ a \in \mn{ind}( \Amc )  $ such that 
		$ ( \Amc, a ) \not \preceq_{\Lmc'} ( \Imc, d ) $, 
			for all $ d \in \Delta^\Imc $. 
            To guarantee that the constructed ontology fits this example, we must ensure that $\Amc $ is inconsistent with \Omc. 
            Hence, we set \[
            \Omc_e := \left\{C^{\ell, \Lmc}_{\Amc, a} \sqsubseteq \bot \right\}  
			,\]	
            where $\ell = |\Amc| \cdot k$.
		\item Otherwise, let $S$ be the maximal \Lmc-simulation from \Amc to \Imc.
        By assumption and Lemma~\ref{prop:qEntIffVarEnt}, there exists some rooted CQ $ p $ 
	in  $ q $ and an  \Lmc-forest \Amc-variation $ p' $ of  $ p $ such that  
    $
         \Imc_{\Amc, S, \Lmc} \models p' 
     $.
      Since $S$ is chosen to be maximal, it suffices to ensure that $\Amc \cup \Omc$ entails the reduction queries of $p'$. 
      For each existential reduction query $\widehat p=\exists x\, D(x) \in \mn{red}(p')$, let $a_{\widehat p}$ be an individual such that $\Imc_{\Amc, S, \Lmc} \models \exists u. D (a_{\widehat p})$. The existence of such an individual is ensured by Lemma~\ref{lem:univrole}.
      We then set 
       \begin{align*}
            \Omc_e := &\left\{C^{\ell, \Lmc}_{\Amc, a_{\hat{p}}} \sqsubseteq \exists u. D (a_{\widehat p}) {\;\bigl \vert\;} \widehat p = \exists x \,D(x)\in \mn{red}(p')\right\}  \\
        &\cup\left\{C^{\ell, \Lmc}_{\Amc, a} \sqsubseteq D {\;\bigl \vert\;}  D(a)\in \mn{red}(p')\right\},
		\end{align*}
        where $\ell = |\Amc| \cdot k$.
	\end{enumerate}
	We now show that $\Omc:=\bigcup_{e \in E^+} \Omc_e$ fits $ E $. For this, we use the following claim.
    \\[2mm]
    {\bf Claim.} \Imc is a model of \Omc.
    \\[1mm]
    \emph{Proof of claim.} 
    It is enough to show that \Imc is a model of $\Omc_e$ for every $e \in E^+$. Let $e \in E^+$ and $ C\sqsubseteq D \in \Omc_e $. By construction of $\Omc_e$, we must have $C=C_{\Amc, a}^{\ell,\Lmc}$, for $\ell= k \cdot |\Amc|$ and some $a \in \mn{ind}(\Amc)$. Using Lemma \ref{lem:charconc} and then Lemma~\ref{lem:finmodsim}, we first obtain $ ( \Amc, a ) \preceq_{\Lmc'}^{\ell} ( \Imc, d )$ and subsequently $( \Amc, a ) \preceq_{\Lmc'} ( \Imc, d )$ for all $d \in C^{\Imc}$, where $ \Lmc'$ is \Lmc without $\bot$. Thus, $C^\Imc \subseteq  aS$ for the maximal $\Lmc'$-simulation $S$ from $\Amc$ to \Imc.
    If $ D = \bot $, then $aS = \emptyset$, by construction of $\Omc_e$, and we obtain $ C^{\Imc} \subseteq aS = \emptyset = D^{\Imc} $, as required.
    Otherwise, $S$ is also the maximal \Lmc-simulation from \Amc to \Imc, and we have $ \Imc_{\Amc, S, \Lmc} \models D(a) $, by construction of $\Omc_e$. An application of Corollary~\ref{lem:IASLsim} yields $ aS \subseteq  D^{\Imc} $ and thus  $C^{\Imc} \subseteq aS\subseteq D^{\Imc}$, proving $\Imc \models \Omc_e$ and therefore the claim.

    \medskip
    To show that \Omc fits the negative examples, let $e=(\Amc, q) \in E^-$. Recall that $\Imc_e$ is a finite model of \Amc such that $\Imc = \Imc_e \uplus \Jmc $, for some finite interpretation \Jmc, and $ \Imc_e \not \models p' $ for all \Lmc-forest \Amc-variations  $p'$ of any CQ in $q$. 
    Since \Lmc-ontologies are invariant under taking components of disjoint unions, the claim and $\Imc = \Imc_e \uplus \Jmc$ implies $\Imc_e \models \Omc$.
    Now, let $S_{id}$ denote the identity on $\mn{ind}(\Amc)$. Clearly, $S_{id}$ is an $\Lmc_\bot$-simulation from \Amc to $\Imc_e$. 
     Thus, $(\Imc_e)_{ \Amc, S_{id}, \Lmc }$ is an \Lmc-forest model of $\Amc $ and, by Lemma~\ref{lem:IASLsubI}, also model of $\Omc$. 
    In order to show $(\Imc_e)_{\Amc,S_{id}, \Lmc } \not \models q$, according to Lemma~\ref{prop:qEntIffVarEnt}, it suffices to prove $(\Imc_e)_{\Amc,S_{id}, \Lmc } \not \models p'$ for every \Lmc-forest \Amc-variation $p'$ of each CQ in $q$. Let $p'$ be such a variation.
    Since $\Imc_e \not\models p'$, Lemma~\ref{prop:qEntIffVarEnt} yields a reduction query $\widehat p$ of $p'$ such that $\Imc_e \not\models \widehat p$.
    If $\widehat p$ is a rooted \Lmc-query, using Corollary~\ref{lem:IASLsim}, we immediately obtain $(\Imc_e)_{\Amc,S_{id}, \Lmc } \not \models \widehat p$.
    Otherwise, $\widehat p$ is an existential \Lmc-query, say $\exists x\, C(x)$. Since $S_{id}$ is an $\Lmc_\bot$-simulation and by assumption $\Imc_e \models C \sqsubseteq \bot$, we may use Lemma~\ref{lem:IASLsubI} to derive $(\Imc_e)_{\Amc, S_{id},\Lmc } \models C \sqsubseteq \bot$ from which  $(\Imc_e)_{\Amc,S_{id}, \Lmc } \not \models \widehat p$ follows. Employing Lemma~\ref{prop:qEntIffVarEnt} one last time yields $(\Imc_e)_{\Amc,S_{id}, \Lmc } \not \models p'$, proving that
   $(\Imc_e)_{\Amc,S_{id}, \Lmc }$ witnesses $\Amc \cup \Omc \not \models q$. 

    It remains to show that \Omc fits the positive examples. Let $ e=( \Amc, q ) \in E^+ $. Clearly, every model \Jmc of \Amc must have $ a \in (C_{\Amc, a}^{\ell,\Lmc})^\Jmc $ for all $a \in \mn{ind}(\Amc)$ and $\ell =k\cdot |\Amc|$. If \Jmc is also a model of \Omc, then $\Jmc \models D(a)$ for all $C_{\Amc, a}^{\ell,\Lmc} \sqsubseteq D \in \Omc_e$, by construction. 
    If $D =\bot$ for any such $D$, this means $\Amc \cup \Omc$ is inconsistent, proving $\Amc \cup \Omc \models q$. Otherwise, let \Jmc be a model of $\Amc \cup \Omc$, and let $p'$ be the \Lmc-forest \Amc-variation of some CQ in $q$, from which $\Omc_e$ is constructed.
    An inspection of all CIs in $\Omc_e$ reveals that $\Jmc \models D(a)$ for all $D(a) \in \mn{red}(p')$ and for all $ \exists x\, D(x) \in \mn{red}(p')$ there exists some $b \in \mn{ind}(\Amc)$ such that $\Jmc \models \exists u. D(b)$. Since the latter, in particular, implies $\Jmc \models\exists x\, D(x)$, \Imc entails all reduction queries of $p'$.
    A final application of Lemma~\ref{prop:qEntIffVarEnt} yields $\Jmc \models p'$ and thus $\Jmc \models q$.
\end{proof}

The proof yields the following alternative characterization.
\begin{lemma}\label{cor:charUCQfinaalt}
 Let $\Lmc \in \{\EL,\ELbot,\ELI,\ELIbot\}$, and $ E= \left( E^+, E^- \right)  $ be a collection of labeled ABox-UCQ examples with 
	$ E^- \neq \emptyset  $. Conditions~(a) and~(b) of Theorem~\ref{thm:charUCQwithFin}
  can be replaced by:
	\begin{enumerate}[label=(\alph*$'$)]
		\item[(a$'$)] $ \Imc = \biguplus_{e \in E^-}  \Imc_e$ where, for each 
		$ e=( \Amc, q ) \in E^- $, $ \Imc_e $ is a model of \Amc such that $ \Imc_e \not \models p' $, for every \Lmc-forest \Amc-variation $p'$ of each CQ in $q$;
        
        \item[(b$'$)] there exists a role name $u$ so that for all $(\Amc, q) \in E^+$ there exists an \Lmc-forest \Amc-variation $p'$ of a CQ in $q$ and an individual $a_{\widehat p} \in \mn{ind}(\Amc)$ for every $\widehat p \in \mn{red}(p')$ such that the following holds:
        for every \Lmc-simulation $S$ from \Amc to \Imc,  
        \[\Imc_{\Amc, S, \Lmc}  \models C(a)  \quad \text{ and } \quad \Imc_{\Amc, S, \Lmc}  \models \exists u. C(a_{\widehat p}), \]
        for all $C(a) \in \mn{red}(p')$ and $\widehat p=\exists x\, C(x) \in \mn{red}(p')$.
	\end{enumerate}
\end{lemma}
\subsection*{Variants and Analysis of Fitting Ontologies}
\label{subsect:DiscFitOnt}
We now discuss the fitting ontologies that arise from the characterization. Let $\Lmc \in \{\EL, \ELbot, \ELI, \ELIbot\}$.
Recall that a fitting \Lmc-ontology \Omc for a collection $E=(E^+, E^-)$ of ABox-(U)CQ examples, as obtained in the characterization, captures the positive examples viewed as implications $\Amc \rightarrow q$. In fact, \Omc even expresses the stronger implication $\Amc \rightarrow p'$ for some \Lmc-forest \Amc-variation $p'$ of some CQ in $q$. These implications are encoded using the different reduction queries of $p'$, a fresh role $u$, and the characteristic concepts of \Amc. This extends the construction of the fitting ontologies from Theorem~\ref{thm:charaq} and, likewise, does not use any auxiliary concept names. However, the introduction of an auxiliary role $u$ is in general unavoidable whenever $E^+$ contains non-rooted UCQs. The following example illustrates this.
\begin{example}
    Let $E = (E^+, E^-)$, where $E^+:=\{(A(a), \exists x\, B(x))\}$ and $E^-:=\{(A(a), B(a))\}$. Clearly, the \EL-ontology $\Omc := \{A \sqsubseteq \exists u. B\}$ fits $E$. However, any \EL-ontology $\Omc'$ that does not contain auxiliary role names cannot contain any existential quantifiers, since no role name occurs in $E$. Thus, if $\Omc'$ fits $E^+$, we must have $\Omc' \models A \sqsubseteq B$. But then $\{A(a)\} \cup \Omc' \models B(a)$, showing that $\Omc'$ cannot fit $E^-$. This proves that every \EL-ontology that fits $E$ must contain an auxiliary role.
\end{example}
If, on the other hand, $E^+$ contains only rooted CQs, then a straightforward inspection of the construction in the proof of Theorem~\ref{thm:charUCQwithFin} shows that the resulting fitting ontology contains no auxiliary symbols.

Another difference to the formerly constructed fitting ontologies lies the depth of the used characteristic concepts, which is not determined directly by the examples themselves, but is linear in the size of $|\Delta^\Imc|\cdot ||E^+||$, where \Imc is the finite interpretation satisfying Conditions~(a) and~(b) from which \Omc is constructed. Indeed, while \Omc only contains polynomial many CIs, the depth of the used characteristic concepts results in an exponential blowup. We may thus use the same trick as in Section~\ref{sect:consistency} to reduce \Omc to polynomial size. 
Introduce the auxiliary concept name $X_{\Amc,a}^k$ for every $\Amc \in E^-$, $a \in \mn{ind}(\Amc)$, $\ell = |\Delta^\Imc|\cdot|\Amc|$ and $k \in \{0, \dots, \ell\}$. Then the ontology $\Omc'$ contains the following CIs for each $C_{\Amc, a}^{\ell,\Lmc} \sqsubseteq D \in \Omc$ and $1 \leq k \leq \ell$:
\[\begin{array}{rcl}
\displaystyle
	 \bigsqcap_{A(a) \in \Amc} \hspace{-0.1cm} A 
		 \sqcap \hspace{-0.1cm}\bigsqcap_{r\left( a,b \right) \in \Amc,  r \in \mn{N}_{\mn{R}}^\Lmc} \hspace{-0.1cm} \exists r.\, X_{\Amc, b}^{k-1} 
	 &\sqsubseteq& X_{\Amc, a}^k \\[5mm]
   \displaystyle \bigsqcap_{A(a) \in \Amc}A \ \sqsubseteq \ X_{\Amc, a}^0\qquad
 X_{\Amc, a_\Amc}^\ell &\sqsubseteq& D,
\end{array}\]
Since $\Omc'$ is merely a succinct representation of \Omc, it also fits $E$. Consequently, whenever $E$ admits a fitting \Lmc-ontology, it also admits a fitting ontology of size polynomial in $|\Delta^\Imc|\cdot||E^+||$. 

Note that this construction also applies to $(\Lmc,\mathrm{AQ})$-ontology fitting.
In this setting, since $\Omc$ is constructed from a completion $\Cmc$ satisfying
Conditions~(a) and~(b) of Theorem~\ref{thm:charaq}, we use $\Cmc$ in place of \Imc. Thus, the bound $|\mn{ind}(\Cmc)| \leq ||E^-||$ implies that the succinct fitting ontology
$\Omc'$ is of size polynomial in $||E||$.

We now present an alternative fitting ontology, generalizing the construction from Section~\ref{sect:consistency}, by encoding a finite interpretation rather than an ABox.
Let $E=(E^+, E^-)$ be a collection of labeled ABox-UCQ examples with $E^- \neq \emptyset$ and assume that \Imc is a finite interpretation that satisfies Conditions~(a) and~(b) of Theorem~\ref{thm:charUCQwithFin}. 
W.l.o.g., we assume that \Imc interprets only symbols used in the corresponding fitting ontology $\Omc$, and denote by $\mn{sig}(\Imc)$ the set of roles and concept names interpreted in $\Imc$ as non-empty. Correspondingly, we write $\mn{sig}(E^+)$ for the set of concept and role names occurring in $E^+$.
For each $d\in \Delta^\Imc$, let $\overline{V_d}$ be a fresh auxiliary concept name. Then, the alternative fitting ontology $\Omc_{\mn{alt}}^\Imc$ contains the following CIs:
\begin{enumerate}[label=(\roman*)]
    \item $\bigsqcap_{d \in \Delta^\Imc}\overline{V_d} 
	 \sqsubseteq \bot $, if $\Lmc \in \{\ELbot, \ELIbot\}$;
     \item $A \sqsubseteq \overline{V_d}  $ for every $d \in \Delta^\Imc \setminus A^\Imc$ and $ A\in \mn{sig}(E^+)$;
     \item $\exists r. \bigsqcap_{(d,e) \in r^\Imc} \overline{V_e} \sqsubseteq
			\overline{V_d}$ for every $d \in \Delta^\Imc$ and every $r \in \mn{N}_{\mn{R}}^\Lmc \cap\mn{sig}(E^+)$;
     \item $\bigsqcap_{d \in \Delta^\Imc \setminus A^\Imc}\overline{V_d}  \sqsubseteq A$ for every $A \in \mn{sig}(\Imc) \cup \mn{sig}(E^+)$;
     \item $\bigsqcap_{d \in P_F^r}
				\overline{V_d} \sqsubseteq \exists r. \bigsqcap_{e 
				\in F}
			\overline{V_e} $ for all  $r\in \mn{N}_{\mn{R}}^\Lmc \cap (\mn{sig}(\Imc) \cup \mn{sig}(E^+))$ and all $F \subseteq 
			\Delta^\Imc$, where $P_F^r:= \{d \in \Delta^\Imc\mid \forall e\in \Delta^\Imc,\ (d,e) \in r^\Imc \text{ implies } e \in F\} $.
\end{enumerate}

Intuitively, the CIs in $\Omc_{\mn{alt}}^\Imc$ of the form (i)-(iii) ensure that any model \Jmc of $\Omc_{\mn{alt}}^\Imc$ admits an \Lmc-simulation (of restricted signature) to \Imc, where (i) corresponds to the (total) Condition, (ii) to the (atom) Condition, and (iii) to the (Rel) and, possibly, (iRel) Condition. The remaining types of CIs, (iv) and (v), ensure that \Jmc cannot be extended in any non-trivial way without violating this simulation property. 

Note that, in contrast to the fitting ontologies constructed from the positive examples, this ontology is of exponential size in $||\Imc||$, due to the number of CIs of the form (v), but is independent of the size and number of positive examples. 

\begin{lemma} \label{lem:altontUCQ}
    Let $\Lmc \in \{\EL, \ELbot, \ELI, \ELIbot\}$, let $E = (E^+, E^-)$ be a collection of labeled ABox-UCQ examples with $E^- \neq \emptyset$ that admits a fitting \Lmc-ontology, and let \Imc be a finite interpretation satisfying Conditions~(a) and~(b) of Theorem~\ref{thm:charUCQwithFin}. Then the ontology $\Omc_{\mn{alt}}^\Imc$ fits $E$.
\end{lemma}
\noindent
\begin{proof}
    Let $\Lmc$, $E$, \Imc and $\Omc_{\mn{alt}}^\Imc$ be as in the lemma. Moreover, in the following we write $\Lmc'$ to denote \Lmc without $\bot$. 
    We first show that $\Omc_{\mn{alt}}^\Imc$ fits the negative examples. 
    Let $\Imc_*$ be the expansion of \Imc obtained by interpreting each $\overline{V_d}$ for $d \in \Delta^\Imc$ as
    \[\overline{V_d}^{\Imc_*} := \Delta^{\Imc} \setminus \{d\}.\] 
    Since this extension only concerns fresh concept names, $\Imc_*$ still satisfies Conditions~(a) and~(b) of Theorem~\ref{thm:charUCQwithFin}.  
    \\[2mm]
    {\bf Claim.} $\Imc_*$ is a model of $\Omc_{\mn{alt}}^\Imc$.
    \\[1mm]
    \emph{Proof of claim.} 
    By construction, $\Imc_*$ satisfies the CI of type~(i).
    Moreover, observe that for all $D \subseteq \Delta^{\Imc_*}$, we have \[
    \bigcap_{d \in D} \overline{V_d}^{\Imc_*} = \Delta^{\Imc_*} \setminus D.
    \]
    In particular, this implies $(\bigsqcap_{d \in \Delta^\Imc \setminus A^\Imc} \overline{V_d})^{\Imc_*} = A^{\Imc_*}$ for every $A \in \NC$, and thus $\Imc_*$ satisfies all CIs of the form~(ii) and~(iv).
    Now consider the CIs of type~(iii) and let $d \in \Delta^{\Imc_*}$ and $r$ be a role (possibly inverse).
    Using the equality above, we obtain 
    \[
    \bigcap_{(d,e) \in r^\Imc} \overline{V_e}^{\Imc_*} = \left\{e \in \Delta^{\Imc_*} \mid (d,e)\notin r^{\Imc}\right\}.
    \]
    Thus, any element $d' \in (\exists r. \bigsqcap_{(d, e) \in r^{\Imc}} \overline{V_e})^{\Imc_*}$ has an $r$-successor $e'$ in $\Imc_*$ with $(d,e') \not\in r^{\Imc}$. Since this implies $d' \neq d$, we conclude $d' 
    \in \overline{V_d}{\vphantom{V}}^{\Imc_*}$, proving that $\Imc_*$ satisfies all CIs of the form~(iii).
    Finally, to show that ${\Imc_*}$ satisfies the CIs of type (v), let $F \subseteq \Delta ^\Imc$, $r$ be a role, and $P_F^r:=\{d \in \Delta^\Imc\mid \forall e \in \Delta^\Imc, \ (d,e) \in r^\Imc \rightarrow \ e \in F\}$. 
    Then
    \[\bigcap_{d \in P_F^r} \overline{V_d}^{\Imc_*} = \Delta^{\Imc_*} \setminus P_F^r = \left\{d \mid \exists (d,e)\in r^{\Imc_*} \text{ s.t. } e \notin F  \right\}.\]
    Hence, if $d' \in (\bigsqcap_{d \in P_F^r} \overline{V_d})^{\Imc_*}$, then there exists some $(d',e') \in r^{\Imc_*}$ such that $e' \notin  F$. Since $\Delta^{\Imc_*} \setminus F = \bigcap_{e \in F} \overline{V_e}{\vphantom{V}}^{\Imc_*}$, we obtain $e' \in \overline{V_e}{\vphantom{V}}^{\Imc_*}$ for all $e \in F$. Consequently, $d' \in (\exists r.\bigsqcap_{e \in F} \overline{V_e})^{\Imc_*} $, as required. Thus, $\Imc_*$ satisfies all CIs in $\Omc^{\Imc}_{\mn{alt}}$, which proves the claim.
    \medskip

    Since $\Imc_*$ satisfies Condition~(a) of Theorem~\ref{thm:charUCQwithFin}, we may consider the disjoint union $\Imc_* = \biguplus_{e \in E^-} \Imc_{*,e}$, where $\Imc_{*,e}$ is a model of \Amc and $\Imc_{*,e} \not \models q$ for all $e = (\Amc ,q)\in E^- $. 
    Since $\Lmc$-ontologies are invariant under taking components of disjoint unions, $\Imc_{*,e}$ is also a model of $\Omc_{\mn{alt}}^\Imc$ and thus a witness of $\Amc \cup \Omc_{\mn{alt}}^\Imc \not \models q$. This finishes the proof that $\Omc_{\mn{alt}}^\Imc$ fits the negative examples.

    It remains to show that $\Omc_{\mn{alt}}^\Imc$ fits the positive examples.
    To this end, we make the semantic consequences of the CIs of type (iv) and (v) explicit in the following claim. 
    \\[2mm]
    {\bf Claim.} Let \Jmc be a model of $\Omc_{\mn{alt}}^\Imc$, $ T= (\Imc, t_i)_{i\in I}$ be a sequence of pointed copies of \Imc, and $B \subseteq \Delta^{\Imc}$ with $\Delta^\Imc \setminus B \subseteq \{t_i\}_{i \in I} $. Then for every $d \in (\bigsqcap_{b \in B} \overline{V_b})^\Jmc$,  we have \[\left(\prod T, t_\Pi\right) \preceq_{\Lmc'} (\Jmc, d).\]
    \emph{Proof of claim.} 
    Let $\Jmc$ and $T$ be as in the claim. 
    Throughout the proof, for every $\bar e \in \Delta^{\prod T}$, we write $\bar e = (e_i)_{i \in I}$. Define the relation $S \subseteq \Delta^{\prod T} \times \Delta^\Jmc$ as
    \[S= \left\{ (\overline e, d) {\;\biggl \vert\;}\substack{\exists B \subseteq \Delta^\Imc \text{ such that }  \Delta^\Imc \setminus B \subseteq \{e_i\}_{i \in I}\\ \text{ and } d \in (\bigsqcap_{e \in B} \overline{V_e})^\Jmc }\right\}.\]
    Observe that by definition, $(t_\Pi , d) \in S$ whenever $d \in (\bigsqcap_{e \in B} \overline{V_e})^\Jmc$ for some $B \subseteq \Delta^\Imc$ with $\Delta^\Imc \setminus B \subseteq \{t_i\}_{i \in I}$. 
    It is thus left to show that $S$ is an $\Lmc'$-simulation from $\prod T$ to $\Jmc$. We first consider (atom). Let $\bar e \in A^{\prod T}$ and $(\bar e, d ) \in S$. 
    By definition of $S$, there exists $ B \subseteq \Delta^\Imc $ such that $\Delta^\Imc \setminus B \subseteq \{e_i\}_{i \in I}$ and $d \in \overline{V_c}{\vphantom{V}}^\Jmc$ for all $c \in B$. Moreover, since $\bar e \in A^{\prod T}$, $\{e_i\}_{i \in I} \subseteq A^\Imc$ follows from Lemma~\ref{lem:prodlem}. Combining this with the previous inclusion, we obtain $\Delta^\Imc \setminus B \subseteq A^\Imc$, and hence $\Delta^\Imc \setminus A^\Imc \subseteq B$. 
    This means that, by assumption, $d \in \overline{V_c}{\vphantom{V}}^\Jmc$ for all $c \in \Delta^\Imc \setminus A^\Imc$. 
    As $\Jmc$ is a model of $\Omc_{\mn{alt}}^\Imc$, it satisfies the CIs of the form (iv) and therefore $d \in A^\Jmc$, as required. 
    To prove that $S$ satisfies (Rel), let $r$ be a role name, $(\bar e, \bar e') \in r^{\prod T}$ and $(\bar e, d) \in S$. By definition of $S$, there exists $B \subseteq \Delta^\Imc$ such that $\Delta^\Imc \setminus B \subseteq \{e_i\}_{i \in I} $ and $d \in \overline{V_c}{\vphantom{V}}^\Jmc$ for all $c \in B$. 
    Let $F = \Delta^\Imc \setminus\{e'_i\}_{i \in I}$ and $P_F^r = \{p \in \Delta^\Imc \mid \forall f \in \Delta^\Imc,\ (p,f) \in r^\Imc \rightarrow f \in F\}$. Since $(\bar e, \bar e') \in r^{\prod T}$, we have $(e_i , e_i') \in r^\Imc$ for all $i \in I$ by the definition of direct products. Consequently, by the choice of $F$, we obtain $\{e_i\}_{i \in I} \cap P_F^r = \emptyset$ and thus $P_F^r \subseteq \Delta^\Imc \setminus \{e_i\}_{i \in I}$. Together with the inclusion $\Delta^\Imc \setminus B \subseteq \{e_i\}_{i \in I}$  from above, this yields $P_F^r \subseteq B$. Hence, we must have $d \in \overline{V_c}{\vphantom{V}}^\Jmc$ for all $c \in P_F^r$. 
    Since $r\in \mn{N}_{\mn{R}}^\Lmc \cap \mn{sig}(\Imc)$, we may apply the CI of type (v) corresponding to $F$ and $r$ to obtain $d \in (\exists r. \bigsqcap_{f \in F} \overline{V_f})^\Jmc$. Thus, there exists some $(d, d') \in r^\Jmc$ such that $d' \in \overline{V_f}{\vphantom{V}}^\Jmc$ for all $f \in F$. By definition of $F$ and $S$, we conclude $(\bar e', d') \in S$. The proof of (iRel) in the case of $\Lmc \in \{\ELI, \ELIbot\}$ is virtually identical and therefore omitted. Hence, $S$ is an $\Lmc'$-simulation from $\prod T$ to \Jmc, which finishes the proof of the claim. 
    \medskip

    Now let $(\Amc, q) \in E^+$. In order to show that $\Omc_{\mn{alt}}^\Imc$ fits $(\Amc, q)$, we need a second important observation: any model \Jmc of $\Amc \cup \Omc_{\mn{alt}}^\Imc$ gives rise to an \Lmc-simulation from \Amc to \Imc defined by 
    \[S = \left\{(a,e) \in \mn{ind}(\Amc) \times \Delta^\Imc \mid a \notin \overline{V_e}^\Jmc\right\}.\]
    That $S$ is indeed an \Lmc-simulation was already shown in the proof of Lemma~\ref{lem:consaltont} using only the CIs of the form~(i)-(iii).
    
    Now let \Jmc be a model of $\Amc \cup \Omc_{\mn{alt}}^\Imc$. 
    The former observation yields an \Lmc-simulation $S$ from \Amc to \Imc. Since $\Imc$ satisfies Condition~(b) of Theorem~\ref{thm:charUCQwithFin}, $\Imc_{\Amc, S, \Lmc} \models q$ and, by Lemma~\ref{prop:qEntIffVarEnt}, there exists an \Lmc-forest \Amc-variation $p'$ of some CQ in $q$ such that $\Imc_{\Amc, S, \Lmc} \models p'$. In the following, we show that $\Jmc \models \widehat p$ for each $\widehat p \in \mn{red}(p')$ from which $\Jmc \models q$ follows, by Lemma~\ref{prop:qEntIffVarEnt}. Let $C(a) \in \mn{red}(p')$. Since $\Imc_{\Amc, S, \Lmc} \models C(a)$, an application of Corollary~\ref{lem:IASLsim} yields $a S \subseteq C^\Imc$. By definition of $S$, we have $a \in (\bigsqcap_{d \notin aS } \overline{V_d})^\Jmc$. Employing the claim for $T = (\Imc, t)_{t \in aS}$ and $B = \Delta^\Imc \setminus aS$ shows that $(\prod T, t_\Pi ) \preceq_{\Lmc'} ( \Jmc, a)$. We may now apply Lemma~\ref{lem:prodlem} to derive $t_\Pi \in C^{\prod T}$, from which $a \in C^\Jmc$ follows by Lemma~\ref{lem:SimPresConc}. Thus, $\Jmc \models C(a)$, as required. 
    Now let $\exists x \, C(x) \in \mn{red}(p')$. Then, $\Imc_{\Amc, S, \Lmc} \models \exists r_1. \dots \exists r_n C(a)$ for some $a \in \mn{ind}(\Amc)$ and some finite sequence $(r_i)_{1\leq i \leq n}$ of $r_i \in \mn{N}_{\mn{R}}^\Lmc$, since $\Imc_{\Amc, S, \Lmc} \models \exists x \, C(x)$ and every element of $\Imc_{\Amc, S, \Lmc}$ is connected by an \Lmc-path to some individual in $\Amc$. The argument now proceeds as in the case $C(a) \in \mn{red}(p')$, resulting in $\Jmc \models \exists r_1. \dots \exists r_n C(a)$ and, consequently, $\Jmc \models \exists x \, C(x)$. This shows that $\Jmc \models q$ and therefore $\Amc \cup \Omc_{\mn{alt}}^\Imc \models q$, completing the proof that $\Omc_{\mn{alt}}^\Imc$ also fits the positive examples.
\end{proof}

In principle, this proof also applies to the alternative fitting ontologies for (\Lmc, AQ)-ontology fitting that are based on fitting completion rather than finite interpretations and, in particular, do not contain CIs of the form~(v). One only has to adjust the second claim by replacing $\Lmc'$ simulations with $0$-\Lmc-simulations and substitute interpretations with completions at the relevant points. As these changes are merely technical, details are omitted.

\section{Proofs for Section~\ref{sect:EL}}
\subsection*{Fundamental Properties of $\Lmc^{\mn{sim}}_{\bot}$}

\begin{lemma}\label{lem:LsiInvSim}
Let $\Lmc \in \{\EL, \ELI\}$. Then $\Lmc_\bot^{\mn{sim}}$-concepts are invariant under $\Lmc$-simulations. That is, for any two pointed interpretations $(\Jmc_1, d_1)$ and $(\Jmc_2, d_2)$ with $(\Jmc_1, d_1)\preceq_\Lmc (\Jmc_2, d_2)$ and every $\Lmc_\bot^{\mn{sim}}$-concept $C$, \[d_1 \in C^{\Jmc_1}\quad \text{ implies }\quad d_2 \in C^{\Jmc_2}.\]  
\end{lemma}
\noindent
\begin{proof}
    Let $\Lmc \in \{\EL, \ELI\}$. It suffices to consider the case where $C$ is of the form $\exists^{\mn{sim}}_\Lmc(\Imc,d)$.  The remaining cases follow by a straightforward induction over the concept constructors. Let $(\Jmc_1, d_1)$ and $(\Jmc_2,d_2)$ be pointed interpretations such that $(\Jmc_1, d_1) \preceq_{\Lmc} (\Jmc_2, d_2)$, and assume $d_1 \in (\exists^{\mn{sim}}_\Lmc (\Imc, d))^{\Jmc_1}$. Then there exists an $\Lmc$-simulation $S_1$ from $\Imc$ to $\Jmc_1$ and an $\Lmc$-simulation $S_2$ from $\Jmc_1$ to $\Jmc_2$ such that $(d, d_1) \in S_1$ and $(d_1, d_2) \in S_2$. Their composition $ S= S_2 \circ S_1$ is an $\Lmc$-simulation from $\Imc$ to $\Jmc_2$, and we have $(d,d_2) \in S$. Hence, $d_2 \in (\exists^{\mn{sim}}_\Lmc (\Imc, d))^{\Jmc_2} $, as required.
\end{proof}
As an easy consequence we obtain the following.
\begin{corollary}\label{cor:simdisju}
    Let $\Lmc \in \{\EL, \ELI\}$. Then $\Lmc_\bot^{\mn{sim}}$-ontologies are invariant under disjoint unions and under taking components of such unions. That is, for every $\Lmc_\bot^{\mn{sim}}$-ontology \Omc and all interpretations $\Imc_1$ and $\Imc_2$, \[\Imc_1 \uplus \Imc_2 \models \Omc\quad \text{ if and only if }\quad\Imc_1 \models \Omc \text{ and } \Imc_2 \models \Omc.\]
\end{corollary}
\noindent
\begin{proof}
    Let $\Lmc \in \{\EL, \ELI\}$, $\Imc_1$ and $\Imc_2$ be interpretations, and \Omc be an $\Lmc_\bot^{\mn{sim}}$-ontology. W.l.o.g., assume that $\Delta^{\Imc_1} \cap \Delta^{\Imc_2} = \emptyset$, and thus $\Imc_1 \uplus \Imc_2 = \Imc_1 \cup \Imc_2$. 
    First, we show that for every $\Lmc_\bot^{\mn{sim}}$-concept $C$, 
    \[ C^{\Imc_1 \uplus \Imc_2} = C^{\Imc_1} \cup C^{\Imc_2}.\]
    Since the identity relation $S_i = \mathrm{id}_{\Delta^{\Imc_i}}$ is a \Lmc-simulation from $\Imc_i$ to $\Imc_1 \uplus \Imc_2$, the inclusion $C^{\Imc_i} \subseteq C^{\Imc_1 \uplus \Imc_2}$ follows directly from Lemma~\ref{lem:LsiInvSim}. For the opposite inclusion, let $d \in C^{\Imc_1 \uplus \Imc_2}$. Observe that $d \in \Delta^{\Imc_i}$ for some $i \in \{1,2\}$. Since $S_i$ is also an \Lmc-simulation from $\Imc_1 \uplus \Imc_2$ to $\Imc_i$, another application of Lemma~\ref{lem:LsiInvSim} yields $d \in C^{\Imc_i}$. Thus, $C^{\Imc_1 \uplus \Imc_2} \subseteq C^{\Imc_1} \cup C^{\Imc_2}$.
    
    Now let $C \sqsubseteq D \in \Omc$. Using the equality above for $C$ and $D$, we obtain $C^{\Imc_1 \uplus \Imc_2} \subseteq D^{\Imc_1 \uplus \Imc_2}$ if and only if $C^{\Imc_i} \subseteq D^{\Imc_i}$ for $i \in \{1,2\}$. Consequently, $\Imc_1 \uplus \Imc_2 \models \Omc $ if and only if $\Imc_i \models \Omc$ for all $i \in \{1,2\}$. 
\end{proof}

\begin{corollary} \label{cor:LsiInvDirProd}
    Let $\Lmc \in \{\EL, \ELI\}$. Then $\Lmc_\bot^{\mn{sim}}$-concepts are invariant under direct products. That is, for every family $S=(\Imc_i,d_i)_{i \in I}$ of pointed interpretations and each $\Lmc_\bot^{\mn{sim}}$-concept~$C$, \[d_\Pi \in C^{\prod S}\quad \text{ if and only if } \quad
    d_i \in C^{\Imc_i} \text{ for all } i \in I.\]
\end{corollary}
\noindent
\begin{proof}
    Let $\Lmc \in \{\EL, \ELI\}$, $S=(\Imc_i, d_i)_{i \in I}$ be a family of pointed interpretations, and $C$ be an $\Lmc_\bot^{\mn{sim}}$-concept. 
    `\emph{only if}'. Assume that $d_\Pi \in C^{\prod S}$. By Lemma~\ref{lem:prodlem}, we have $(\prod S, d_{\Pi}) \preceq_\Lmc (\Imc_i, d_i)$ for all $i \in I$. An application of Lemma~\ref{lem:LsiInvSim} yields $d_i \in C^{\Imc_i}$ for all $i \in I$, as required.
    
    `\emph{if}'. Suppose $d_i \in C^{\Imc_i}$ for all $i \in I$. It suffices to treat the case where $C=\exists^{\mn{sim}}_{\Lmc}(\Jmc, e)$. The remaining cases follow by a straightforward induction over the concept constructors. By assumption and the semantics of the simulation quantifier, $(\Jmc, e) \preceq_\Lmc (\Imc_i, d_i)$ for all $i \in I$. Since $(\Jmc^{\downarrow \Lmc}_e , e) \to (\Jmc, e)$ by Point~(1) of Lemma~\ref{lem:unravbasic} and since every homomorphism is an \Lmc-simulation, we obtain $(\Jmc^{\downarrow \Lmc}_e , e) \preceq_\Lmc (\Imc_i, d_i) $ for all $i \in I$. Thus, $(\Jmc^{\downarrow \Lmc}_e , e) \to (\Imc_i, d_i)$ follows by Point~(4) of the same lemma. This allows us to use Lemma~\ref{lem:prodlem} and obtain $(\Jmc^{\downarrow \Lmc}_e , e) \to (\prod S, d_\Pi)$. Viewing this homomorphism as an \Lmc-simulation and composing it with $(\Jmc, e) \preceq_\Lmc (\Jmc^{\downarrow \Lmc}_e , e)$ from Lemma~\ref{lem:unravbasic} yields $(\Jmc, e) \preceq_\Lmc (\prod S, d_\Pi)$. By the semantics of the simulation quantifier, this proves $d_\Pi \in C^{\prod S}$.
\end{proof}

\subsection*{Upper bound for $\EL_{(\bot)}$}

\charVarSimQuant*

\noindent
\begin{proof}
    Let \Lmc and $E$ be as in the lemma. Moreover, let $u\in \NR$ be a fresh role name not occurring in $E$ and let $\Lmc'$ denote \Lmc without $\bot$. 
    For each $e=(\Amc,q) \in E^+$, we define the set $\Gamma_e$ of $\Lmc_{\bot}^{\prime\mn{sim}}$-ontologies as follows:
    \begin{itemize}
        \item for every \Lmc-forest \Amc-variation $p'$ of a CQ in $q$ and every way to choose an individual $a_{\widehat q}\in \mn{ind}(\Amc)$ for every existential reduction query $\widehat q$ of $p'$,  the ontology $\Omc$  that contains the following CIs:
        \begin{align*}
            \exists_\Lmc^{\mn{sim}}(\Amc, a) &\sqsubseteq C &\text{ for every }C(a) \in \mn{red}(p');\\
            \exists_\Lmc^{\mn{sim}}(\Amc, a_{\widehat q}) &\sqsubseteq \exists u. C &\hspace*{-0.1cm}\text{for every } \widehat q=\exists x\, C(x) \in \mn{red}(p').
        \end{align*}
        \item if $\Lmc \in \{ \EL_\bot,\ELI_\bot \}$, then 
     for every $a \in \mn{ind}(\Amc)$, the ontology \[\Omc_{a} = \{\exists_\Lmc^{\mn{sim}}(\Amc, a) \sqsubseteq \bot\}.
    \]
    \end{itemize}
    It is easy to see that every ontology in $\Gamma_e$ is of size linear in $||e||$. It remains to show that this choice is as required.
    \smallskip 
    
    `1~$\Rightarrow$~2' Assume that $E$ admits a fitting \Lmc-ontology. By  Corollary~\ref{cor:charUCQfinaalt}, there exists a finite interpretation \Imc that satisfies Conditions~(a$'$) and~(b$'$). Let $u'$ be the role name for which \Imc satisfies Condition~(b$'$) and define $\Jmc$ to be the (finite) interpretation obtained from $\Imc$ by setting $u^{\Jmc}:= u^{\prime\Imc}$. 
    Since $u$ does not occur in $E$, $\Jmc$ still satisfies Condition~(a$'$) of Corollary~\ref{cor:charUCQfinaalt} and therefore also Condition~(a) of this theorem.
    To prove 2, it therefore remains to show that $\Jmc$ satisfies Condition~(b). 
    To this end, let $e=(\Amc, q) \in E^+$. First, assume that there is no \Lmc-simulation from \Amc to \Jmc. This means we have $\Lmc \in \{\ELbot, \ELIbot\}$ and there exists an individual $a\in\mn{ind}(\Amc)$ such that $(\Amc, a) \not \preceq_{\Lmc'} (\Jmc, d)$ for all $d \in \Delta^\Jmc$. Pick the corresponding ontology $\Omc_a = \{\exists_{\Lmc'}^{\mn{sim}}(\Amc, a) \sqsubseteq \bot\} \in \Gamma_e$. By assumption, $(\exists_{\Lmc'}^{\mn{sim}}(\Amc, a))^\Jmc = \emptyset$. This shows that \Jmc is a model of $\Omc_a \in \Gamma_e$, as required.
    Now, suppose that there exists an \Lmc-simulation from \Amc to \Jmc. 
    Observe that by construction of \Jmc and a careful inspection of Definition~\ref{def:refcand}, we have that $\Jmc$ satisfies Condition~(b$'$) of Corollary~\ref{cor:charUCQfinaalt} for the role name $u$.
    Using this fact, we obtain an \Lmc-forest \Amc-variation $p'$ of some CQ in $q$ and an individual $a_{\widehat p} \in \mn{ind}(\Amc)$ for every $\widehat p \in \mn{red}(p')$ such that the conditions in (b$'$) are met. 
    Pick the ontology $\Omc \in \Gamma_e$  constructed in regard to $p'$ and and the choice of individuals wrt. $\mn{red}(p')$. Consider any $\exists^{\mn{sim}}_{\Lmc'}(\Amc, a) \sqsubseteq D \in \Omc$, and let $S$ be the maximal \Lmc-simulation from \Amc to \Jmc. Since $S$ is maximal, $aS = (\exists^{\mn{sim}}_{\Lmc'}(\Amc, a))^\Jmc$. 
    By Condition~(b$'$), $\Jmc_{\Amc, S,\Lmc} \models D(a)$ to which we apply Corollary~\ref{lem:IASLsim} and obtain $aS \subseteq D^\Jmc$. Thus, \Jmc is a model of $\exists^{\mn{sim}}_{\Lmc'}(\Amc, a) \sqsubseteq D$ and thereby $\Omc$. This proves that \Jmc satisfies Condition~(b), and we are done. 
    \smallskip

    `2~$\Rightarrow$~1' Assume that there exists a finite interpretation \Imc that satisfies Conditions~(a) and~(b) of the statement. We prove that $E$ admits a fitting \Lmc-ontology by showing that $\Imc$ satisfies the conditions in Corollary~\ref{cor:charUCQfinaalt}. Since $\Imc$ already satisfies Condition~(a$'$) of this Corollary by assumption, it is enough to show that \Imc satisfies Condition~(b$'$) wrt. the role name $u$. 
    Let $e = (\Amc, q)\in E^+$ and $\Omc \in \Gamma_e$ such that $\Imc \models \Omc$. 
    If there exists no \Lmc-simulation from \Amc to \Imc, there is nothing to show. Otherwise, let $S$ be an \Lmc-simulation from \Amc to \Imc.
    In the case $\Lmc \in \{\ELbot, \ELIbot\}$, this simulation is left-total and therefore $\emptyset \subsetneq aS \subseteq (\exists_{\Lmc'}^{\mn{sim}}(\Amc, a))^\Imc$ for all $a \in \mn{ind}(\Amc)$. Consequently, \Omc cannot be of the form $\Omc_a = \{\exists_\Lmc^{\mn{sim}}(\Amc, a) \sqsubseteq \bot\}$ for any $a \in \mn{ind}(\Amc)$, since \Imc is a model of \Omc.  
    Thus, \Omc is based on an \Lmc-forest \Amc-variation $p'$ of some CQ in $q$ and a 
    choice of individuals $a_{\widehat p} \in \mn{ind}(\Amc)$ for every $\widehat p \in \mn{red}(p')$. Let $D(a)$ be such that either $ D(a) \in \mn{red}(p')$ or $D(a) = \exists u. C(a_{\widehat p})$, where $\widehat p = \exists x\, C(x) \in \mn{red}(p')$ . By construction, $\exists_\Lmc^{\mn{sim}}(\Amc, a)  \sqsubseteq D \in \Omc$. Since \Imc is a model of \Omc, we obtain \[aS \subseteq (\exists_\Lmc^{\mn{sim}}(\Amc, a) )^\Imc \subseteq D^\Imc.\] 
    An application of Corollary~\ref{lem:IASLsim} yields $\Imc_{\Amc, S, \Lmc} \models D(a)$. Thus, \Imc satisfies Condition~(b$'$) of Corollary~\ref{cor:charUCQfinaalt}. 
\end{proof}

Since the algorithms for $(\EL_{(\bot)}, \text{UCQ})$ and $(\ELI_{(\bot)}, \text{UCQ})$-ontology fitting are virtually the same, except that the latter replaces the guessing of variations by an enumeration, we prove their correctness simultaneously. 

\elicorrectness*
\begin{lemma}
    The algorithm of Theorem~\ref{thm:ELIrootedUCQinExpTime} accepts its input $E$ if and only if there is an \Lmc-ontology that fits $E$.
\end{lemma}
\noindent
\begin{proof}
    Let $\Lmc \in \{\EL,\ELbot, \ELI, \ELIbot\}$, $E=(E^+, E^-)$ be a collection of ABox-UCQ examples with $E^- \neq \emptyset$ and $u \in \NR$ a fresh role.
    For each $e \in E^+$, let $\Gamma_e$ be the set of $\Lmc^{\prime\mn{sim}}_\bot$-ontologies from the proof of Theorem~\ref{lem:charVarSimQuant}, where $\Lmc'$ is \Lmc without $\bot$. 
    Recall that, at their core, the two algorithms presented above check the following condition:
    \begin{enumerate}[label=]
        \item[] For every $e \in E^+$ there exists an $\Omc_e \in \Gamma_e$ such that $\Amc \cup \bigcup_{e \in E^+} \Omc_e \not\models^{\mn{fin}} p'$ for every $(\Amc, q) \in E^-$ and \Lmc-forest \Amc-variation $p'$ of a CQ in $q$.
    \end{enumerate}
    To prove correctness of the algorithm, by Theorem~\ref{lem:charVarSimQuant}, it therefore suffices to show the following equivalence:
    \begin{enumerate}
        \item there exists a finite interpretation \Imc that satisfies Conditions~(a) and~(b) of Theorem~\ref{lem:charVarSimQuant};
        \item for every $e \in E^+$ there exists an $\Omc_e \in \Gamma_e$ such that $\Amc \cup \bigcup_{e \in E^+} \Omc_e \not\models^{\mn{fin}} p'$ for every $(\Amc, q) \in E^-$ and \Lmc-forest \Amc-variation $p'$ of a CQ in $q$.
    \end{enumerate}

    `1 $\Rightarrow$ 2' Assume there exists finite model \Imc that satisfies Conditions~(a) and~(b) of Theorem~\ref{lem:charVarSimQuant}. 
    Since \Imc satisfies Condition~(b), for every $e \in E^+$ there exists an $\Omc_e \in \Gamma_e$ such that $\Imc \models \bigcup_{e\in E^+}\Omc_e$. 
    Let $e = (\Amc,q) \in E^-$, and $\Imc_e$ be the model of \Amc from Condition~(a) that satisfies $ \Imc_e \not \models p' $ for every \Lmc-forest \Amc-variation $p'$ of any CQ in $q$. 
    Since $\Lmc_\bot^{\prime\mn{sim}}$ is invariant under taking components of disjoint unions by Corollary~\ref{cor:simdisju}, we obtain $\Imc_e \models \bigcup_{e\in E^+}\Omc_e$. Hence, $\Imc_e$ witnesses $\Amc \cup \bigcup_{e\in E^+}\Omc_e \not \models^{\mn{fin}} p'$ for every \Lmc-forest \Amc-variation $p'$ of any CQ in $q$.

    \medskip
    `2 $\Rightarrow$ 1' Assume that for each $e \in E^+$ there exists some $\Omc_e\in \Gamma_{e}$ for which Condition~2 holds. Let $\Omc:=\bigcup_{e\in E^+}\Omc_e $, let $e = (\Amc, q) \in E^-$, and let $FV_e$ denote the set of all \Lmc-forest \Amc-variations of all CQs in $q$. By assumption, for each $p' \in FV_e$ there exists a finite model $\Imc_{e,p'}$ of $\Amc \cup \Omc$ such that $\Imc_{e,p'} \not \models p'$. Let $M_e$ be the set of these interpretations and define $\Imc_e$ to be the interpretation obtained from $\prod M_e$
    by renaming $(a,a,...,a)$ to $a$ for every $a \in \mn{ind}(\Amc)$. It is an easy exercise to show that in the special case where $\Amc \cup \Omc$ admits a finite universal model $\Jmc_e$, we have $\Jmc_e \not \models p' $ for every $p' \in FV_e$. In this case, we may take $M_e = \{\Jmc_e\}$, making $\Imc_e$ coincide with $\Jmc_e$. 
    
    The following claim captures how $\Imc_e$ combines the relevant query entailment properties of the different interpretations $\Imc_{e,p'}$.
    \\[2mm]
    {\bf Claim.} $\Imc_e$ is a finite model of $\Amc \cup \Omc$ and $\Imc_e \not \models p'$ for every $p' \in FV_e$.
    \\[1mm]
    \emph{Proof of claim.} Since $\Imc_e$ is a finite product of finite models of \Omc and since $\Lmc_\bot^{\mn{sim}}$ is invariant under direct products, $\Imc_e$ is a finite model of $\Omc$. It is straightforward to verify that $\Imc_e$ is also a model of \Amc. To show the second part of the claim, let $p' \in FV_e$. 
    Assume for contradiction that $\Imc_e \models p' $, witnessed by a homomorphism $h$ from $p'$ to $\Imc_e$ such that $h(a)=a$ for all $a \in \mn{ind}(p')$. 
    By construction of $\Imc_e$ and Lemma~\ref{lem:prodlem}, $h$ extends to a homomorphism $h'$ from $p'$ to $\Imc_{e,p'}$ such that $h'(a)= a$ for all $a \in \mn{ind}(p')$. However, this contradicts the assumption that $\Imc_{e,p'} \not \models p'$. Hence, we must have $\Imc_e \not \models p'$. This proves the claim.

    \medskip
    We now show that the finite interpretation $\Imc = \biguplus_{e \in E^-} \Imc_e$ satisfies Conditions~(a) and~(b) of Theorem~\ref{lem:charVarSimQuant}. By the claim, \Imc already satisfies Condition~(a).  Moreover, since $\Lmc_\bot^{\mn{sim}}$ is invariant under disjoint unions, and $\Imc_{e} \models \Omc$, we obtain $\Imc \models \Omc$. Hence, \Imc also satisfies Condition~(b).
\end{proof}


Note that the proof also yields a bound on the size of the finite interpretations witnessing the existence of a fitting ontology, provided that the relevant witness models for $\Amc \cup \bigcup_{e \in E^-}\Omc_e \not\models p'$ are bounded in size by some $f(||E||)$. In general, this bound is exponential in $||E||\cdot \log (f(||E||))$, due to the possibly exponential number of \Lmc-forest \Amc-variations and the direct product involved in the construction.
However, in the presence of finite universal models, this blowup can be avoided, as observed in the proof, leading to a polynomial bound in $f(||E||)$. 
In the context of the discussion about the size of fitting ontologies in Section~\ref{sect:appendUCQ}, the preceeding bound therefore implies that, since $\ELbot^{\mn{sim}}$-knowledge bases always admit finite universal models of polynomial size \cite{lutz2010enriching}, if $E$ admits a fitting $\EL_{(\bot)}$-ontology, it also admits one of size polynomial in $||E||$.

\subsection*{Lower Bound for $\EL_{(\bot)}$}
\CQSigmaptwohard*
    We derive this lower bound via a reduction from the complement of the \mn{2}-\mn{COLORING}-\mn{EXTENSION} problem. 
    To state the problem formally, we introduce some basic notions.
    Let $G=(V,E)$ be an undirected graph and $P \subseteq V$ a set of vertices. 
    A $k$-coloring of $G$ is a function $c_k: V \to \{1, \dots, k\}$ such that $c_k(v) \neq c_k(w)$ for all $(v,w)\in E$. There is a well known one-to-one correspondence between the $k$-colorings of $G$ and homomorphisms from $G$ to $K_k$, where $K_k$  denotes the complete graph of $k$ vertices \cite{DBLP:books/daglib/0013017}. Let $G_{|P}$ denote the subgraph of $G$ induced by $P$. We say that a $k$-coloring $c_k$ of $G_{|P}$ is \emph{extendable} to a $(k+1)$-coloring of $G$ if there exists a $(k+1)$-coloring $c_{k+1}$ of $G$ such that $c_k(v)=c_{k+1}(v)$ for all $v \in P$.

    The \mn{2}-\mn{COLORING}-\mn{EXTENSION} problem for $(G,P)$ asks whether every 2-coloring of $G_{|P}$ is extendable to a 3-coloring of $G$. This problem is known to be $\Pi^P_2$-complete \cite{szeider2005generalizations}.
    
    Given an undirected graph $G=(V,E)$ and $P \subseteq V$, in the following we show how to construct a collection of ABox-rooted CQ-examples $E_{G,P}=(E^+, E^-)$ such that $E_{G,P}$ admits a fitting $\Lmc$-ontology iff there exists a 2-coloring of $G_{|P}$ that is not extendable to a 3-coloring of $G$.
    W.l.o.g., assume that $V=\{1, \ldots, n\}$ and $P=\{1, \ldots, k\}$. We define $E_{G,P}=(E^+,E^-)$ using the characterization of Theorem~\ref{thm:charUCQwithFin} and the correspondence between $k$-colorings and homomorphisms into $K_k$. 
    There is only a single positive and a single negative example.
    We begin with the positive example $(\Amc_2,q_P)$. It is designed to enforce the choice of a 2-coloring of $G_{|P}$ via the homomorphism witnessing $\Amc_2 \cup \Omc \models q_p$, where \Omc is a fitting \Lmc-ontology. We achieve this by choosing 
    $\Amc_2$ to be (the directed variant of) $K_2$ and $q_P$ to be $G_{|P}$. Moreover, we force the query $q_P$ to map entirely into $\mn{ind}(\Amc_2)$ by adding to both $\Amc_2$ and $q_P$ reflexive loops
    for a role name $s$.  The chosen 2-coloring (i.e., the homomorphism) must then be ``memorized'' in \Omc via concept names $V_v$ tied to the respective variables $x_v$ in $q_P$ for each $v \in P$, together with dedicated concept assertions $T_i(b_i)$ in $\Amc_2$ corresponding to the two colors. Since $q_P$ must be rooted, we also add an anchor individual $b_1$ and connect it by $s$ to all variables.
    Formally, define $E^+ := \{(\Amc_2, q_S)\}$, where
    \begin{align*}
        \Amc_2 :=\;&  
    \{r(b_1, b_2), r(b_2, b_1), T_1(b_1), T_2(b_2)\}\\
    & \cup \{s(b_i,b_j) \mid 1\leq i,j \leq 2\} 
    \end{align*}
    and
    \begin{align*}
    q_P :=\;& \exists x_1\ldots \exists x_k\; \bigwedge_{1 \leq i \leq k} s(x_k, x_k) \land \bigwedge_{v \in P} V_v(x_v) \\ &\qquad\qquad\quad\land \bigwedge_{(u,v)\in E_{|P}} r(x_u, x_v) \land r(x_v , x_u) \\ 
    & \qquad\qquad\quad\land \bigwedge_{v \in P} s(b_1, x_v).
    \end{align*}
    Here, $E_{|P}$ denotes the restriction of $E$ to $P$.
     It remains to define a negative example $(\Amc_3, q_G)$ that checks that the 2-coloring chosen by the positive example and memorized by \Omc via the concepts $V_v$ is not extendable to a 3-coloring of $G$. Recall that a fitting ontology \Omc should satisfy $\Amc_3 \cup \Omc \not\models q_G$. 
     We reuse the same trick as before: encoding $K_3$ in $\Amc_3$ and encoding $G$ in $q_G$ via the dedicated role name $r$. On its own, this would merely ensure that $G$ is not 3-colorable. To tie the construction to the 2-coloring chosen by the positive example, we attach, for each $v \in P$, the concept $V_v$ to the corresponding variable $x_v$ in $q_G$ and include the concept assertions $T_i(a_i)$ in $\Amc_3$. With this restriction, the negative example guarantees that there exists no 3-coloring of $G$ that agrees with the 2-coloring of $G_{|P}$ chosen by $(\Amc_2, q_P)$ and stored in \Omc. To put it differently, this 2-coloring of $G_{|P}$ is not extendable to a 3-coloring of $G$. To complete the construction, we add an individual and roles to make $q_G$ a rooted CQ.
     Formally, define $E^-:=\{(\Amc_3, q_G)\}$, where
     \begin{align*}
         \Amc_3 =\;& \{T_1(a_1), T_2(a_2)\}  \cup \{s(a_i,a_j) \mid 1 \leq i,j \leq 3\}\\ 
         &\cup \{ r(a_i, a_j), r(a_j, a_i) \mid 1 \leq i < j \leq 3\}
     \end{align*}
     and
     \begin{align*}
         q_G=\;& \exists x_1 \ldots \exists x_n \bigwedge_{1\leq i \leq n} s(a_1,x_i) \land \bigwedge_{v \in P} V_v(x_v) 
         \\ &\qquad\qquad\quad\land \bigwedge_{(v,w)\in E} r(x_v,x_w) \land r(x_w,x_v)
         .
     \end{align*}
     Note that the construction of $E_{G,P}$ is linear in the size of $G$. It remains to prove the validity of the reduction.
    \begin{lemma}
        Let $\Lmc\in\{\EL, \ELbot\}$, $G=(V,E)$ be an undirected graph and $P \subseteq V$. Then, $G_{|P}$ admits a 2-coloring that is not extendable by a 3-coloring of $G$ if and only if $E_{G,P}$ admits a fitting \Lmc-ontology.
    \end{lemma}
    \noindent
    \begin{proof}
     `\emph{if}'. Suppose there exists a fitting \Lmc-ontology for $E_{G,P}$. By Theorem \ref{thm:charUCQwithFin}, there exists a finite interpretation \Imc that satisfies Conditions~(a) and~(b). 
    Since $E^-$ is a singleton, we may assume that $\Imc = \Imc_e$, where $e=(\Amc_3, q_G ) \in E^-$ and $\Imc_e$ is the interpretation witnessing Condition~(a). This means that $\Imc$ is a model of $\Amc_3$ and $\Imc \not \models q_G$. We begin with the following claim:
    \\[1mm]
	{\bf Claim.} There exists a 2-coloring $c_2:P \to \{1,2\}$ of $G_{|P}$ such that $a_{c_2(v)} \in V_v^\Imc$ for all $v \in P$. \vspace{0.1cm}

    Observe that $S=\{(b_1, a_1), (b_2,a_2)\}$ is an $\Lmc_\bot$-simulation from $\Amc_2$ to $\Amc_3$. Thus, $S$ is also an \Lmc-simulation from $\Amc_2$ to \Imc. Since $\Imc$ satisfies Condition~(b), $\Imc_{\Amc_2, S, \Lmc}$ entails $q_P$. Let $h$ be the witnessing homomorphism. The reflexive $s$-roles in $q_P$ ensure that the image of $h$ lies entirely in $\mn{ind}(\Amc_2)$, since $\Imc_{\Amc_2, S, \Lmc}$ is a forest model.
    Define $c_2: P \to \{1,2\}$ by setting $c_2(v)=i$ whenever $h(x_v)=b_i$. To show that $c_2$ is a proper 2-coloring of $G_P$, let $(u,v) \in E_{|P}$ and suppose $c_2(u)=i$. Then $h(x_u)=b_i$. An inspection of $q_P$ and $\Imc_{\Amc_2, S, \Lmc}$ reveals that $r(x_u,x_v), r(x_v, x_u)\in q_P$ and $(b_i,b_i)\notin r^{\Imc_{\Amc_2, S, \Lmc}}$.  This implies $h(x_v)=b_j$ with $i\neq j$. Hence, $c_2(v) \neq c_2(u)$, proving that $c_2$ is a proper 2-coloring of $G_{|P}$. 
    It remains to verify the additional property of the claim. First, note that $h(x_v)=b_{c_2(v)}$ for all $v \in P$. Since $h$ is a homomorphism, we obtain $b_{c_2(v)}\in V_v^{\Imc_{\Amc_2, S, \Lmc}}$. By Corollary~\ref{lem:IASLsim} and the definition of $S$, we conclude $a_{c_2(v)}\in V_v^{\Imc}$. 

    \medskip
    We now show that the 2-coloring $c_2$ of the claim is not extendable to a 3-coloring of $G$. It suffices to prove the equivalent statement that every 3-coloring of $G$ disagrees with $c_2$ on some vertex in $ P$.
    We use the well known fact that every 3-coloring $c_3$ of an undirected graph $H$ gives rise to a homomorphism $h$ from $H$ to $K_3=(\{1,2,3\}, E')$ such that $c_3(v)=h(v)$. 
    Let $c_3: V \to \{1,2,3\}$ be a 3-coloring of $G$. By construction of $(\Amc_3, q_G)$, this 3-coloring induces a homomorphism $h$ from $q_G \setminus \{V_v(x_v)\mid v \in P\}$ to $\Amc_3$ such that $h(x_v)=a_{c_3(v)}$ for all $v \in V$.  Note that $h$ cannot be a homomorphism from $q_G$ to \Imc, as \Imc satisfies Condition~(a). Hence, there exists some $v \in P$ such that $h(x_v)= a_{c_3(v)} \not \in V_v^\Imc$. 
    The claim yields $a_{c_2(v)} \in V_v^\Imc$ which implies $c_2(v) \neq c_3(v)$.
    
    \medskip
     `\emph{only if}'. Let $c_2: S \rightarrow \{1,2\}$ be a 2-coloring of $G_{|P}$ that is not extendable to a 3-coloring of $G$. Based on $c_2$, define the finite interpretation \Imc as follows:
    \begin{align*}
        \Delta^\Imc &= \{a_1, a_2, a_3\},\\
        T_1^\Imc &= \{a_1\}, \\T_2^\Imc &= \{a_2\},\\
        V_v^\Imc&= \{a_{c_2(v)}\} && \text{ for all } v \in P,\\
        r^\Imc &= \{(a_i, a_j) \mid 1 \leq i,j \leq 3, i\neq j\},\\
        s^\Imc &= \{(a_i, a_j) \mid 1 \leq i,j \leq 3\}.
    \end{align*}
    All concepts and roles not defined are interpreted as empty. 
    
    It remains to show that \Imc satisfies Conditions~(a) and~(b) of Theorem~\ref{thm:charUCQwithFin}. 
    We start with (a). Since $E^-$ is a singleton, we may assume $\Imc_e = \Imc$ for $e=(\Amc_3, q_G)\in E^-$. It is immediate that $\Imc_e$ is a model of $\Amc_3$. Contrary to what we want to show, suppose $\Imc_e \models q_G$ with witnessing homomorphism $h$. 
        By definition of $\Amc_3$ and $q_G$, every homomorphism $h$ from $q_G$ to $\Imc_e$ gives rise to a 3-coloring $c_3: G \to \{1,2,3\}$, defined by $c_3(v)=i$ if $h(x_v)= a_i$. 
        The atoms $V_v$ in $q_G$ imply 
         $a_{c_3(v)}= h(x_v) \in V_v^\Imc$ for all $v \in P$. Hence, $c_3(v)=c_2(v)$ for all $v \in P$.  This contradicts the assumption that $c_2$ is not extendable to a 3-coloring of $G$. Thus $\Imc_e \not\models q_G$, meaning that $\Imc$ satisfies Condition~(a), as required.  
        Next, we consider Condition~(b), for which we only need to treat the single positive example $(\Amc_2, q_P)$. 
        Since $T_1^\Imc$ and $T_2^\Imc$ are singletons, the only non-empty $\Lmc$-simulation from $\Amc_2$ to $\Imc$ is $S=\{(b_1,a_1), (b_2,a_2)\}$. Define the mapping $h: \mn{var}(q_P) \cup \{b_1\} \to \mn{ind}(\Amc_2) $ by $h(b_1)=b_1$ and $h(x_v)=b_{c_2(v)}$ for all $v \in P$. It is straightforward to verify that $h$ is a homomorphisms from $q_P \setminus \{V_v(x_v) \mid v \in P\}$ to $\Amc_2$. Moreover, by construction $a_{c_2(v)} \in V_v^\Imc$, from which we obtain  $\Imc_{\Amc_2, S, \Lmc} \models V_v(b_{c_2(v)})$, for all $v \in P$. Hence, $h$ is a homomorphism from $q_P$ to $\Imc_{\Amc_2, S, \Lmc} $. By definition of $h$, this proves $\Imc_{\Amc_2, S, \Lmc} \models q_P$. If $\Lmc = \EL$, we additionally need to consider the empty \Lmc-simulation $\emptyset$ from $\Amc_2$ to $\Imc$. By definition of $\Imc_{\Amc_2, \emptyset, \Lmc}$, we have $\Imc_{\Amc_2, \emptyset, \Lmc} \models V_v(b_i)$ for all $v \in P$ and $1 \leq i \leq 2$. Clearly, this implies that the mapping $h$ defined above is again a homomorphism from $q_P$ to $\Imc_{\Amc_2, \emptyset, \Lmc} $ with $h(b_1) = b_1$, witnessing $\Imc_{\Amc_2, \emptyset, \Lmc} \models q_P$, as required.
    \end{proof}
\section{Proofs for Section~\ref{sect:ELI}}
\subsection*{Upper bound for $\ELI_{(\bot)}$}
\finmodontologyfinmodkb*
\noindent
\begin{proof}
    We present the proof for the restricted case. The proof of the unrestricted case can be obtained by removing every occurrence of ``\mn{fin}'', ``finite'' and ``finitely''. (Replacing ``a'' by ``an'' at the appropriate positions is left to the reader)
    
    We first prove Point~(1).
    `\emph{only if}' 
    We argue by contraposition. Assume that $\Amc \cup \Omc \not \models^{\mn{fin}} C(a)$. Then there exists a finite model $\Imc$ of $\Amc \cup \Omc $ such that $a \not \in C^\Imc $. Let \Jmc be the finite interpretation obtained from \Imc by interpreting the fresh symbols as:
    \begin{align*}
        u^\Jmc &:= \Delta^{\Imc} \times \mn{ind}(\Amc)\\
        S_{\Amc, b}^\Jmc &:= (\exists^{sim}_\Lmc ({\Amc}, b))^\Imc \qquad \text{ for all } b \in \mn{ind}(\Amc)
    \end{align*}
    Because $u$ and the concepts $S_{\Amc, b}$ do not occur in \Omc, $\Jmc$ is a model of $\Omc$.
    Moreover, \Jmc is also a model of $\Omc_\Amc$, since by construction $b \in S_{\Amc, b}^\Jmc$ for all $b \in \mn{ind}(\Amc)$. In particular $a \in S_{\Amc, a}^\Jmc$ and $a \not \in C^\Jmc$. Hence, $\Jmc \not \models S_{\Amc, a} \sqsubseteq C$ from which we conclude $\Omc \cup \Omc_\Amc \not\models^{\mn{fin}} S_{\Amc, a} \sqsubseteq C$. 
    
    \medskip
    `\emph{if}' 
     We again argue by contraposition. Assume that $\Omc \not \models^{\mn{fin}} S_{\Amc, a} \sqsubseteq C $.
    Then there exists a finite witness model $\Imc $ of $\Omc$ such that $S_{\Amc, a}^\Imc \not \subseteq C^\Imc$. 
    
    We proceed as follows. In the spirit of Definition \ref{def:tripleint}, we construct a finite model $\Jmc_\Amc$ of \Amc by attaching to each individual $b$ of \Amc a copy of a uniform product of \Imc with itself. 
    The element attached to $b$ is chosen such that it enumerates exactly the elements of $S_{\Amc,b}^\Imc$.
    In this way, we can view $b$ as a representative of $S_{\Amc,b}^\Imc$.
    We then construct a set of simulations between these uniform copies that extend to a simulation from $\Jmc_\Amc$ into each copy. In the last part, we use these simulations to show that $\Jmc_\Amc$ is a model of $\Amc \cup \Omc$, and then use the fact that $a$ is a representative of $S_{\Amc,a}$ and the assumption to prove the lemma.

    Let $b_m\in \mn{ind}(\Amc)$ be an individual such that for each $c \in \mn{ind}(\Amc)$, there exists a surjection from  $S_{\Amc ,b_m}^\Imc$ to $S_{\Amc ,c}^\Imc$, i.e., we have 
    \[b_m  \in  \mathrm{argmax}_{c \in \mn{ind}(\Amc)} \;|S_{\Amc,c}^\Imc|.\]
    Then define $(\Jmc, \bar w)$ as
    \[
    (\Jmc, \bar w):=\prod_{d \in S_{\Amc ,b_m}^\Imc} (\Imc, d).
    \]
    Based on \Jmc, we obtain for each $c\in \mn{ind}(\Amc)$ the interpretation $\Jmc_c$ from \Jmc  
    by renaming an element $(d_i)_{i \in I} \in \Delta^\Jmc$ that satisfies $\{d_i \}_{i \in I} = S_{\Amc ,c}^\Imc$ to $c$. The existence of such an element follows from the choice of $b_m$. 
    Observe that all interpretations $\Jmc_c$ are isomorphic (ignoring individuals), as they arise from the same product. 
    For the same reason, we can make the following crucial observation:
    \\[1mm]
	{\bf Claim.} Let $c \in \mn{ind}(\Amc)$. Then there exists a left-total \Lmc-simulation $R_c$ from $\Jmc_c$ to $\Jmc$ such that $cR_c = S_{\Amc,c}^\Jmc$.
    \\[2mm]
    \textbf{Proof of claim:}
    By construction of $\Jmc_c$, it suffices to construct a left-total \Lmc-simulation $R$ from $\Jmc$ to itself such that for each $c \in \mn{ind}(\Amc)$ and $\bar v_c =(d_i)_{i \in I}\in \Delta^\Jmc$, where $\{d_i\}_{i \in I} = S_{\Amc ,c}^\Imc$, we have  \[\bar v_c R = S^\Jmc_{\Amc ,c} .\]
    Indeed, after renaming $\bar v_c$ to $c$, this relation yields the required simulation $R_c$ from $\Jmc_c$ to $\Jmc$.
    
    We construct $R$ from a union of homomorphisms. Intuitively the uniform product structure of \Jmc allows us to construct, for each tuple, homomorphisms that arbitrarily re-index its components.
    By construction of \Jmc and Point~1 of Lemma~\ref{lem:prodlem}, for each $\bar v=(d_i)_{i \in I} \in \Delta^{\Jmc}$ and every $i \in I$ there exists a homomorphism $h_{\bar v, i}$ from \Jmc to \Imc such that $h_{\bar v, i}(\bar v)=d_i$. Using Point~3 of Lemma~\ref{lem:prodlem} we combine these homomorphisms pointwise in every possible way: for every $\bar v=(d_i)_{i \in I}\in \Delta^\Jmc$ and each function $f: I \to I$ there exists a homomorphism $g_{\bar v, f}$ from $\Jmc$ to $\Jmc$ such that $g_{\bar v, f}(\bar v) =(h_{\bar v, f(i)}(\bar v))_{i \in I}= (d_{f(i)})_{i \in I}$. Now, define \[
    R= \bigcup_{\substack{\bar v \in \Delta^\Jmc, \\f:I \to I}}\left\{(\bar d, \bar e) {\;\bigl \vert\;}   g_{\bar v, f}(\bar d)= \bar e \right\}.
    \]
    Since $R$ is a non-empty union of homomorphisms (viewed as binary relations), it is a left-total \Lmc-simulation.
    It remains to show that $R$ satisfies the required equality. Fix $c \in \mn{ind}(\Amc)$ and $\bar v_c = (d_i)_{i \in I} \in \Delta^\Jmc$ with $\{d_i \}_{i \in I} = S_{\Amc ,c}^\Imc$. Since by construction $\bar v_c \in S_{\Amc, c}^\Jmc$ and $R$ is an \Lmc-simulation, we already have $\bar v_c R \subseteq S_{\Amc, c}^\Jmc$. For the opposite inclusion, let $\bar v = (e_i)_{i\in I} \in S_{\Amc, c}^\Jmc$. By construction of $\Jmc$ as a product, we obtain $\{e_i\}_{i\in I} \subseteq S_{\Amc, c}^\Imc =\{d_i\}_{i \in I} $. Hence, there exists a function $f: I \to I$ such that $e_i = d_{f(i)}$ for all $i \in I$. Then  $g_{\bar v_c, f}(\bar v_c) = \bar v$ and therefore $(\bar v_c, \bar v)\in R$. 
    This yields $\bar v_c R \supseteq S_{\Amc, c}^\Jmc$, as required.
    
    \medskip

    We now construct the finite interpretation $\Jmc_\Amc$ that witnesses $\Amc \cup \Omc \not \models^{\mn{fin}} C(a)$ from the interpretations $\Jmc_c$ defined above. Throughout the remainder of the proof, we assume, w.l.o.g., that $\Delta^{\Jmc_b} \cap \Delta^{\Jmc_c} = \emptyset$ for all $b,c \in \mn{ind}(\Amc)$ with $b \neq c$.
    Then, let $\Jmc_\Amc $ be the interpretation obtained from 
    $\Jmc'=\bigcup_{c \in \mn{ind}(\Amc)} \Jmc_c$ by adding, for every $r(c_1,c_2) \in \Amc$, the pair $(c_1,c_2)$ to $r^{\Jmc'}$.
    Since $\Imc$ is a model of $\Omc_\Amc$, $A(c) \in \Amc$ implies $c \in S_{\Amc, c}^\Imc \subseteq A^\Imc$.
    Hence,  $\Jmc_\Amc$ is a finite model of $\Amc$.
    It remains to show that $\Jmc_\Amc$ is a model of $\Omc$.
    First observe that $\Jmc$ and every $\Jmc_c$ are models of $\Omc \cup \Omc_\Amc$, because $\Lmc^{\mn{sim}}_\bot$ is invariant under direct products, by Corollary~\ref{cor:LsiInvDirProd}.
    To establish that $\Jmc_\Amc \models \Omc$, we relate the two interpretations via an \Lmc-simulation from $\Jmc_\Amc$ to \Jmc.
    Let \[R = \bigcup_{c \in \mn{ind}(\Amc)} R_c,\]
    where $R_c$ is the left-total \Lmc-simulation from $\Jmc_c$ to \Jmc obtained from the previous claim.
    \\[1mm]
	{\bf Claim.} $R$ is a left-total \Lmc-simulation from $\Jmc_\Amc$ to $\Jmc$.
    \\[2mm]
    \textbf{Proof of claim:} 
    That $R$ satisfies (Atom) and is left-total follows immediately from the construction of $\Jmc_\Amc$ and from the fact that each $R_c$ is a left-total \Lmc-simulation from $\Jmc_c$ to $\Jmc$. It remains to verify (Rel) and (iRel). Let $r$ be a role, $(d_1, e_1 )\in r^{\Jmc_\Amc}$ and  $( d_1, d_2) \in R$. 
    We distinguish two cases:
    \begin{itemize}
        \item If $\{ d_1,  e_1\} \subseteq \Delta^{\Jmc_c}$ for some $c \in \mn{ind}(\Amc)$, then $(d_1, e_1)\in r^{\Jmc_c}$ and
    $( d_1,  d_2) \in R_c$. Since $R_c$ is a left-total \Lmc-simulation from $\Jmc_c$ to $\Jmc$, there is some $(d_2, e_2) \in r^\Jmc$ such that $(e_1, e_2) \in R_c \subseteq R$. 
        \item Otherwise $\{ d_1,  e_1\} \not\subseteq \Delta^{\Jmc_c}$ for all $c \in \mn{ind}(\Amc)$. By construction of $\Jmc_\Amc$, $d_1 = b, e_1 = c\in \mn{ind}(\Amc)$, and either  $r(b,c) \in \Amc$ or $r^-(c,b) \in \Amc$. Since $d_1 = b$, we obtain $(d_1, d_2)\in R_b$. This implies $ d_2 \in S_{\Amc, b}^\Jmc$ by the first claim. Because $\Jmc$ is a model of $\Omc_{\Amc}$, there exists an $ e_2 \in \Delta^\Jmc$ such that $(d_2,  e_2) \in r^\Jmc$ and $ e_2 \in S_{\Amc, c}^\Jmc$. Using the first claim a second time yields $( e_1, e_2 ) \in R_c \subseteq R$. 
    \end{itemize}
    This proves that $R$ is a left-total \Lmc-simulation, and therefore the claim.

    \medskip
    Recall that $\Jmc_c$ is isomorphic to \Jmc (ignoring individuals). Under this isomorphism, for every $c \in \mn{ind}(\Amc)$, we write $Z_c$ for the relation $R$ viewed as a left-total \Lmc-simulation from $\Jmc_\Amc$ to $\Jmc_c$. By construction of $R$
    and since $\Jmc_c \subseteq \Jmc_\Amc$, the relation $Z_c$ contains all endomorphisms of $\Jmc_c$ (viewed as binary relations). In particular, we have $\mathrm{id}_{\Delta^{\Jmc_c}} \subseteq Z_c$.
    
    We now use these simulations to show that $\Jmc_\Amc$ is a model of $ \Omc$. Let $D \sqsubseteq E \in \Omc$ and assume $d \in D^{\Jmc_\Amc}$. By design of $\Jmc_\Amc$, there is some $c \in \mn{ind}(\Amc)$ such that $d \in \Delta^{\Jmc_c}$.
    By Lemma~\ref{lem:LsiInvSim}, $\Lmc^{\mn{sim}}$-concepts are invariant under simulations, and thus $dZ_c \subseteq D^{\Jmc_c} \subseteq E^{\Jmc_c}$. 
    Since $( d, d) \in Z_c$, this implies $d \in E^{\Jmc_c}$. 
    By construction of $\Jmc_\Amc$ we have $\Jmc_c \subseteq \Jmc_\Amc$.
    Thus, the identity relation is an \Lmc-simulation from $\Jmc_c$ to $\Jmc_\Amc$. A second application of Lemma~\ref{lem:LsiInvSim} yields $d \in E^{\Jmc_\Amc}$.  
    Therefore, $D^{\Jmc_\Amc} \subseteq E^{\Jmc_\Amc}$, and hence $\Jmc_\Amc \models \Omc$. Since we have already established that $\Jmc_\Amc \models \Amc$, we conclude that $\Jmc_\Amc $ is a finite model of $\Amc \cup \Omc$.
    Recall that the initial assumption guarantees the existence of some $d \in S_{\Amc, a}^\Imc$ such that $d \not \in C^\Imc $. Since $\Jmc$ is a product of $\Imc$ with itself, Lemma~\ref{lem:prodlem} yields an element $\bar d \in S_{\Amc, a}^\Jmc $ with $\bar d \not \in C^\Jmc$, witnessing $S_{\Amc, a}^\Jmc \not \subseteq C^\Jmc$.
    Moreover, by construction of the \Lmc-simulation $R$ from $\Jmc_\Amc$ to $\Jmc$, we have $a R = S_{\Amc, a}^\Jmc$.
    Hence, $a R \not \subseteq C^\Jmc$ and thus $a \not \in C^{\Jmc_\Amc}$, by contraposition of Lemma~\ref{lem:SimPresConc}. 
    Consequently, $\Jmc_\Amc$ witnesses $\Amc \cup \Omc \not \models^{\mn{fin}} C(a)$.
    \bigskip

    To prove Point~(2), let $a \in \mn{ind}(\Amc)$ and set $\Omc' := \Omc \cup \{C \sqsubseteq \bot\}$. An application of Point~(1) yields \[ \Omc' \cup \Omc_\Amc \models^{\mn{fin}} S_{\Amc,a} \sqsubseteq \bot \;\;\text{if and only if}\;\; \Amc \cup\Omc' \models^{\mn{fin}} \bot(a).\]
    Moreover, $\Amc \cup \Omc' \models^{\mn{fin}} \bot(a)$ holds precisely when
    $\Amc \cup \Omc \cup \{C \sqsubseteq \bot\}$ is finitely inconsistent,
    which in turn is equivalent to
    $\Amc \cup \Omc \models^{\mn{fin}} \exists x\,C(x)$.
    This proves Point~(2).
\end{proof}

\subsection*{Deciding $\ELI^{\mn{sim}}_\bot$ subsumption}

We aim to prove the following.

\thmELIsimsubsexptime*
To decide $\ELI_\bot$ subsumption with respect to $\ELI^{\mn{sim}}_\bot$-ontologies, we translate the problem into the two-way multi-modal $\mu$-calculus with simultaneous fixed points, denoted by $L_\mu$. Let $\mn{VAR}$ be a countably infinite set of variables. A formula of the two-way $\mu$-calculus with simultaneous fixed points (\emph{$L_\mu$-formula}) is formed according to the syntax rule
\[
\varphi ::= X \mid A \mid \neg \varphi \mid \varphi \land \psi \mid \langle r\rangle\varphi
 \mid \nu_k X_1 \dots X_n.\, \varphi_1,\dots,\varphi_n,
\]
where $X, X_1,\dots,X_n \in \mn{VAR}$, $A \in \NC \cup\{\top, \bot\}$, $r$ is a (possibly inverse) role, $1 \leq k \leq n$, and $\varphi,\psi,\varphi_1,\dots,\varphi_n$ are $L_\mu$-formulas. In the last clause, every free occurrence of an $X_i$ in a $\varphi_j$ must be positive, i.e., within the scope of an even number of negations. Here, we call a variable \emph{free} if it is not in the scope of a $\nu_k$ operator that binds it, and  \emph{bound} otherwise. We can evaluate a $L_\mu$-formula with free variables in an interpretation by treating these variables as concept names. We say that an $L_\mu$-formula is an \emph{$L_\mu$-sentence} if it does not contain any free variables.  The derived operators $\varphi \lor \psi$, $[r]\varphi$ and $\varphi \rightarrow \psi$ are defined as usual by $\neg (\neg \varphi \land \neg \psi)$, $\neg \langle r\rangle \neg \varphi$ and $\neg \varphi \lor \psi$, respectively. In addition, for a finite set $R$ of role names, we define
$[U]_R \varphi$ as an abbreviation for
$$
  \nu_1 X . (\varphi \wedge \bigwedge_{r \in R} ([r]X \wedge [r^-]X)),
$$
expressing that $\varphi$ is true at every element reachable from the current element on a path that travels along roles from $R$, forwards or backwards.

The semantics of $L_\mu$-formulas is defined via interpretations and variable assignments in the usual way. For a concrete definition, we refer to~\cite{DBLP:books/el/07/BradfieldS07}.
In particular, given an interpretation \Imc, the interpretation of an $L_\mu$-formula $\nu_k X_1 \dots X_n.\, \varphi_1,\dots,\varphi_n$ is the $k$-th component $W_k$ of the unique greatest fixed point $(W_1, \dots, W_n)$, with respect to component-wise set inclusion, of the equation system 
%
$$
\begin{array}{rcl}
W_1 &=& \varphi_1^{\Imc[X_1 \to W_1, \dots, X_n \to W_n]},\\[1mm]
&\dots&\\[1mm]
W_n &=& \varphi_n^{\Imc[X_1 \to W_1, \dots, X_n \to W_n]}.
\end{array}
$$
%

We use the following result.
\begin{theorem}
\label{thm:mucalcinexptime}
    Finite satisfiability in the two-way $\mu$-calculus with simultaneous fixed points is \ExpTime-complete.
\end{theorem}
We are not aware that Theorem~\ref{thm:mucalcinexptime}
can be found in exactly this form in the literature.
However, it is folklore that the standard translation of an $\Lmc_\mu$-formula without simultaneous fixed points into a two-way parity  automaton given in
   \cite{vardi1998reasoning} 
   and also used in \cite{bojarnczyk2002two} can be generalized to simultaneous fixed points while yielding an automaton with a polynomial-sized set of states and  acceptance condition.
   It remains to decide emptiness of the
   resulting automaton, which is possible in \ExpTime both in the unrestricted \cite{vardi1998reasoning} and in the finite case \cite{bojarnczyk2002two}.
   
It only remains to verify that, given an $\ELI^{\mn{sim}}_\bot$ ontology \Omc and two $\ELIbot$ concepts $E, F$,  we can construct an $L_\mu$ sentence $\varphi$ of size polynomial in $||\Omc|| + ||E|| +||F|| $ such that $\Omc \not \models^{\mn{fin}} E \sqsubseteq F$ if and only if $\varphi$ is finitely satisfiable.

The construction of $\varphi$  relies on a mapping $\pi$ that sends $\ELI^{\mn{sim}}$-concepts to $L_\mu$-sentences. This translation is based on the standard translation $\pi$ of \ELIbot-concepts to $L_\mu$-sentences, inductively defined as follows: 
\begin{align*}
    \pi(\bot) &:= \bot, \\
    \pi(\top) &:= \top, \\
    \pi(A) &:= A \qquad\qquad\qquad A\in \NC,\\
    \pi(C \sqcap D) &:= \pi(C) \wedge \pi(D),\\
    \pi(\exists r. C) &:= \langle r\rangle \pi(C).
\end{align*}
Clearly, $||\pi(C)||$ is polynomial in the size of $C$. Moreover, it is well known that the following holds \cite{Baader2017}, where $R(\Omc)$ denotes the set of all role names used in the ontology \Omc.
\begin{proposition}\label{prop:modal}
    Let $\Omc$ be an \ELIbot ontology and $E$, $F$ \ELIbot-concepts. Then, $\Omc \not\models^{\mn{fin}} E \sqsubseteq F$ if and only if \[[U]_{R(\Omc)}\big ( \bigwedge_{C\sqsubseteq D \in \Omc} (\pi(C) \rightarrow \pi(D)) \big )\wedge \pi(E) \wedge \neg \pi(F)\]
    is finitely satisfiable.
\end{proposition}
Although, in its original form, the proposition is stated with the genuine universal modality $[U]$, we may employ the same connected-model argument as used by \cite{Schild91}, to replace $[U]$ with $[U]_R(\Omc)$.

It remains to show that $\pi$ and Proposition~\ref{prop:modal} can be extended to the simulation quantifier. We set  \[\pi(\exists^{\mn{sim}}_\ELI (\Imc, d_k)) : = \nu_k X_1 \dots X_n. \varphi_1, \dots \varphi_n,\]
where $\Delta^{\Imc}=\{d_1, \dots, d_n\}$, and for  $1 \leq i \leq n$,
\[
\varphi_i = \bigwedge_{d_i \in A^\Imc} A \wedge \bigwedge_{(d_i,d_j)\in r^\Imc} \langle r\rangle X_j.\]
Note that $||\pi(\exists^{\mn{sim}}_\ELI (\Imc, d_k))||$ is polynomial in  $||\Imc||$. 

\begin{lemma}
    Let $(\Jmc, e)$ be a pointed interpretation. Then, \[e \in (\exists^{\mn{sim}}_\ELI (\Imc, d_k))^\Jmc \quad\text{ iff }\quad e\in (\pi(\exists^{\mn{sim}}_\ELI (\Imc, d_k)))^\Jmc.\]
\end{lemma}
\noindent
\begin{proof}
    `\emph{if}' Assume that $e\in (\pi(\exists^{\mn{sim}}_\ELI (\Imc, d_k)))^\Jmc$ and
    let $(W_1, \dots, W_n)$ be the greatest solution of the simultaneous fixed point system defined by $\pi(\exists^{\mn{sim}}_\ELI (\Imc, d_k))$. This means that, for $1 \leq i \leq n$, we have 
    \[W_i = \varphi_i^{\Jmc[X_1 \to W_1, \dots, X_n\to W_n]} \quad \text{ and } \quad e \in W_k.\]
    Now, define the relation $S \subseteq \Delta^\Imc \times \Delta^\Jmc$ as \[
    S:= \{(d_i, d) \mid d \in W_i \}.
    \]
    Since by definition $(d_k,e)\in S$, it remains to prove that $S$ is an \ELI-simulation. That $S$ satisfies (Atom) is clear by inspecting the definition of $\varphi_i$. To prove (Rel) and (iRel), let $(d_i, e_1)\in S$ and $(d_i, d_j) \in r^\Imc$ for some role $r$. Then by definition of $S$, we have $e_1 \in W_i = \varphi_i^{\Jmc[X_1 \to W_1, \dots, X_n\to W_n]}$. The construction of $\varphi_i$ and semantics of $\langle r\rangle$ yields an element $e_2 \in W_j$ with $(e_1, e_2) \in r^\Jmc$. Moreover, by definition of $S$, $e_2 \in W_j$ implies $(d_j, e_2)\in S$, as required.

    \medskip
    `\emph{only if}' Let $e \in (\exists^{\mn{sim}}_\ELI (\Imc, d_k))^\Jmc$. Then by the semantics of simulation quantifiers, there exists an \ELI-simulation $S$ from $\Imc$ to $\Jmc$ such that $(d_k, e)\in S$. Define the sets $W_i^S := d_i S$ for all $1 \leq i \leq n$. Then for $1 \leq i \leq n$, we have $W_i^S \subseteq \varphi_i^{\Jmc [X_1 \to W_1^S, \dots, X_n \to X_n^S]}$. This can be easily verified by the property (Atom) corresponding to the first conjunction, and (Rel) and (iRel) corresponding to the second conjunction. Thus, $W_k^S$ is a post fixed point of $\pi(\exists^{\mn{sim}}_\ELI (\Imc, d_k))$ in \Jmc. Since $\pi(\exists^{\mn{sim}}_\ELI (\Imc, d_k))$ is a greatest fixed point formula, we conclude $e \in W_k^S \subseteq  (\pi(\exists^{\mn{sim}}_\ELI (\Imc, d_k)))^\Jmc $.
\end{proof}

Proposition~\ref{prop:modal} can thus be extended in the following way.
\begin{proposition}
\label{prop:ELIsimtomucalc}
    Let $\Omc$ be an $\ELIbot^{\mn{sim}}$ ontology and $E$, $F$ \ELIbot-concepts. Then, $\Omc \not\models^{\mn{fin}} E \sqsubseteq F$ if and only if \[[U]_{R(\Omc)}\big ( \bigwedge_{C\sqsubseteq D \in \Omc} (\pi(C) \rightarrow \pi(D)) \big )\wedge \pi(E) \wedge \neg \pi(F)\]
    is finitely satisfiable.
\end{proposition}
Clearly the  $L_\mu$ sentence in Proposition~\ref{prop:ELIsimtomucalc} is of size polynomial in $||\Omc|| + ||E|| + ||F||$.

 \subsection*{Lower bound for \ELI.}
\ELIfitExpT*
As a proof, we provide a \PTime-reduction to the complement of AQ entailment in \ELI, which is known to be \ExpTime-hard \cite{BaaderEtAl-OWLED08DC}.

Let \Omc be an \ELI-ontology, \Amc an ABox, and $A_q(a)$ be an AQ. 
We construct a collection of labeled ABox-CQ examples $E=(E^-,E^+)$ such that $\Amc \cup \Omc \not \models A_q(a)$ if and only if there exists a fitting \ELI-ontology for $E$ if and only if there exists a fitting \ELIbot-ontology for $E$.

To keep the reduction simple, we may assume w.l.o.g. that $\Omc$ is in \emph{normal form}. That is, every CI in \Omc is of one of the following patterns: \[
A \sqsubseteq B, \quad A_1 \sqcap A_2 \sqsubseteq B,\quad
    A \sqsubseteq \exists r. B, \quad\exists r.  A \sqsubseteq B,
\]
where $A, A_1, A_2, B$ are concept names or $\top$.
Indeed, if \Omc is not in normal form, one can construct in polynomial time an ontology $\Omc'$ such that $\Amc \cup \Omc' \models A_q(a)$ if and only if $\Amc \cup \Omc \models A_q(a)$ \cite{Baader2017}.

We now describe the reduction, starting by translating the CIs of \Omc into positive examples and then defining $\Amc$ and $A_q(a)$ to be the single negative example. For the reader's convenience, it may be helpful to revisit the characterization from Theorem~\ref{thm:charUCQwithFin}.

First, we introduce a fresh concept name $A_{\top}$ that allows us to explicitly refer to the $\top$-concept in the ABoxes and queries we construct below. A suitable translation function $\cdot^\mn{cv}: \NC \cup \{\top\} \to \NC$ is then defined by extending the identity on \NC with $\top^{\mn{cv}} := A_\top$.
Using this translation, we encode the CIs of \Omc as positive examples by adding the following elements to  $E^+$:
\begin{itemize}
    \item $(\{A^{\mn{cv}}(a)\}, B^{\mn{cv}}(a)) $ for every $ A \sqsubseteq B \in \Omc$;
    \item $(\{A_1^{\mn{cv}}(a), A_2^{\mn{cv}}(a)\}, B^{\mn{cv}}(a))$ for every $ A_1 \sqcap A_2 \sqsubseteq B \in \Omc$;
    \item $(\{A^{\mn{cv}}(a)\}, \exists x\ r(a,x)\land B^{\mn{cv}}(x) \land A_\top(x))$ for every  $A \sqsubseteq \exists r. B \in \Omc$;
    \item $(\{A^{\mn{cv}}(b), r(a,b)\}, B^{\mn{cv}}(a))$ for every $ \exists r. A \sqsubseteq B \in \Omc$.
\end{itemize}

The set $E^-$ consists of the single negative example $(\Amc \cup \{A_\top(a) \mid a \in \mn{ind}(\Amc)\}, A_q(a))$. 

It remains to show the correctness of the reduction.
 \begin{lemma}
    Let $\Lmc \in \{\ELI, \ELIbot\}$, \Amc an ABox, \Omc an \ELI-ontology, $A_q(a)$ an AQ, and $E=(E^+,E^-)$ the labeled collection of examples constructed above. Then the following are equivalent:
    \begin{enumerate}[label=(\roman*)]
        \item $\Amc \cup \Omc \not \models A_q(a)$;
        \item there exists a fitting $\ELI$-ontology for $E$;
        \item there exists a fitting $\ELIbot$-ontology for $E$.
    \end{enumerate}
 \end{lemma}
 \noindent
 \begin{proof}
     Let $\Lmc$, \Omc, \Amc, $A_q(a)$, and $E$ be as in the Lemma.
     
     `(i)$~\Rightarrow~$(ii)' Assume that $\Amc \cup \Omc \not \models A_q(a)$. Since \ELI enjoys the finite model property \cite{Baader2017}, $\Amc \cup \Omc \not \models^{\mn{fin}} A_q(a)$. Let $\Imc$ be the corresponding finite witness model. W.l.o.g. we may assume that  $A_\top^\Imc = \Delta^\Imc$, because $A_{\top}$ is a fresh concept name. Then, $(A^{\mn{cv}})^\Imc = A^\Imc$ for all $A \in \NC \cup \{\top\}$.
     We now prove Condition~(ii) by showing that \Imc satisfies Conditions~(a) and~(b) of Theorem~\ref{thm:charUCQwithFin}. 
     Condition~(a) is already satisfied by the assumption and the fact that $\mn{ind}(\Amc)\subseteq A_\top^
     \Imc$, ensuring that $\Imc$ is a model of $\Amc \cup \{A_\top(a) \mid a \in \mn{ind}(\Amc)\}$.
     For Condition~(b), observe that every positive example corresponds to a CI in \Omc that is satisfied by \Imc. 
     Since the different types of examples can be handled similarly, we illustrate the argument only for examples that arise from CIs of the form $A \sqsubseteq \exists r. B$. Let $A \sqsubseteq \exists r. B \in \Omc$ and $(\Amc',q')$ be the corresponding positive example.
     By definition, $\Amc':= \{A^{\mn{cv}}(a)\}$ and $q':=\exists x \ r(a,x) \land B^{\mn{cv}}(x) \land A_\top(x)$.
     Suppose that $S$ is an $\ELI$-simulation from $\Amc'$ to $\Imc$. Since $S$ is an \ELI-simulation and $\Imc$ is a model of \Omc, we have 
     \[
        aS \subseteq (A^{\mn{cv}})^\Imc = A^\Imc \subseteq (\exists r. B)^\Imc=(\exists r. B^{\mn{cv}} \sqcap A_\top)^\Imc.
     \]
     By Corollary~\ref{lem:IASLsim}, this implies $\Imc_{\Amc', S , \ELI} \models \exists r. B^{\mn{cv}} \sqcap A_\top(a)$ and thus $\Imc_{\Amc', S , \ELI} \models q'$. Consequently, \Imc also satisfies Condition~(b) and is therefore a witness for the existence of a fitting \ELI-ontology.

     \medskip 
     `(ii)$~\Rightarrow~$(iii)' Trivial, since every \ELI-ontology is also an \ELIbot-ontology.
     
     \medskip
     `(iii)$~\Rightarrow~$(i)' Suppose there exists a fitting $\ELIbot$-ontology for $E$. By Theorem~\ref{thm:charUCQwithFin}, there exists an interpretation \Imc that satisfies Conditions~(a) and~(b) of the theorem. Using Condition~(a) and the fact that there is only one negative example, we may assume that \Imc is a model of $\Amc \cup \{A_\top(a) \mid a \in \mn{ind}(\Amc)\}$ such that $\Imc \not \models A_q(a)$. Let \Jmc be the restriction of \Imc to $A_\top^\Imc$. Then, $(A^{\mn{cv}})^\Jmc = A^\Jmc$ for all $A \in \NC \cup \{\top\}$.
     It is straightforward to verify that \Jmc is a model of $\Amc$ and that $\Jmc \not \models  A_q(a)$. It remains to show that \Jmc is a model of \Omc.
     Again, we exemplarily treat only one type of CI, the remaining cases are similar.  Let $\exists r. A \sqsubseteq B \in \Omc$ and let $(\Amc',q')$ be the corresponding example.
      By definition, $\Amc'= \{A^{\mn{cv}}(b), r(a,b)\}$ and $q' = B^{\mn{cv}}(a)$. Assume that $d \in (\exists r. A )^\Jmc$.
     Since $\Jmc \subseteq \Imc$, we have $d \in (\exists r. A )^\Imc$. Thus, there exists an $\ELIbot$-simulation $S$ from $\Amc'$ to \Imc such that $(a,d) \in S$. Because \Imc satisfies Condition~(b) of Theorem~\ref{thm:charUCQwithFin}, we obtain $\Imc_{\Amc', S, \ELIbot} \models B^{\mn{cv}}(a)$, and hence $d \in (B^{\mn{cv}})^\Imc$ by Corollary~\ref{lem:IASLsim}. By construction of $\Jmc$ this implies  
     $d \in (B^{\mn{cv}})^\Jmc = B^\Jmc$. Thus, \Jmc satisfies the concept inclusion $\exists r. A \sqsubseteq B \in \Omc$.
 \end{proof}

\cleardoublepage

\end{document}